\documentclass[11pt, titlepage]{article}
\usepackage[sort&compress, square, comma, authoryear]{natbib}
\usepackage{graphicx, lscape}
\usepackage{graphics, color}
\usepackage{epsfig}
\usepackage{epstopdf}
\usepackage{wrapfig}
\usepackage{caption}
\usepackage{subcaption}
\usepackage{amsmath, amsthm, amssymb, amscd}
\usepackage{latexsym}
\usepackage{multirow} 
\usepackage{array}
\usepackage{color}
\usepackage{rotating}
\usepackage[hyphens]{url}
\usepackage[hidelinks]{hyperref}  
\usepackage{cleveref}
\usepackage{fullpage}
\usepackage{fancyvrb}
\usepackage{pdfpages}
\usepackage{fmtcount}
\usepackage[margin=3.3cm]{geometry}
\usepackage{verbatim}
\usepackage[english]{babel}
\usepackage[utf8]{inputenc}
\usepackage{tikz} 
\usetikzlibrary{arrows, decorations.pathmorphing, backgrounds, fit, positioning, shapes.symbols, chains}
\usetikzlibrary{decorations.markings}
\usepackage{lscape}
\usepackage{algpseudocode}
\usepackage{algorithm}
\usepackage{kotex}
\usepackage{soul}
\usepackage{booktabs}
\usepackage{enumerate}
\usepackage{colortbl}
\long\def\comment#1{} 
\usepackage{booktabs}
\usepackage{yhmath}  

\usepackage{xcolor}
\definecolor{jm}{rgb}{0.5, 0.1, 0.3}
\definecolor{jh}{rgb}{0.2, 0.1, 0.6}
\definecolor{dark_green}{rgb}{0.1, 0.5, 0.1}

\newtheorem{definition}{Definition}

\newtheorem{proposition}{Proposition}

\newtheorem{theorem}{Theorem}
\newtheorem{lemma}{Lemma}
\newcommand\norm[1]{\left\lVert#1\right\rVert}  

\DeclareMathOperator*{\argmin}{arg\,min}

\usepackage{mathtools}

\newcommand{\bdm}{\begin{displaymath}}
\newcommand{\edm}{\end{displaymath}}

\newcommand\blfootnote[1]{%
  \begingroup
  \renewcommand\thefootnote{}\footnote{#1}%
  \addtocounter{footnote}{-1}%
  \endgroup
}

\begin{document}

\newpage
\begin{center}
{\LARGE\bf Spherical Principal Curves
\medskip
}
\vskip 7mm

{\large\sc Jongmin Lee, Jang-Hyun Kim, and Hee-Seok Oh \blfootnote{The first two authors contributed equally to this work.}}\\
{\large Seoul National University\\
Seoul 08826, Korea}
\end{center}
\vskip 5mm

\noindent 
{\bf Abstract}: This paper presents a new approach for dimension reduction of data observed on spherical surfaces. Several dimension reduction techniques have been developed in recent years for non-Euclidean data analysis. As a pioneer work, \cite{Hauberg} attempted to implement principal curves on Riemannian manifolds. However, this approach uses approximations to process data on Riemannian manifolds, resulting in distorted results. This study proposes a new approach to project data onto a continuous curve to construct principal curves on spherical surfaces. Our approach lies in the same line of \cite{Hastie} that proposed principal curves for data on Euclidean space. We further investigate the stationarity of the proposed principal curves that satisfy the self-consistency on spherical surfaces. The results on the real data analysis and simulation examples show promising empirical characteristics of the proposed approach.


\pagenumbering{arabic}


\section{Introduction}
A variety of dimension reduction techniques have been developed to represent and analyze data on Euclidean space. Recently, there have been growing interests in the analysis of non-Euclidean data with a variety of applications; directional data \citep{Mardia1977, Mardia2014, Gray}, shape data \citep{Kendall1984, Huckemann2006, Huckemann2010, Mallasto}, and motion analysis \citep{Hauberg, Mallasto}. For example, \cite{Siddiqi} and \cite{Cippitelli} introduced a Cartesian product of sphere $S^2$ and $\mathbb{R}$ for medial representation and skeleton data, respectively. For these representations, the conventional dimension reduction methods on Euclidean space have been modified by considering geodesics on non-Euclidean space \citep{Fletcher2004, Huckemann2006, Huckemann2010, Jung2011, Jung2012, Panaretos}. As a study closely related to our proposal, \cite{Hauberg} developed principal curves on Riemannian manifolds. However, \cite{Hauberg}  uses an approximate method by projecting data onto a \textit{finite} set of points, unlike the original principal curve in \cite{Hastie} which projects data onto a \textit{continuous} curve. This approximate projection causes a problem that may project different data points onto a single point mistakenly. This study proposes a new principal curve approach for spherical data by projecting the data onto a continuous curve without any approximations and improves the performance of dimension reduction. Our proposed approach is two-fold: One is an extrinsic approach that requires the setting of additional embedding space for a given manifold. The other is an intrinsic approach that does not need an embedding space. This intrinsic approach is difficult to calculate \citep{Srivastava}, but it is necessary to develop principal curves on generic manifolds. In this study, we investigate the stationarity of the principal curves on spherical surfaces from both approaches.

The remainder of this paper is organized as follows. Section 2 briefly reviews conventional principal curves and intrinsic and extrinsic means on manifolds. In \Cref{circle}, a newly developed exact principal circle on spheres is studied, which is used for the initialization of the proposed principal curves. \Cref{proposed} presents the proposed principal curves with a practical algorithm and investigates the stationarity of them theoretically. In \Cref{numerical:experiment}, the experimental results of the proposed method are provided through real earthquake data from the U.S. Geological Survey, real motion capture data, and simulation studies on $S^2$ and $S^4$. \Cref{proofs} discusses a justification of \textit{exact} projection step and rigorous proofs of theoretical properties of the proposed principal curves. Finally, concluding remarks are given in Section 7.

The main contributions of this study can be summarized as follows: (a) We propose both extrinsic and intrinsic approaches to form principal curves on $d$-sphere $S^d$, $d\ge 2$. (b) We verify the stationarity of the proposed principal curves on $S^d$. (c) We show the usefulness of the proposed method through real data analysis and simulation studies.


\section{Backgrounds}
\subsection{Principal Curves}
The principal curve in \cite{Hastie} can be considered as a nonlinear generalization of PCA that finds an affine subspace maximizing the variance of the projections of data. A curve is a function from one-dimensional closed interval to a given space, and a curve $f$ is called self-consistent or a \textit{principal curve} of a random variable $X$ if the curve satisfies 
\begin{equation}
\label{pc}
f(\lambda) = \mathbb{E}[X \ | \ \lambda_{f}(X)=\lambda],
\end{equation}
where $\lambda_f(x)$ is a projection index of a point $x$ onto the curve $f$. It implies that $f(\lambda)$ is the average of all data points projected onto $f(\lambda)$ itself. One of the most important consequences of the self-consistency is that the principal curve is a critical point with respect to reconstruction error for small perturbations \citep{Hastie}. However, it is difficult to formulate a principal curve by solving the self-consistency equation of (\ref{pc}). Thus, \cite{Hastie} represented a curve as the first order spline, connected by $T$ points. Then, they iteratively updated the curve to achieve the self-consistency condition using the following two steps, \textit{projection} and \textit{expectation}: (a) In the projection step, the given data are projected onto the curve. (b) In the expectation step, $T$ points of the curve are updated to satisfy the self-consistency. 

\subsection{Means on Manifolds}
Manifold is a topological space that locally resembles a Euclidean space. \textit{Riemannian manifold} $M$ is a smooth manifold equipped with smoothly varying inner product on tangent space. A \textit{(minimal) geodesic} is the shortest curve between two points in $M$ and its length is called \textit{geodesic distance}, denoted by $d_{Geo}(\cdot,\, \cdot)$. The class of Riemannian manifold includes a variety of spaces, such as Euclidean space $\mathbb{R}^d$, sphere $S^d$ \citep{Mardia2014, Mardia1977, Gray}, $\mathbb{R}P^d$ and $\mathbb{C}P^d$ (Kendall's shape space; \citep{Kendall1984, Huckemann2006, Huckemann2010}), $PD(d)$ (space of $d\times d$ symmetric positive definite matrices; \citep{Fletcher2007, Mallasto}), and product space of $S^2$ (medial representations; \citep{Siddiqi, Fletcher2004, Jung2011}). For more details about Riemannian manifold, see \cite{Boothby}. 

The concept of the expected value of a distribution can be naturally extended to manifolds, called \textit{Fr\'echet mean}. Given a probability distribution $Q$ on $M$ with a distance $\rho(\cdot, \cdot)$, the Fr\'echet mean $m\in M$ is defined as 
\begin{equation*}
\argmin\limits_{m\in M} \int \rho^2(m,\, x) Q(dx). 
\end{equation*}
The Fr\'{e}chet mean with geodesic distance is termed \textit{intrinsic mean} \citep{Bhattacharya2003}. Meanwhile, by embedding a given manifold $M$ into Euclidean space $\mathbb{R}^d$, the Fr\'echet mean can be calculated using Euclidean distance in $\mathbb{R}^d$, called \textit{extrinsic mean}. With an embedding $\xi:M \hookrightarrow \mathbb{R}^d$, the extrinsic mean is defined as  
\begin{equation*}
\label{frechetmean1}
\argmin\limits_{m\in M} \int \|\xi(m) - \xi(x)\|^2 Q(dx).
\end{equation*}
It is equivalent to the projection of the expectation in $\mathbb{R}^d$ to $M$ \citep{Bhattacharya2003}. That is, given a projection mapping $\pi: \mathbb{R}^d \to M$ defined as $\pi(y) = \argmin\limits_{m\in M} \|\xi(m) - y\|$, the extrinsic mean can be calculated as
$\pi\big( \int \xi(x) Q(dx)\big)$.
The extrinsic mean is computationally efficient compared to the intrinsic mean \citep{Bhattacharya2012}, and for a distribution $Q$ supported in a small region, the extrinsic mean is close with the intrinsic mean \citep{Bhattacharya2005}.

\subsection{Principal Curves on Riemannian Manifolds}\label{sec:hauberg}
\cite{Hauberg} proposed principal curves on Riemannian manifolds by expressing a curve as a set of $T$ points, $f = \{C_1,\, \ldots,\, C_{T} \}$, joined by geodesics. The estimation algorithm of the curve follows that of \cite{Hastie} with an approximation. Specifically, the mean operation in the expectation step is performed by intrinsic mean, and the projection is conducted by finding the nearest point in $f$ as 
\begin{equation*}
    \mbox{proj}(x) = \argmin_{C_i\in f} d_{Geo}\big(x,\, C_i\big),   
\end{equation*}
which is not an \textit{exact} projection onto the continuous curve.

\section{Enhancement of Principal Circle for Initialization} \label{circle}
Methods for fitting circles to data on $S^2$ are actively used in many applications, especially in astronomy and geology, to recognize undisclosed patterns of data \citep{Mardia1977, Gray}. This section improves the principal circle to be used as an initialization of the principal curves proposed in \Cref{proposed}. 

\subsection{Principal Geodesic and Principal Circle}
The principal curve algorithm of \cite{Hastie} uses the first principal component as the initial curve, which is easily calculated by singular value decomposition (SVD) of the data matrix in Euclidean space. Along with this line, the proposed principal curve algorithm in \Cref{proposed} requires an initial curve. The principal geodesic analysis (PGA) by \cite{Fletcher2004} can be considered as a generalization of PCA that performs dimension reduction of data on the Cartesian product of simple manifolds, such as $\mathbb{R}^3$, $S^2$, and $\mathbb{R}_+$. To this end, \cite{Fletcher2004} projected each manifold component of the data into a tangent space at the intrinsic mean of each component. As a result of the tangent space approximation of each component, data are approximated by points in Euclidean space, so applying PCA allows dimension reduction to be performed through the inverse process of the tangent projection, \textit{i.e.} {\it exponential map} that preserves a distance and angle at a base point. For spherical cases, they mainly perform tangent space projection using an inverse exponential map, called {\it log map}. The explicit forms of exponential and log maps of $S^2$ are described in \cite{Fletcher2004, Jung2011} and \cite{Jung2012}. 
 
However, PGA always results in a great circle going through the intrinsic mean on the sphere, as shown in \Cref{fig:pgpc}, and the class of great circles on a sphere is sometimes limited to suitably fit a dataset on the sphere \citep{Jung2011, Hauberg}. For example, the left panel of \Cref{fig:pgpc} shows earthquake data from the U.S. Geological Survey showing the location (blue dot) of significant earthquakes with Mb magnitude 8 or higher around the Pacific since 1900. The data will be analyzed in detail in \Cref{numerical:experiment}. In \Cref{fig:pgpc}, while the result (pink) by PGA does not fit the data correctly, our principal circle (red), presented later in \Cref{exact}, improves the representation of the data. Further, in the right panel of \Cref{fig:pgpc}, our principal circle suitably fits the circular simulated data, whereas the result (pink) by PGA does not capture the variation of the data. The PGA's failure stems from the fact that the above two data sets are far from their intrinsic means, as noted in \cite{Jung2011}, \cite{Jung2012}, and \cite{Hauberg}.
\begin{figure}
	\centering
	\includegraphics[scale=0.25]{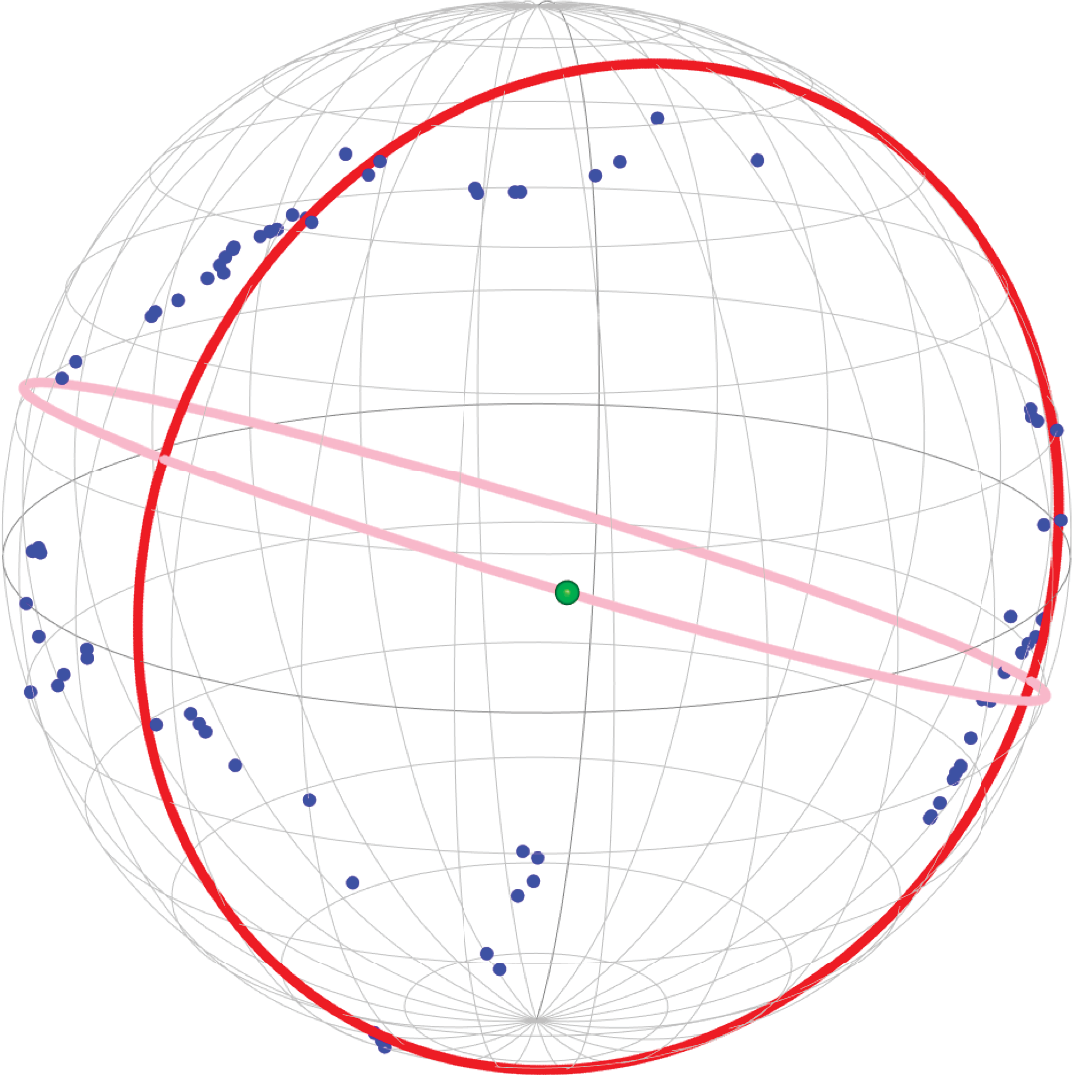}
	\hspace{1cm}
	\includegraphics[scale=0.25]{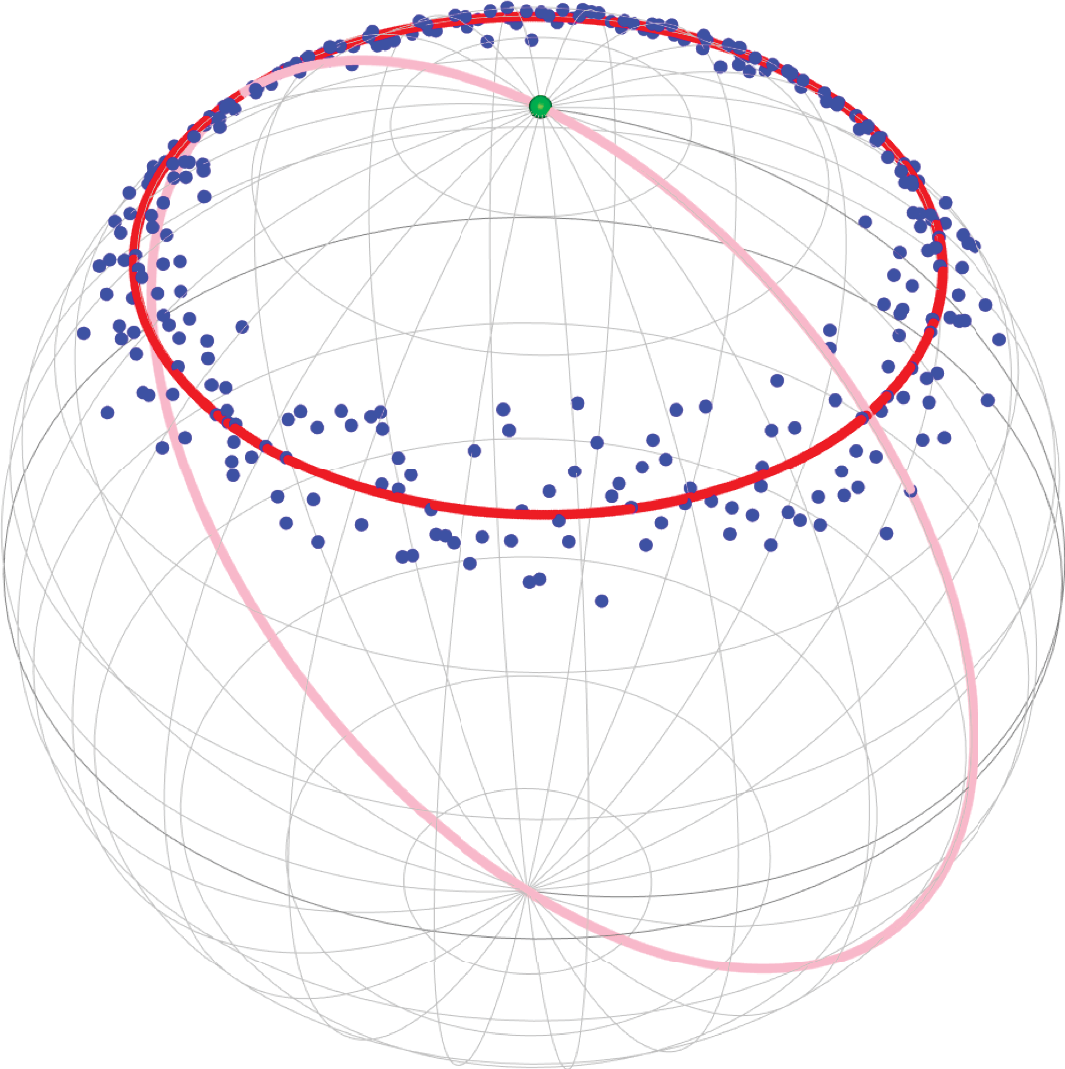}
	\caption{Left: Spherical distribution of significant earthquakes (blue) with its intrinsic mean (green), the result (pink) by PGA, and the result (red) by our proposed principal circle. Right: Circular simulated data (blue) with its intrinsic mean (green), the result (pink) by PGA, and the result (red) by our proposed  principal circle.}
	\label{fig:pgpc} 
\end{figure}

In the literature, there is an attempt by \cite{Jung2011} that generalizes the PGA to a circle on $S^2$. The circle on $S^2$ that minimizes a reconstruction error is called \textit{principal circle}, where the reconstruction error is defined as the total sum of squares of the geodesic distance between the curve and the data. \cite{Jung2011} used a double iteration algorithm that uses the log map to project the data into the tangent space and then finds the principal circle. However, this approach has two problems. First, using the tangent approximation when minimizing the distance may causes numerical errors. If the data points are located away from the mean, the numerical errors may increase because there is no local isometry between the sphere and its tangent plane according to the \textit{Gauss's Theorema Egregium} (see p. 363-370 of \cite{Boothby} or Ch 8 of \cite{Tu} for details). Second, due to the topological difference between the sphere and the plane, the existence of principal circles in the tangent plane is not guaranteed. For example, \Cref{fig:inhomo} shows simulated data, where the underlying structure is a great circle, and the intrinsic mean is the North Pole $(0,\, 0,\, 1)$, where the data points are mostly concentrated around the North Pole. From the compactness of the sphere, the least-squares circle always exists regardless of the data structure. It is an advantage of the intrinsic approaches. On the other hand, the least-squares circle does not exist if the data points projected onto the tangent space at their intrinsic mean are collinear, as shown in the middle and right panels of \Cref{fig:inhomo}. It coincides that several circle fitting procedures in a plane, such as \cite{Kasa} and \cite{Coope}, fail when the data points are collinear, as noted in \cite{Umbach}. Moreover in this case, the (tangent) plane cannot consider the periodicity of the data, as opposed to the left panel of \Cref{fig:inhomo}. Ignoring the periodic structure of data, as noted in \cite{Eltzner}, may reduces the efficiency of a method.
\begin{figure}
	\centering
	\includegraphics[scale=0.21]{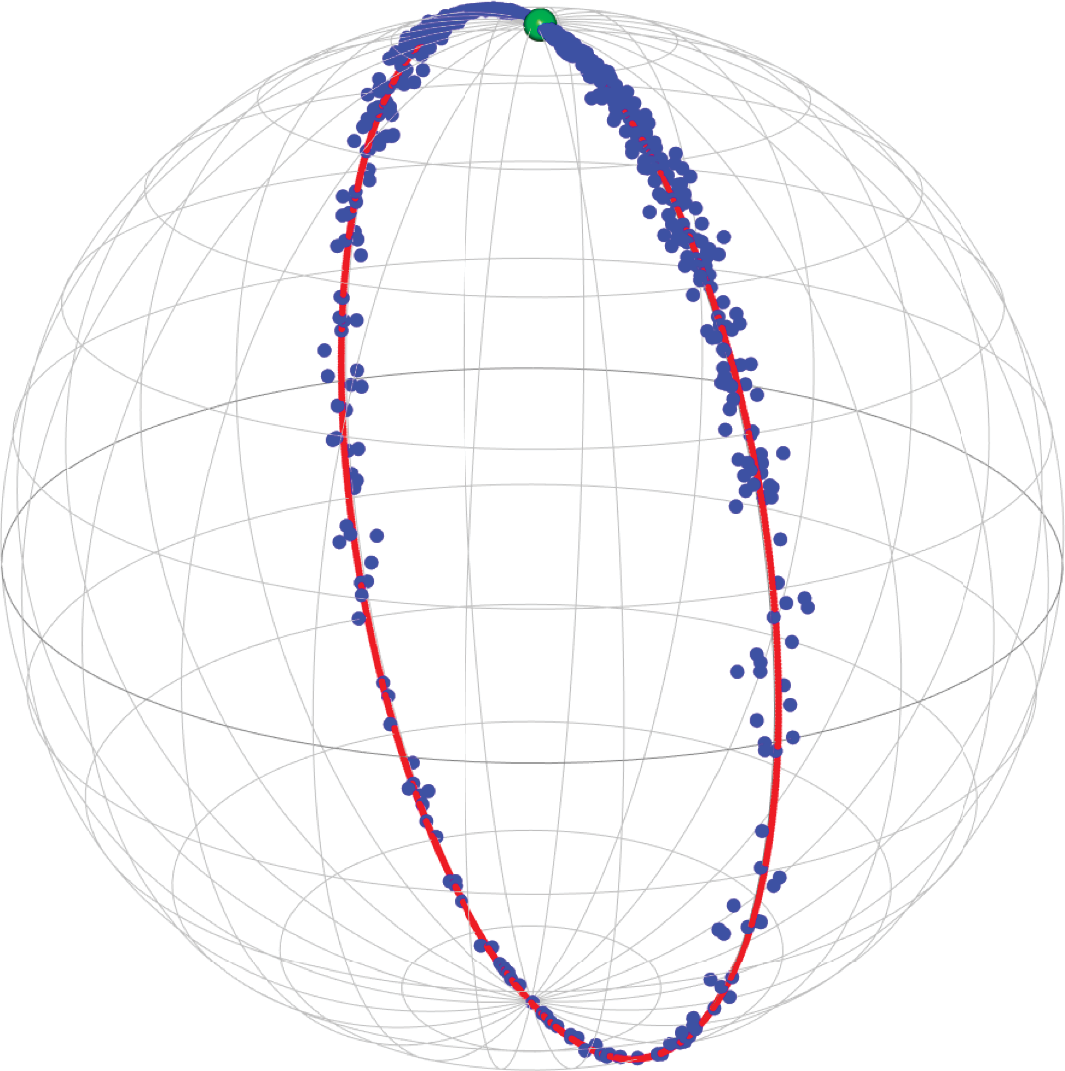}
	\includegraphics[scale=0.22]{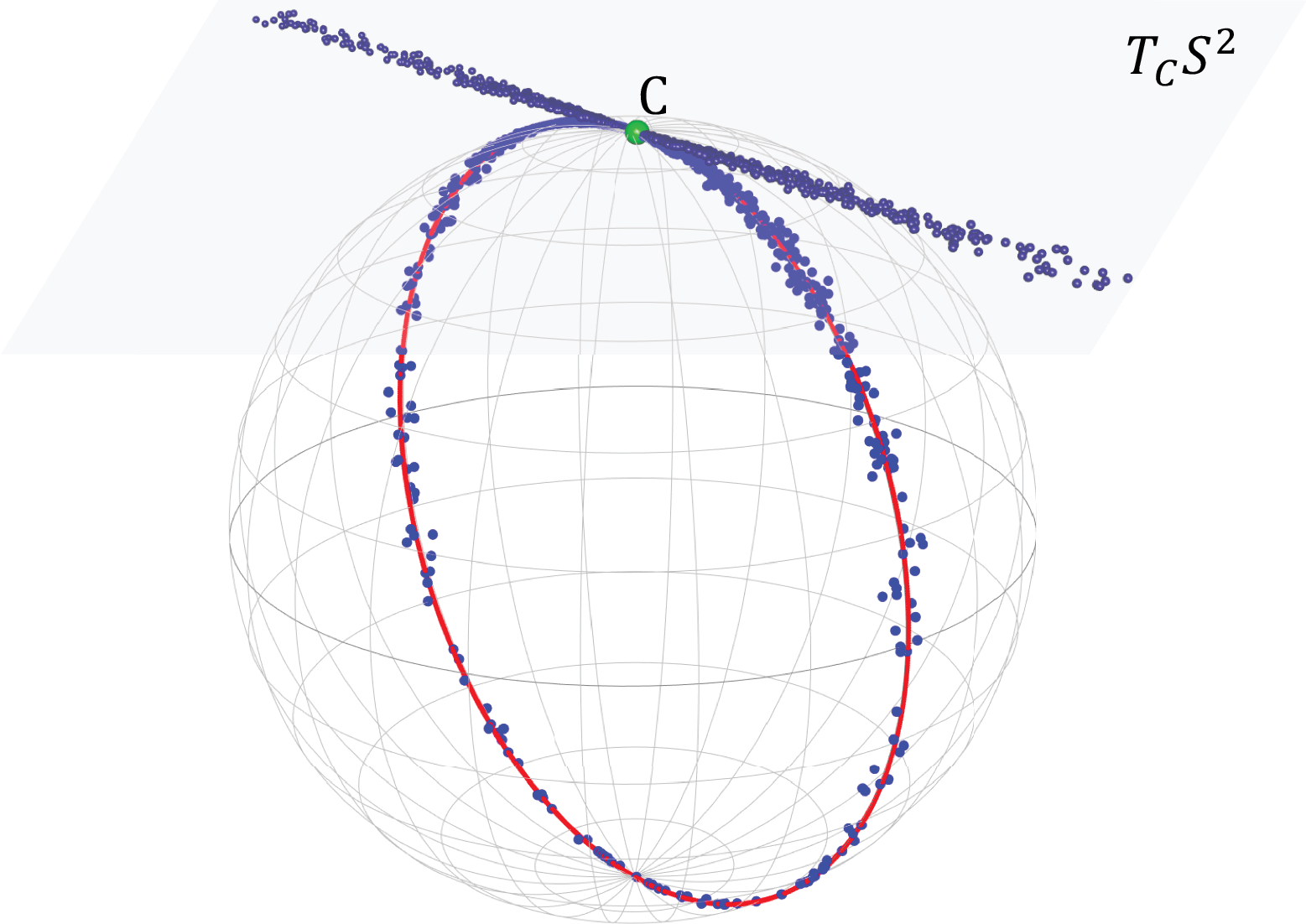}
    \includegraphics[scale=0.17]{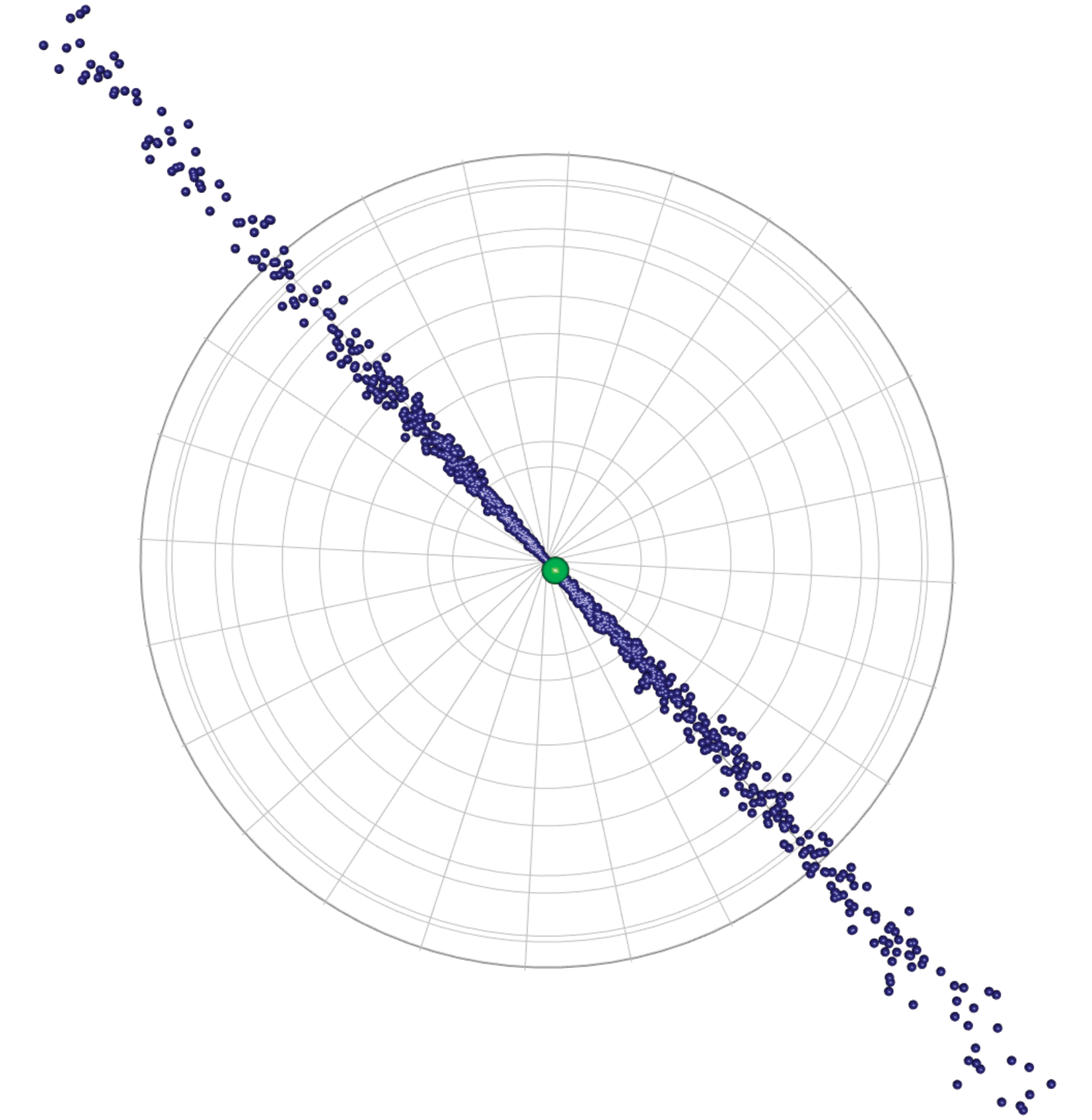}
	\caption{Left: Simulated data points (blue) with the intrinsic mean $(0,\, 0,\, 1)$ (green), and the result by our proposed principal circle (red). Middle: The projected points from the sphere onto the tangent plane at $C=(0,\, 0,\, 1)$. Right: The projected points viewed from above the Northern Hemisphere.}
	\label{fig:inhomo} 
\end{figure}

This study proposes a new principal circle that does not rely on tangent projection for better initialization of the proposed principal curve presented in \Cref{proposed}. We obtain the constraint-free optimization problem by expressing the center of the circle using the spherical coordinate system in \Cref{exact} and \Cref{hyper}.

\subsection{Exact Principal Circle}\label{exact}
For our principal circle, we consider an intrinsic optimization algorithm that does not use any approximations. Let $d_{Geo}(x,\, y)$ be the geodesic distance between $x,\, y \in S^2$. For a given dataset $D$ and a circle $C$ on $S^2$, let $\delta(D,\, C)$ be the sum of squares of distances between circle and data, defined as
\[
\delta(D,\, C) = \sum_{x\in D} d_{Geo}\big(x,\,  \mbox{proj}_C(x)\big)^{2}, 
\]
where $\mbox{proj}_C(x)$ denotes a projection of $x$ on $C$. The goal is to find a circle $C$ on $S^2$ that minimizes $\delta(D,\, C)$. To solve this optimization problem, we represent a circle $C$ by a center $c$ of the circle and a radius $r\in [0,\, \pi]$, the geodesic distance between the center $c$ and the circle $C$. This representation is not unique \citep{Jung2011}. For example, let $c^\prime \in S^2$ be the antipodal point of $c$ that is diametrically opposite to $c$ on $S^2$, then $(c,\, r)$ and $(c^\prime,\, \pi-r)$ represent the same circle $C$. Nevertheless, it is not crucial to the optimization problem because we simply find a representation of the least square circle. By using a spherical coordinate system, it is able to parameterize $c$ as $(\theta,\, \rho)$, where $\theta$ denotes the azimuthal angle and $\rho$ is the polar angle. By symmetry of the circle, $d_{Geo}\big(x,\, \mbox{proj}_C(x)\big)$ can be easily calculated by
\[
d_{Geo}\big(x,\, \mbox{proj}_C(x)\big) = d_{Geo}(x,\, c) - r.
\]
Thus, we have  
\begin{equation}
\label{ss1}
\delta(D,\, C) = \sum_{x\in D} \big(d_{Geo}(x,\, c) - r\big)^{2}.
\end{equation}

With letting $c = (\theta_{c},\, \rho_{c})$ and $x = (\theta_{x},\, \rho_{x})$ in the spherical coordinate system, the geodesic distance $d_{Geo}(x,\, c)$ is given by the spherical law of cosines with three points $c$, $x$, and the polar point (see Lemma 3 in \Cref{proof:stationarity} below for details)
\begin{equation}
\label{dgeo1}
d_{Geo}(x,\, c)=\arccos\big(\cos\rho_{c}\cos\rho_{x} +\sin\rho_{c}\sin\rho_{x}\cos(\theta_{c} - \theta_{x})\big).
\end{equation}
By putting (\ref{dgeo1}) into (\ref{ss1}), it follows that $\delta(D,\, C)$ is represented as a three-parameter differentiable function $\delta_{D}(\theta_{c},\, \rho_{c},\, r)$ in domain $[0,\, 2\pi] \times [0,\, \pi] \times [0,\, \pi]$ as follows,
\begin{eqnarray}
\label{delta1}
\delta_{D}(\theta_c,\, \rho_c,\, r) = \sum_{x\in D} \big(\arccos\big(\cos\rho_c\cos\rho_x + \sin\rho_c \cdot \sin\rho_x\cos(\theta_c - \theta_x)\big)-r \big)^2.
\end{eqnarray}
Since $[0,\, 2\pi]\times [0,\, \pi] \times [0,\, \pi]$ is compact, the function $\delta_{D}(\theta_{c},\, \rho_{c},\, r)$ holds a global minimum value. Thus, it can apply the gradient descent method to find the solution. Here is the algorithm to find a principal circle from the above description.
\begin{algorithm}[H]
	\caption{~~Exact Principal Circle by gradient descent}
	\label{alg1}
	\begin{algorithmic}
		\State Initialize $(\theta_{c},\, \rho_{c},\, r)$ as ($\overline{\theta},\, \overline{\rho},\, \pi/2$)
		\While{ ($\Delta \delta(D,\, C) \ge \mbox{threshold}$) } 
		\State $(\theta_{c},\, \rho_{c},\, r) \leftarrow (\theta_{c}, \rho_{c},\, r) - \beta \nabla \delta_{D}(\theta_{c},\, \rho_{c},\, r) $
		\EndWhile  
	\end{algorithmic}
\end{algorithm}
As in many nonlinear least-square algorithms, such as Gauss-Newton algorithm and Levenberg-Marquardt algorithm (see Ch 4 of \cite{Scales} for details), the above Algorithm 1 may converge to a local minimum or a saddle point instead of the global minimum, since $\delta_D(\theta_c,\, \rho_c,\, r)$ is non-convex. Thus, initial values should be selected carefully. If the data points in $D$ are not too apart and localized, then it is reasonable to choose $(\theta_x,\, \rho_x,\, \pi/2)$ for some $x\in D$ as an initial. The spherical coordinates of the intrinsic mean of $D$ with radius $r=\pi/2$, denoted by ($\overline{\theta},\, \overline{\rho},\, \pi/2)$, if necessary with varying $r\in [0,\, \pi]$, is also recommended as initial values. In the case of a non-localized data set, one can implement the algorithm with various initial settings as much as one wants, compare the consequences of $\delta$, and finally choose the circle with the lowest $\delta$ as the principal circle. Note that, in existing methods for fitting circles to data on spheres, such as \cite{Gray, Jung2011}, and \cite{Jung2012}, there are no assurances that their algorithms finally achieve the circle minimizing (\ref{ss1}). Although $\delta$ is not convex globally, it is convex on a neighborhood of a global minimum point. Hence, it is reasonably expected that if an initial value is suitably close to an optimum point, then Algorithm 1 converges to the optimum. A specification about the neighborhood for which $\delta$ is convex, and rigorous proof for convergence of \Cref{alg1} on that neighborhood remains a challenge. In the real data analysis and the simulated studies later on \Cref{numerical:experiment}, however, implementations of Algorithm 1 with several initial values result in almost the same principal circles and converge rapidly. Thus, there are no practical difficulties in our experiments. In addition, $\beta$ is the step size of Algorithm 1, and it relies on the dataset $D$. The algorithm may diverge when $\beta$ is large (e.g., greater than .01). In simulated examples and real data on \Cref{numerical:experiment}, we use .001. Since too small $\beta$ causes computational time to be high, an appropriate $\beta$ should be selected properly throughout experiments from a relatively larger value of $\beta$ to the lower one.

\subsection{Extension to Hyperspheres}\label{hyper}
In the case of high-dimensional spheres, to find a one-dimensional circle that attempts to represent a given data closely, we provide both extrinsic and intrinsic ways. The former is easy to implement and more computationally feasible because it uses an extrinsic approach and is not  exactly found. The latter directly extends the exact principal circle in the previous section into higher-dimensional spheres using the framework of \textit{principal nested spheres} \citep{Jung2012}; however, it takes time to compute compared to the former approach.

\subsubsection{Circle as an Initialization}
 Later in \Cref{simul:S^4}, we will use the following extrinsic method as an initial estimate of the spherical principal curves for waveform simulated data on $S^4$. Specifically, we consider $S^d=\{y=(y_1,\, y_2,\, ...,\, y_{d+1})\in \mathbb{R}^{d+1} \ | \  \sum_{i=1}^{d+1} y_{i}^2=1\}$ for $d\ge 2$, as an embedded surface in the ambient space $\mathbb{R}^{d+1}$. That is, $\left\{x_i \right\}_{i=1}^n \subset S^d \hookrightarrow \mathbb{R}^{d+1}$ are regarded as elements in $\mathbb{R}^{d+1}$, not taking into account a nonlinear dependence of the data; though, ensuring lower computational complexity. Note that any one-dimensional circle on $S^d$ is an intersection of a two-dimensional plane and $S^d$. Hence, the strategy is to find the 2-plane $P \subset \mathbb{R}^{d+1}$ that closely represents the data $\left\{x_i \right\}_{i=1}^n$ with respect to the standard distance in $\mathbb{R}^{d+1}$, rather than geodesic distance in $S^d$. That is, the plane $P$ is the two-dimensional vector subspace of $\mathbb{R}^{d+1}$ spanned by first two principal components of the data, and then $P\cap S^d$ is a one-dimensional circle to find. Although the extrinsic circle is capable of approximating the meaningful data, there may be some instances that need more precise initial estimate for the data.

\subsubsection{Exact Principal Circle}\label{hypercircle}
For a better initial guess of the proposed principal curves, we provide an exact principal circle on $S^d=\{y=(y_1,\, y_2,\, ...,\, y_{d+1})\in \mathbb{R}^{d+1} \ | \  \sum_{i=1}^{d+1} y_{i}^2=1\}$ for $d\ge 3$. The arguments in the \Cref{exact} can be applied to higher-dimensional spheres $S^d$ for $d\ge 3$ if the geodesic distance of Equation (\ref{ss1}) can be precisely calculated. To this end, let $D=\left\{x_i\right\}_{i=1}^n$ be a dataset on $S^d$, and denote a $(d-1)$-dimensional subsphere on $S^d$ as $C$. Using a spherical coordinates for $S^d$, $x=(x_1,\, x_2,\, ...,\, x_d,\, x_{d+1})\in S^d\subset \mathbb{R}^{d+1}$ can be parametrized as 
\begin{eqnarray*}
x_1 &=& \cos(\varphi_1) \\ 
x_2 &=& \sin(\varphi_1) \cos(\varphi_2), \\
x_3 &=& \sin(\varphi_1) \sin(\varphi_2) \cos(\varphi_3) \\ 
\vdots & & \\
x_d &=& \sin(\varphi_1)\cdot \cdot \cdot \sin(\varphi_{d-1})\cdot \cos(\varphi_d) \\
x_{d+1} &=& \sin(\varphi_1)\cdot \cdot \cdot \sin (\varphi_{d-1}) \cdot \sin(\varphi_d), 
\end{eqnarray*}
where $\varphi_1,\, \varphi_2,\, \cdots, \varphi_{d-1},\, \varphi_d$ are angular coordinates with $\varphi_{d}\in [0,\, 2\pi)$ and the others ranging over $[0,\, \pi)$. Note that $d_{Geo}(x,\, c) = \arccos(x\cdot c)$, where $\cdot$ denotes the (standard) inner product in $\mathbb{R}^{d+1}$. Thus,
\begin{eqnarray}
\label{dgeo2}
d_{Geo}(x,\, c)=\arccos\big(\cos(\varphi_{1c})\cos(\varphi_{1x}) + \sum_{k=1}^{d-2} \big[\prod_{i=1}^k \sin(\varphi_{ic})\sin(\varphi_{ix})\big]\cdot\cos(\varphi_{(k+1)c})\cos(\varphi_{(k+1)x}) \nonumber \\
+   \big[\prod_{i=1}^{d-1}\sin(\varphi_{dc})\sin(\varphi_{dx})\big] \cdot \cos(\varphi_{dc}-\varphi_{dx}) \big), \nonumber \\
\end{eqnarray}
where $\left\{\varphi_{ic}\right\}_{i=1}^d$ and $\left\{\varphi_{ix} \right\}_{i=1}^d$ are the corresponding angular coordinates of $c$ and $x$, respectively. By putting (\ref{dgeo2}) into (\ref{ss1}), it follows that $\delta(D,\, C)$ is represented as a $(n+1)$-parameter differentiable function $\delta_{D}(\varphi_{1c},\, ...,\, \varphi_{dc},\, r)$ in domain $[0,\, \pi]^{d-1}\times[0,\, 2\pi]\times[0,\, \pi]$ as follows, 
\begin{align}
\label{delta2}
\delta_{D}(\varphi_{1c},\, \ldots,\, \varphi_{dc},\, r)=\sum_{x\in D} \bigg(\arccos&\Big(\cos(\varphi_{1c})\cos(\varphi_{1x}) \nonumber \\ 
&+\sum_{k=1}^{d-2} \big[\prod_{i=1}^k \sin(\varphi_{ic})\sin(\varphi_{ix})\big]\cdot\cos(\varphi_{(k+1)c})\cos(\varphi_{(k+1)x}) \nonumber \\
&+\big[\prod_{i=1}^{d-1}\sin(\varphi_{ic})\sin(\varphi_{ix})\big] \cdot \cos(\varphi_{dc} - \varphi_{dx})\Big) -r\bigg)^2.
\end{align}
Note that, in the case of $d=2$, the above equation (\ref{delta2}) becomes (\ref{delta1}). $\delta_{D}$ holds a global minimum value due to the compactness of the domain $[0,\, \pi]^{d-1}\times[0,\, 2\pi]\times[0,\, \pi]$. Therefore, an exact principal circle on $S^d$ can be obtained by gradient descent, the same way in Algorithm \ref{alg1}, except that the number of parameters is $d+1$. Let $(\overline{\varphi_1},\,  \overline{\varphi_2},\, \ldots,\, \overline{\varphi_d})$ denote the spherical coordinates of the intrinsic mean of $D$. Here is the algorithm to find a principal circle on $S^d$.
\begin{algorithm}[H]
	\caption{~~Exact Principal Nested Sphere on Hypersphere $S^d$}
	\label{alg2}
	\begin{algorithmic}
		\State Initialize $(\varphi_{1c},\, \varphi_{2c},\, ...,\, \varphi_{dc},\, r)$ as ($\overline{\varphi_1},\, \overline{\varphi_2},\, ...,\, \overline{\varphi_d},\, \pi/2$)
		\While{ ($\Delta \delta(D,\, C) \ge \mbox{threshold}$) } 
		\State $(\varphi_{1c},\, \varphi_{2c},\, \ldots,\, \varphi_{dc},\, r) \leftarrow (\varphi_{1c},\, \varphi_{2c},\, \ldots,\, \varphi_{dc},\, r) - \beta \nabla{\delta_{D}} (\varphi_{1c},\, \varphi_{2c},\, ...,\, \varphi_{dc},\, r)$. 
		\EndWhile  
	\end{algorithmic}
\end{algorithm}
It is possible that Algorithm \ref{alg2} converges to a local minimum or a saddle point of $\delta_{D}$, owing to its non-convexity. Therefore, an initial value should be carefully chosen, for instance, a data point in $D$ and the intrinsic mean of $D$. The discussions about initial values and step size $\beta$ are the same as those of Algorithm \ref{alg1}. 

By applying the Algorithm \ref{alg2} to a given data iteratively, we can obtain a one-dimensional sphere, \textit{i.e.,} an exact principal circle on $S^d$ that can be the initialization of the spherical principal curves. For more details about the procedure, see \cite{Jung2012}. It is noteworthy that from the perspective of the principal nested spheres, our method can be applied to find nested spheres in an \textit{exact} way.


\section{Proposed Principal Curves}\label{proposed}
This section presents our new \textit{exact} principal curves on $d$-sphere $S^d$ for $d\ge 2$ from both intrinsic and extrinsic perspectives. We further investigate the stationarity of the proposed principal curves. 

\subsection{Exact Projection Step on $S^d$}
As mentioned in \Cref{sec:hauberg}, the approach of \cite{Hauberg} does not perform the exact projections onto curves. On the other hand, the exact projections on $S^d$ for $d\ge 2$ are carried out in our method, which results in more elaborated principal curves. To this end, we parameterize the curve as a set of $T$ points joined by geodesics as in \cite{Hauberg}. Specifically, we first project the data point to each geodesic segment of the curve and then obtain the exact projection on the curve by choosing the closest geodesic segment. Let $\lambda_f(x)$ be the projection index of a point $x$ to the curve $f(\lambda)$ for $\lambda \in [0,\, 1]$, 
\begin{align}\label{eq:proj}
    \lambda_f(x) = \argmin_\lambda d_{Geo}\big(x,\, f(\lambda)\big).
\end{align} 
The projection of $x$ onto the curve can be obtained as $f\big(\lambda_f(x)\big)$.

The following subsections describe a procedure for projecting a point onto a geodesic segment on $S^d$. Given $A$,\, $B$,\, $C \in S^{d}\subset \mathbb{R}^{d+1}$, we find the closest point to $C$ on the geodesic segment joining $A$ and $B$. When $A=B$, the process is obvious, and in the case of $A=-B\in \mathbb{R}^{d+1}$, there is no unique geodesic connecting $A$ and $B$. Hence, we only consider the case that $A$ and $B$ are linearly independent, \textit{i.e.}, $(A\cdot B)^2\neq 1$, where $\cdot$ denotes the dot product in $\mathbb{R}^{d+1}$. We first deal with the projection on $S^2$ and then extend it into hyperspherical cases.

\subsubsection{Projection on $S^2$}\label{sec:s2}

\begin{figure*}[!ht]
    \centering
    \begin{subfigure}[b]{0.27\textwidth}
        \includegraphics[width=\textwidth]{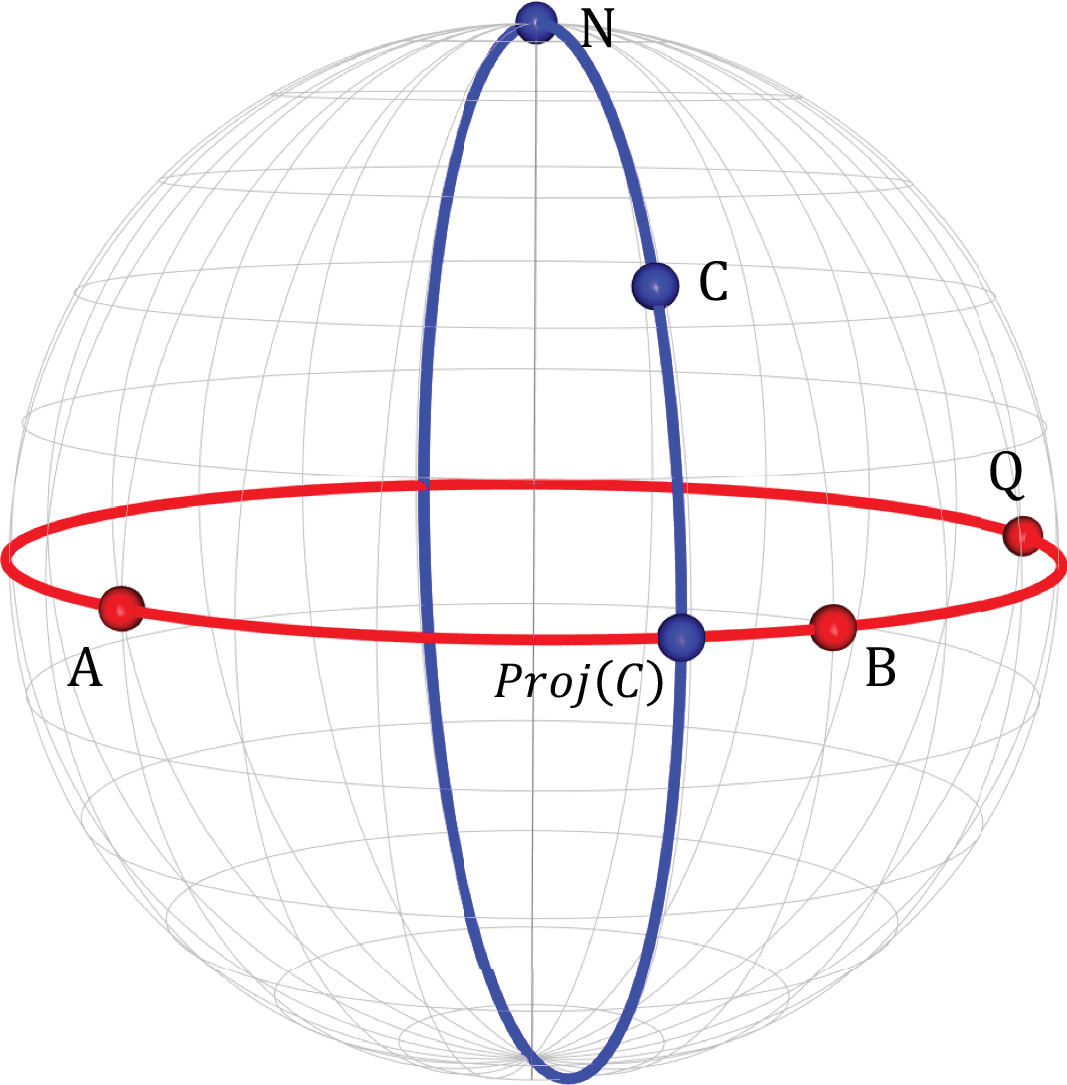}
        \caption{}
        \label{fig:proj(a)}
    \end{subfigure}
        \hspace{1cm}
    \begin{subfigure}[b]{0.27\textwidth}
        \includegraphics[width=\textwidth]{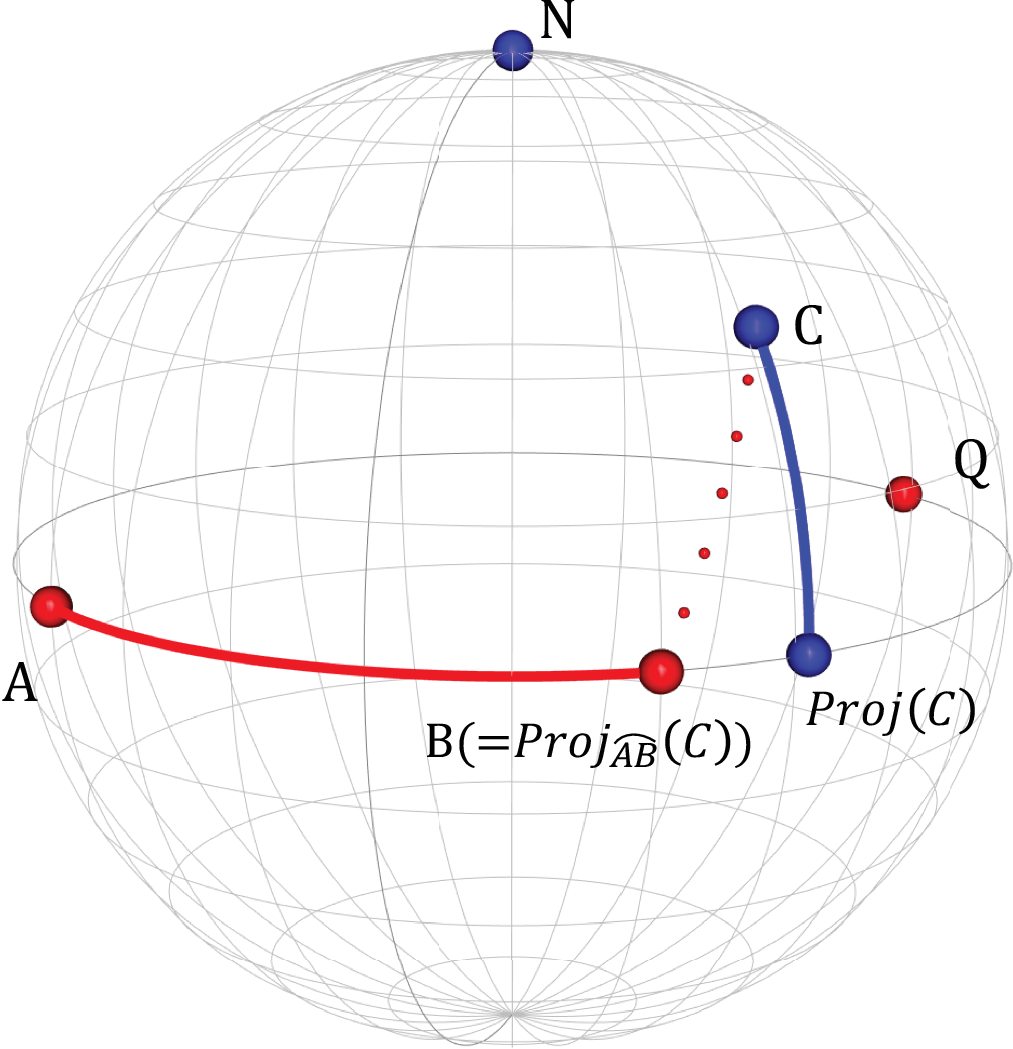}
        \caption{}
        \label{fig:proj(b)}
    \end{subfigure}
        \hspace{1cm}
    \begin{subfigure}[b]{0.27\textwidth}
        \includegraphics[width=\textwidth]{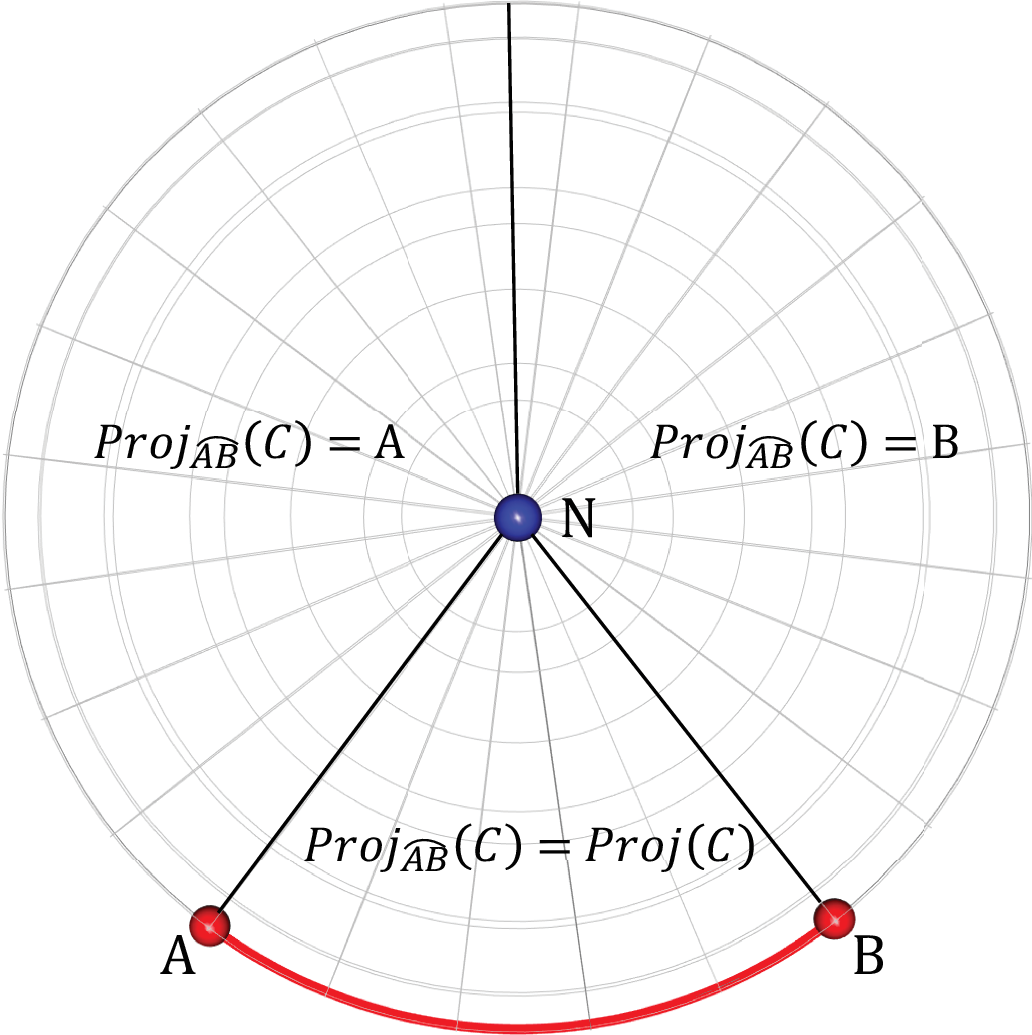}
        \caption{}
        \label{fig:proj(c)}
    \end{subfigure}
    \caption{Illustration of the projection procedure on $S^2$. 
(a) The case that $C$ is projected inside $\wideparen{AB}$, \textit{i.e.}, $\mbox{proj}_{\wideparen{AB}}(C)=\mbox{proj}(C)$ and $I\ge 0$. The projection of $C$ is an intersection point of two great circles. (b) The case that $C$ is projected onto $B$ in a non-orthogonal way (red dotted line), \textit{i.e.}, $\mbox{proj}(C)\ne \mbox{proj}_{\wideparen{AB}}(C)=B$ and $I<0$. (c) An image of the sphere viewed from above the Northern Hemisphere in the projection of $C$. 
}
    \label{fig:proj}
\end{figure*}
 
Before describing the projection procedure on $S^2$, it is important to notice that $(A\cdot B)^2 \neq 1$ is equivalent to $A\times B\neq 0$, where $\times$ denotes the cross product in $\mathbb{R}^3$. In addition, if $A\times B/\norm{A\times B}=\pm C$, then any points on geodesic through $A$ and $B$ have the same distance from $C$. From now on, we assume $A\times B/\norm{A \times B}\ne \pm C$.

\Cref{fig:proj} shows the projection procedure. We define the North Pole $N$ concerning $A$ and $B$ as $N=\frac{A\times B}{\norm{A \times B}} \in S^2$ and a center $Q$ of the great circle through $N$ and $C$ as $Q= \frac{N\times C}{\norm{N\times C}} \in S^2$ that is contained in the great circle through $A$ and $B$. Then, the projection of $C$ onto the great circle through $A$ and $B$, $\mbox{proj}(C)$, becomes an intersection point of two great circles, as shown in \Cref{fig:proj(a)}, 
\begin{equation*}
   \mbox{proj}(C) = Q\times N = \frac{(A\times B) \times C}{{\norm{(A \times B)\times C}}}\times \frac{(A\times B)}{\norm{A\times B}} \in S^2.
\end{equation*}
Note that $\mbox{proj}(C)$ is not always included in the geodesic segment $\wideparen{AB}$ joining $A$ and $B$ as \Cref{fig:proj(b)}. For this reason, we define an indicator $I=-\big(A-\mbox{proj}(C)\big)\cdot \big(B-\mbox{proj}(C)\big)$, indicating whether $\mbox{proj}(C)$ is inside $\wideparen{AB}$ or not, \textit{i.e.,} orthogonally projected onto $\wideparen{AB}$ or not. Finally, the projection of $C$ onto $\wideparen{AB}$, $\mbox{proj}_{\wideparen{AB}}(C)$, is 
\begin{equation*}
    \mbox{proj}_{\wideparen{AB}}(C)=
    \begin{cases}
     \mbox{proj}(C), &\mbox{if }I\ge 0 \\
     \argmin_{E\in \left\{A,\, B \right\}} d_{Geo}(C,\, E), & \mbox{if }I<0. 
    \end{cases}
\end{equation*}

\subsubsection{Projection on Hypersphere}
For $A,\, B,\, C \in S^{d}\subset \mathbb{R}^{d+1}$, if $B\cdot C=C\cdot A=0$, then all points on $\wideparen{AB}$ have the same geodesic distance of $\pi/2$ from $C$, which is verified in \Cref{proof:proj}; hence, assume that $A$, $B$, and $C$ do not satisfy $B\cdot C=C\cdot A=0$. Let $V$ be a two-dimensional vector space in $\mathbb{R}^{d+1}$ spanned by $A$ and $B$.

\begin{figure}[!ht]
  \centering
  \includegraphics[scale=0.30]{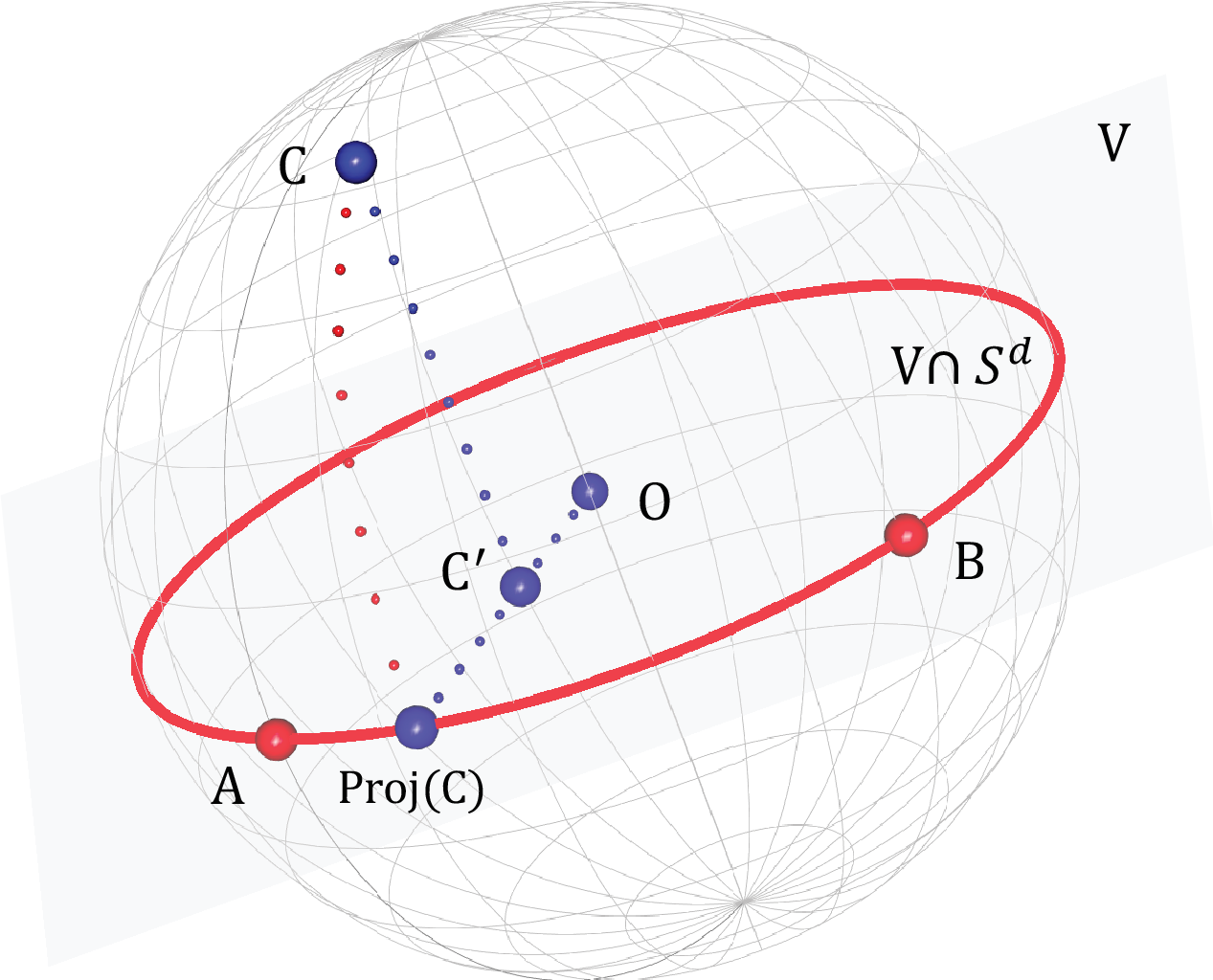}
  \caption{Illustration of the projection procedure on $S^d$.}
  \label{fighyper}
\end{figure}

As shown in \Cref{fighyper}, we aim to find the projection of $C$ onto $V\cap S^d$, $\mbox{proj}(C)$, by following two steps: (Step 1) Locate the projection of $C$ onto $V$, $C'$. (Step 2) Find the projection of $C'$ onto $V\cap S^d$. Note that the resulting projection is equivalent to the projection of $C$ onto $V\cap S^d$, $\mbox{proj}(C)$. The rigorous justification of the above procedure is provided in \Cref{proof:proj}.

(Step 1): We find the closest point $C'\in V$ from $C$. Let $C' = \mu A + \lambda B$ for $\mu,\, \lambda \in \mathbb{R}$. Then $C'$ should satisfy the orthogonal condition, $(C-C')\cdot A=(C-C')\cdot B=0$. By plugging the equation $C'=\mu A+\lambda B$ into the above condition and solving the systems of linear equations with respect to $\mu$ and $\lambda$, it follows that 
\[
C'=\frac{C\cdot A-(A\cdot B)(B\cdot C)}{1-(A\cdot B)^2}A + \frac{B\cdot C-(A\cdot B)(C\cdot A)}{1-(A\cdot B)^2}B,
\]
where the denominator is non-zero and $C'\neq 0 \in \mathbb{R}^{d+1}$ because of the assumptions; $(A\cdot B)^2\neq 1$, and $A$, $B$, and $C$ do not satisfy $B\cdot C=C\cdot A=0$.

(Step 2): The projection of $C'$ onto $V\cap S^d$, $\mbox{proj}(C)$, is obtained by just normalizing $C'$ so that it is in $S^d$. Therefore, we have 
\begin{eqnarray*}
     \mbox{proj}(C) = \frac{C'}{\norm{C'}} = \frac{\big(C\cdot A-(A\cdot B)(B\cdot C)\big)A+ \big(B\cdot C - (A\cdot B)(C\cdot A)\big)B}{\norm{\big(C\cdot A-(A\cdot B)(B\cdot C)\big)A+ \big(B\cdot C - (A\cdot B)(C\cdot A)\big)B}}.
\end{eqnarray*}

Similarly, we define the indicator  $I=-\big(A-\mbox{proj}(C)\big)\cdot \big(B-\mbox{proj}(C)\big)$ to find the projection of $C$ onto $\wideparen{AB}$, $\mbox{proj}_{\wideparen{AB}}(C)$. Due to the fact that $A$, $B$, and $\mbox{proj}(C)$ are in the one-dimensional unit circle $V\cap S^d$, we obtain $I\ne 0$ unless $\mbox{proj}(C)=A$ or $B$. Since $I$ is continuous with respect to $\mbox{proj}(C) \in V\cap S^d$, it indicates that whether $\mbox{proj}(C)$ is in $\wideparen{AB}$ or not. We finally obtain $\mbox{proj}_{\wideparen{AB}}(C)$ as
\begin{equation*}
    \mbox{proj}_{\wideparen{AB}}(C)=
    \begin{cases}
     \mbox{proj}(C), &\mbox{if }I\ge0 \\
     \argmin_{E\in \left\{A,B \right\}} d_{Geo}(C,\, E), & \mbox{if }I<0. 
    \end{cases}
\end{equation*}
Note that the distance between $C$ and $\wideparen{AB}$ is the geodesic distance from $C$ to $\mbox{proj}_{\wideparen{AB}}(C)$, which can be calculated as 
\begin{equation}
\label{dgeo}
d_{Geo}(C,\, \mbox{proj}_{\wideparen{AB}}(C))=\arccos(C\cdot \mbox{proj}_{\wideparen{AB}}(C)).
\end{equation}

\subsection{Expectation Step on $S^d$}
The expectation step follows the principal curve of \cite{Hauberg}, \textit{i.e.,} updates the weighted average with smoothing that makes the curve closer to the self-consistency condition. Suppose that we have $n$ data points $D = \{x_i\}_{i=1}^n$ and the corresponding projection indices $\{\lambda_i\}_{i=1}^n$, where $\lambda_i=\lambda_f(x_i)$ for $i=1,\, \ldots,\, n$. Let $T$ denote the number of points of an initial curve. Then, the local weighted smoother iteratively updates the $t^{\text{th}}$ point of the principal curve, $C_t$, with the weighted mean of data points. In this study, we use a quadratic kernel $k(\lambda) = (1-\lambda^2)^2\cdot \delta_{|\lambda|\le 1}$, as \cite{Hauberg}, and the weight of each data point is given by $w_{t,\, i}= k(|\lambda_f(C_t) - \lambda_i| / \sigma)$, where $\sigma = q\cdot (\mbox{length\ of}\ f)$.

\subsubsection{Extrinsic Approach}
The extrinsic mean on $S^d$ can be calculated by considering the canonical embedding $S^d \hookrightarrow \mathbb{R}^{d+1}$. Specifically, for a curve $f = \{C_1,\, ...,\, C_{T}\}$ and each point $C_{t}$, the extrinsic mean is obtained by averaging the data points represented in Euclidean coordinates as
\begin{align}
    m_t(D,\, f) =\sum_{i=1}^n w_{t,\, i}x_i / \lVert \sum_{i=1}^n w_{t,i}x_i \rVert, ~~t=1,\, \ldots,\, T
\end{align}
where $\norm{\cdot}$ is the standard norm in $\mathbb{R}^{d+1}$. Then $C_{t}$ is updated by $m_t(D,\, f)$. The extrinsic approach is advantageous in terms of the computational complexity compared to the intrinsic approach. Furthermore, the extrinsic way ensures the stationarity of the principal curves on hyperspheres $S^d$ for $d\ge 2$, which will be discussed in \Cref{sec:stationarity}.

\subsubsection{Intrinsic Approach}
From the intrinsic perspective, the weighted mean of data points can be obtained by the optimization 
\begin{align}
\label{Imean}
    m_t(D,\, f) = \argmin_{x}\sum_{i=1}^n w_{t,\, i} d^2_{Geo}(x,\, x_i),~~ t=1,\, \ldots,\, T,
\end{align}
and then each $C_t$ is updated by $m_t(D,\, f)$. The intrinsic mean exists uniquely if the points are in an open hemisphere of $S^d$, \textit{i.e.,} $\exists p\in S^d$ \textit{s.t.} $d_{Geo}(x_{i},\, p) < \frac{\pi}{2}$ for $1\le i\le n$ \cite{Buss}. Since the intrinsic mean cannot be obtained in a closed form, to solve \Cref{Imean}, algorithms based on tangent space approximation, such as \cite{Buss, Fletcher2004}, can be used.  

Before closing this section, as an alternative measure of the centrality of data, the geometric median can be considered to robustify the principal curves for a dataset that might contain outliers instead of the extrinsic or intrinsic mean. Median principal curves and their associated characteristics can be developed along with the same line of our procedure. Because of the limitation of space, this part is not discussed in the current paper.

\subsection{Algorithm}
\subsubsection{Initialization}
For a better estimation of principal curves, we initialize a principal curve as an exact principal circle on $d$-sphere $S^d$. The detailed descriptions of the circle and its algorithm were previously provided in \Cref{circle}.

\subsubsection{Spherical Principal Curves}
The proposed spherical principal curves on $S^d$ can be obtained by \Cref{alg:pc} below. 
\begin{algorithm}[H]
	\caption{~~Spherical Principal Curves}
    \label{alg:pc}
	\begin{algorithmic}
		\State Initialize curve $f = \{C_1,\, \ldots,\, C_{T}\}$.
		\State Parameterize the curve as $f(\lambda)$ by unit speed.
		\State Calculate $\lambda_f(x_i)$ in \Cref{eq:proj} for $i = 1,\, \ldots,\, n$. 
		\State Calculate errors $\delta(D,\, f) = \sum_{i=1}^n d^2_{Geo}\bigr(x_i,\,f\big(\lambda_f(x_i)\big)\bigr)$. 
		
		\While{ ($\Delta \delta(D,\, f) \ge \mbox{threshold}$) } 
		\State (Expectation) $C_t \leftarrow m_t(D,\, f)$ for $t=1,\, \ldots,\, T$. 
		\State Reparameterize the curve by unit speed.
		\State (Projection) Calculate $\lambda_f(x_i)$ for $i = 1,\, \ldots,\, n$.
		\State Calculate $\delta(D,\, f) = \sum_{i=1}^n d^2_{Geo}\bigr(x_i,\, f\big(\lambda_f(x_i)\big)\bigr)$. 
		\EndWhile  
	\end{algorithmic}
\end{algorithm}
Note that $d^2_{Geo}\bigr(x_i,\, f\big(\lambda_f(x_i)\big)\bigr)$ is calculated by (\ref{dgeo}). 
As far as Euclidean space is concerned as embedding space, the extrinsic approach is advantageous for computational efficiency \citep{Bhattacharya2012}. However, if the data points are not contained within local regions at the expectation step, the intrinsic method may have better performances than the extrinsic one. Furthermore, the intrinsic approach can be attractive because of its inherent metric.

\subsection{Stationarity of Principal Curves}\label{sec:stationarity}
For a random vector $X$ in $\mathbb{R}^d$, $d\in \mathbb{N}$, the stationarity of the principal curve of $X$ is given by \cite{Hastie} as 
\begin{equation}
\label{eq:stat}
\frac{\partial \mathbb{E}_X [d^2(X,f+\epsilon g)]}{\partial\epsilon} \biggr|_{\epsilon = 0} = 0,  
\end{equation}
where $f$ and $g$ are smooth curves in $\mathbb{R}^d$ satisfying $\norm{g}\le 1$ and $\norm{g^\prime}\le 1 $, and $d(X,\, f)$ denotes the (Euclidean) distance from $X$ to the curve $f$.

However, since spheres are not vector spaces such as $\mathbb{R}^d$, additions are not directly defined on spheres. Thus, it is necessary to redefine some concepts, such as addition and perturbation, in order to extend the properties of the principal curves in Euclidean space to spheres. To this end, we conversely consider $f+g$ instead of $g$. Specifically, let $f$ and $f+g$ be smooth curves on $d$-sphere parameterized with $\lambda \in [0,\, 1]$. Then, we define $f+\epsilon g$ in a pointwise sense as follows. 

\begin{definition}
For $a, b\in S^d$ and $\epsilon\in [0,\, 1]$, div($a,\, b,\, \epsilon$) is a set of points on geodesics between $a$ and $b$ satisfying $\forall c\in div(a,\, b,\, \epsilon$), $d_{Geo}(a,\, c) = \epsilon d_{Geo}(a,\, b)$ and $c$ is on a geodesic between $a$ and $b$.
\end{definition}
Note that if $d_{Geo}(a,\, b)< \pi$, then the geodesic between $a$ and $b$ on $S^d$ is unique. In this case, $div(a,\, b,\, \epsilon)$ is a single point set and $div(a,\, b,\, -\epsilon)$ can be defined as a reflection of $div(a,\, b,\, \epsilon)$ with respect to $a$.

\begin{definition}
Let $f$ and $f+g$ be smooth curves on $S^d$ parameterized with $\lambda \in [0,\, 1]$ satisfying $\norm{g} < \pi$, where $\norm{g}:= \max_{\lambda \in [0,\, 1]} d_{Geo}\big(f(\lambda),\,  (f+g)(\lambda)\bigr)$. Then, for $\epsilon\in [-1,\, 1]$, $f+\epsilon g$ is a curve on $S^d$, where $(f+\epsilon g)(\lambda) = div(f(\lambda),\, (f+g)(\lambda),\, \epsilon)$,\, $\forall \lambda \in [0,\, 1]$.
\end{definition}
Note that $f+\epsilon g$ is a smooth curve on $S^d$. For a detailed proof, refer to the proposition 1 in \Cref{proof:stationarity}. Let $X$ be a random vector on $S^d$ that has a probability density function. Then, we call $f$ as an \textit{extrinsic principal curve} of $X$, if $f$ is self-consistent with $X$ in the embedding space as
\[
\pi\big(\mathbb{E} [\xi(X) \ | \ \lambda_f(X)=\lambda]\big) = f(\lambda)\ ~\mbox{for}~\mbox{\textit{a.e.}}~\lambda,
\]
where $\xi: S^d \rightarrow \mathbb{R}^{d+1}$ is the canonical embedding and $\pi: \mathbb{R}^{d+1} \setminus \{0\} \to S^d$ by $X \to \frac{X}{\norm{X}}$ is the standard projection (retraction) from $\mathbb{R}^{d+1}$ to $S^d$. In analogy to \Cref{eq:stat}, we provide the following theorem on spheres. Note that $\cdot$ represents the standard inner product on $\mathbb{R}^{d+1}$ and $d_{Geo}(X,\, f+\epsilon g)$ denotes the geodesic distance from $X$ to the curve $f+\epsilon g$.

\begin{definition}
 $\norm{g'}:=\max_{\lambda \in [0,\, 1]} \norm{g'(\lambda)}$, where $\norm{g'(\lambda)} = \max_{\epsilon \in [0,\, 1]} \norm{\frac{\partial^2 (f + \epsilon g)(\lambda)}{\partial \lambda \partial \epsilon}}$.
\end{definition}

\begin{theorem}
Let $f$, $f+g: [0,\, 1] \to S^{d}$, $d\ge 2$ be smooth curves satisfying $\norm{g} < \pi$ and $\norm{g'}\le 1$. Let $X$ be a random vector on $S^2$ or a random vector on $S^d$, $d\ge 3$ with $X\in C(\zeta)$, where $C(\zeta) := \{x \in S^d \ | \ |f''(\lambda_{f}(x))\cdot x| > \zeta \}$ for a small $\zeta > 0$. Then $f$ is an extrinsic principal curve of $X$ if and only if
    \begin{equation}
    \label{eq:thm1}
        \frac{\partial \mathbb{E}_X[\cos\big(d_{Geo}(X,\, f+\epsilon g)\big)]}{\partial\epsilon} \biggr|_{\epsilon = 0} = 0. 
    \end{equation}
\end{theorem}
\begin{proof}
    See \Cref{proof:stationarity}.
\end{proof}
Note that since $2-2\cos x \approx x^2$ for small $x$, \Cref{eq:thm1} can be interpreted as an analogy of  \Cref{eq:stat}.

We further consider the intrinsic perspective of the stationarity. We define a curve $f$ as an \textit{intrinsic principal curve} of $X$ if the intrinsic mean of $X$ conditioned on $\lambda_{f}(X)=\lambda$ is equal to $f(\lambda)$ for \textit{a.e.} $\lambda$, 
\[
\mathbb{E}_{int}[X \ | \ \lambda_f(X)=\lambda] = f(\lambda)\ ~\mbox{for}~\mbox{\textit{a.e.}}~\lambda,
\]
where $\mathbb{E}_{int}[\cdot]$ represents an intrinsic mean of a random variable on $S^d$.

Note that the intrinsic mean of a random variable $Y$ on $S^d$ is unique if $d_{Geo}(Y,\, p)< \frac{\pi}{2}$ \textit{a.s.} for $\exists p\in S^d$, \textit{i.e.,} the support of $Y$ is in an open hemisphere \citep{Pennec}. We verify that the intrinsic principal curves on $S^2$ satisfy the stationarity. 
\begin{theorem}
Let $f$, $f+g: [0,\, 1] \to S^2$ be smooth curves satisfying $\norm{g} < \pi$ and $\norm{g'}\le 1$. Let $X$ be a random vector on $S^2$ with $X \in B(\zeta)$, where $B(\zeta) := \{x\in S^2 \ | \ |f''(\lambda_{f}(x))\cdot x| > \zeta \}$ for a small $\zeta > 0$. Then, $f$ is an intrinsic principal curve of $X$ if and only if   
\begin{equation}
	\label{eq:thm2}
	\frac{\partial \mathbb{E}_X[d^{2}_{Geo}(X,\, f+\epsilon g)]}{\partial\epsilon} \biggr|_{\epsilon = 0} = 0. 
	\end{equation}
\end{theorem}
\begin{proof}
See \Cref{proof:stationarity}.
\end{proof}
The constraints $C(\zeta)$ and $B(\zeta)$ in Theorems 1 and 2 are required to ensure the differentiation of the projection index $\lambda_{f+\epsilon g}(X)$ with respect to $\epsilon$. Note that the constraints are almost negligible by setting $\zeta$ infinitesimally small; see Lemmas 4 and 6 in \Cref{proof:stationarity} for details.

We finally remark that the stationarity of the principal curves in Euclidean space provides a rationale for the principal curves in \cite{Hastie} that is a nonlinear generalization of the linear principal component. Following the same line, the above stationarity results provide a theoretical justification that the proposed approaches directly generalize the principal curves in \cite{Hastie} from Euclidean space to spheres. In the intrinsic approach, the case of $S^d$ with $d\ge 3$ remains a challenge.

\section{Numerical Experiments}\label{numerical:experiment}
This section conducts numerical experiments with real data analysis and simulated examples to assess the practical performance of the proposed methods. 
The experiment can be reproduced at \url{https://github.com/Janghyun1230/Spherical-Principal-Curve}. Moreover, we provide R package, \textbf{spherepc} at \url{https://cran.r-project.org/package=spherepc}, which implements the spherical principal curves for a variety of datasets lying on $S^2$.

\subsection{Real Data Analysis}

\subsubsection{Earthquake Data on $S^2$}

\begin{figure*}[!ht]
	\centering
	\includegraphics[scale=0.125]{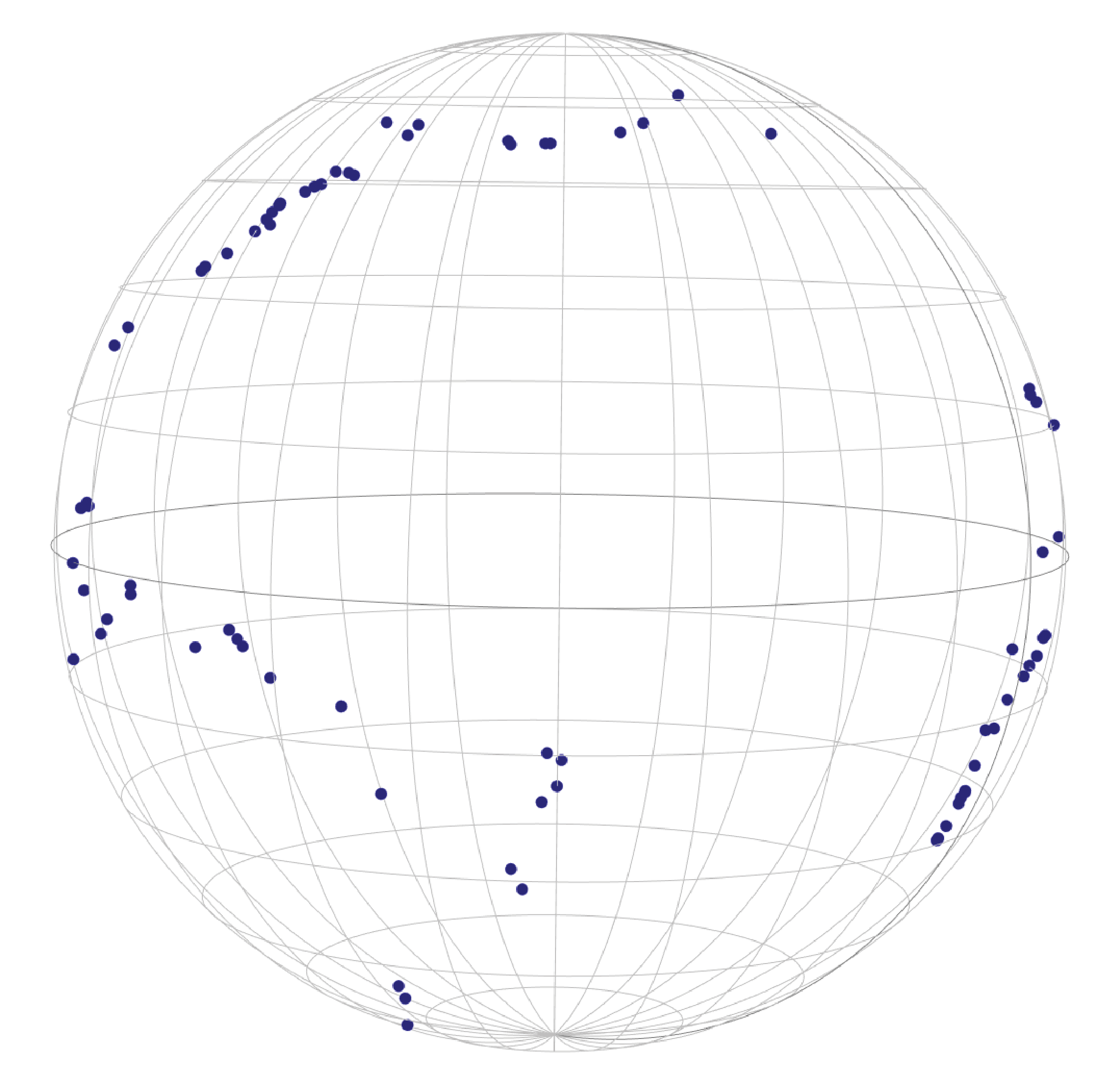}
	\includegraphics[scale=0.125]{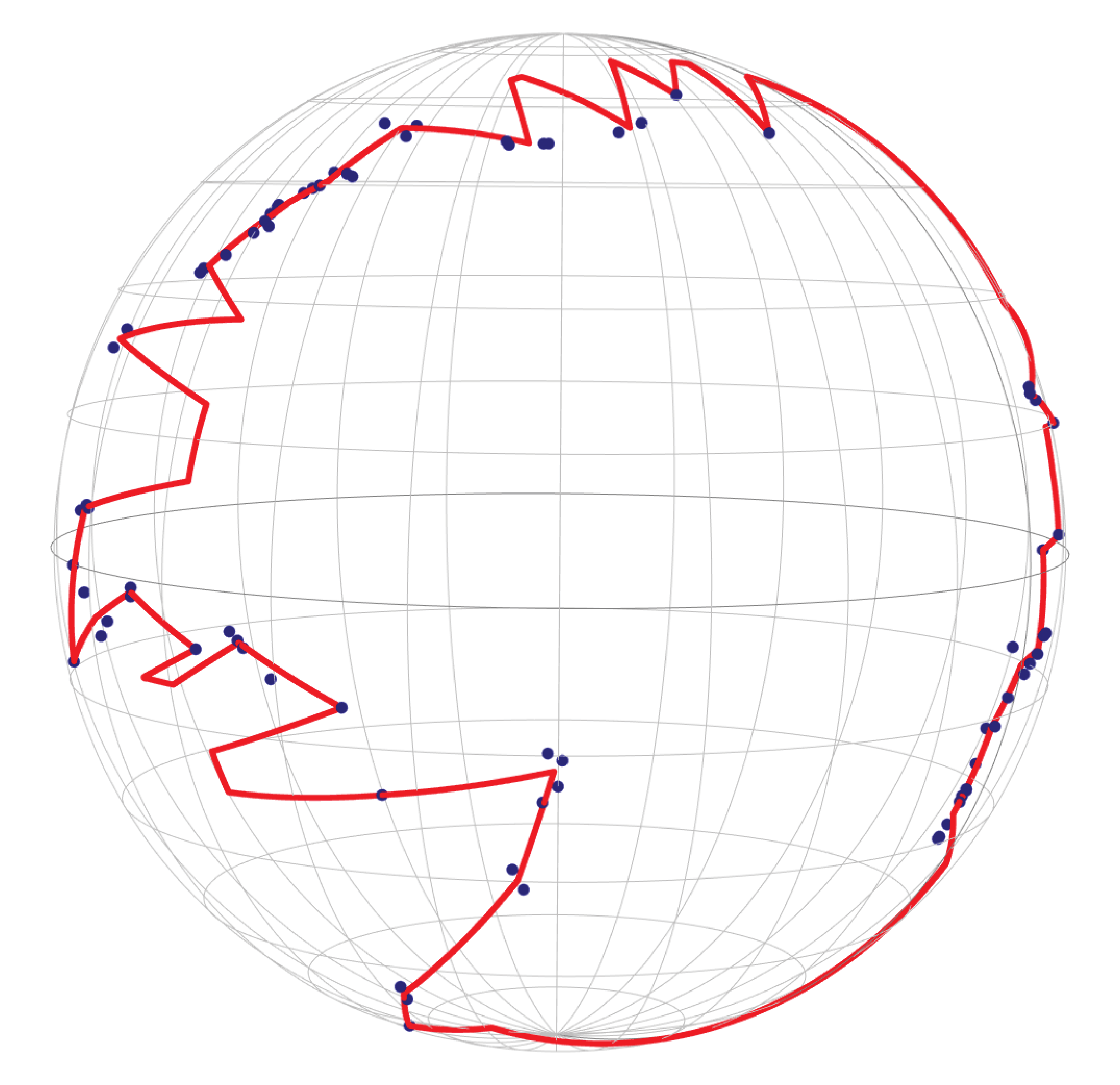}
	\includegraphics[scale=0.125]{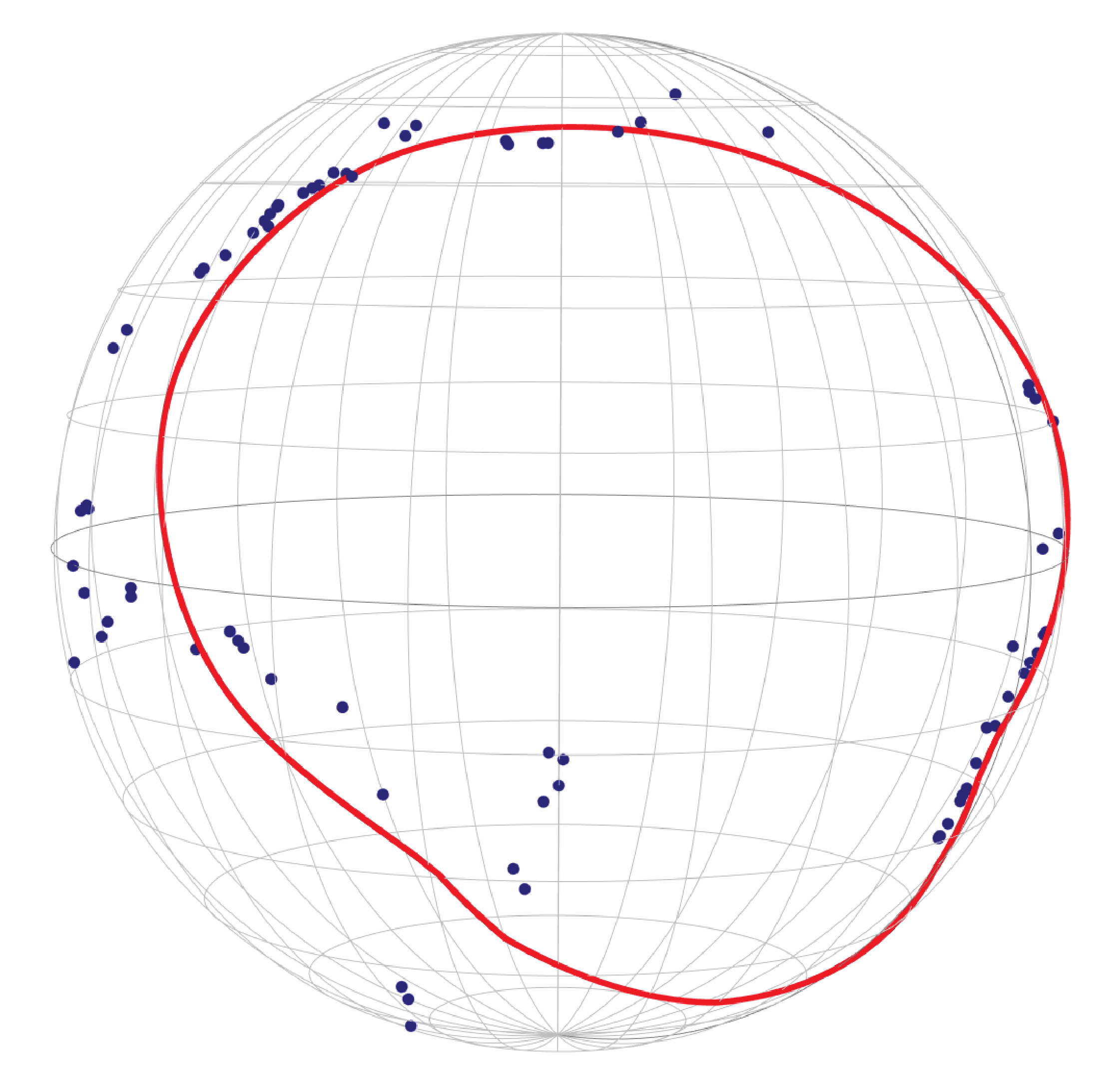}
	\caption{The distribution of earthquake data (left). The proposed extrinsic principal curves of $T=500$ with $q=0.01$ (middle) and $q=0.2$ (right). Blue points represent the observations and red lines are the fitted curves.}
	\label{fig:earth}
\end{figure*}

We consider earthquake data from the U.S. Geological Survey (\url{https://earthquake.usgs.gov/earthquakes/map/}) in \Cref{fig:earth} that represent the distribution of significant earthquakes (8+ Mb magnitude) around the Pacific Ocean since 1900. As shown in the figure, 77 observations are distributed in the vicinity of the borders between the Pacific, Eurasian, and Nazca plates. Since the plates are gradually moving towards different directions, recognizing the unrevealed patterns of borders provides essential information about seismological events such as earthquakes and volcanoes \citep{Mardia1977, Biau}. In the following experiment, we utilize the spherical principal curves to recover the plates' borders by extracting curvilinear features of the observations.

We have implemented the proposed principal curves connected by $T=500$, with various values of hyperparameter $q$ that is the bandwidth of kernel in the expectation step. \Cref{fig:earth} shows the results with $q=0.01$ and 0.2. We observe that a small $q$ produces a wiggly and overfitted curve. It is noteworthy that the choice of $q$ affects the quality of the fitted curve. \cite{Duchamp} proved that principal curves are always the saddle point of the expectation of the squared distance from a particular random variable, pointing out that cross-validation is not reliable for the model selection of principal curves, \textit{i.e.,} determination of $q$. \cite{Kegl2000} defined principal curves that minimize reconstruction errors in the constraint of the curve length, but used a heuristic way to determine the corresponding hyperparameter, the length of the curves. In the current study, the value of $q$ is selected by visual inspection through all our experiments. An objective way to select $q$ is left for future research. 

\begin{figure}[]
	\centering
	\includegraphics[scale=0.15]{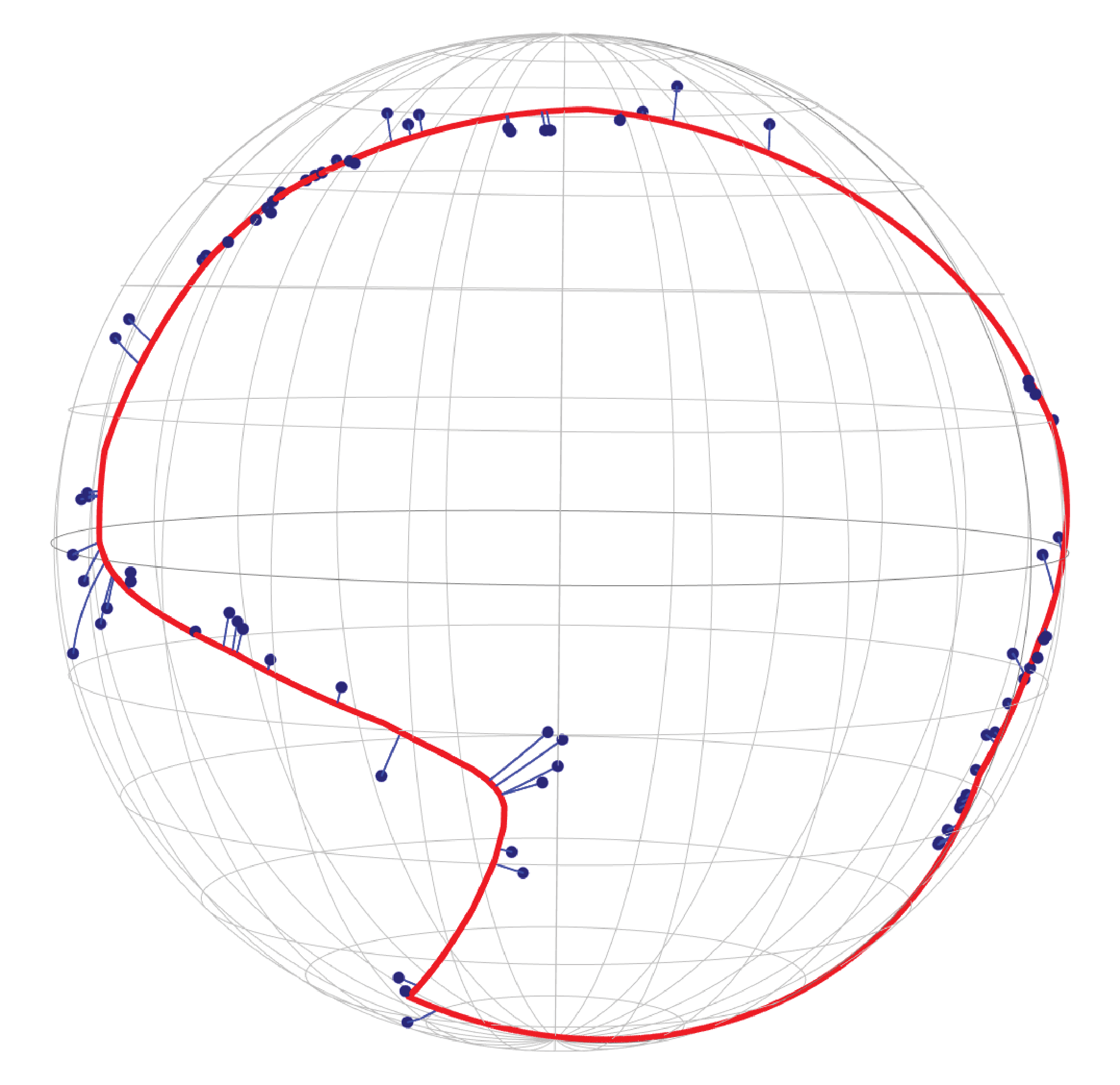}
	\hspace{0cm}
	\includegraphics[scale=0.15]{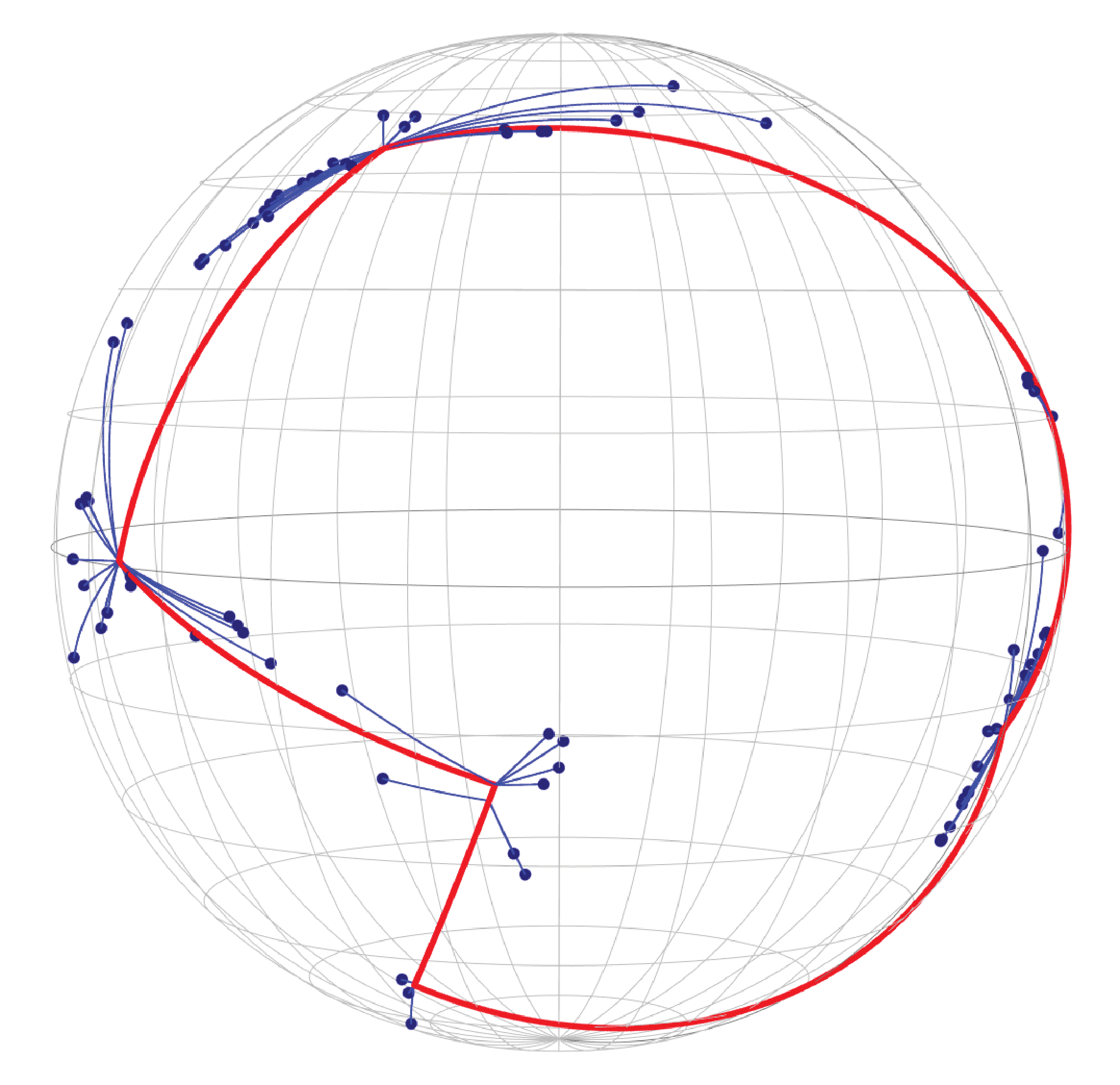}
	\vspace{-1mm}
	\caption{Projection results by the proposed extrinsic method (left) and Hauberg's method (right) with $T=77$ and $q=0.1$.} 
	\label{fig:compare} 
\end{figure}

As one can see, the proposed extrinsic curve represents a given data as a continuous curve, while the Hauberg method projects several local data at one point.

We further compare the proposed extrinsic principal curves with the method of \cite{Hauberg}. \Cref{fig:compare} shows both results with $q=0.1$, where the purple lines represent the fitted curves, and the blue lines represent the projections from the data to the curve. The proposed extrinsic principal curve continuously represents the given data on the curve, while the method of \cite{Hauberg} projects several local points to a single location. The comparison is further summarized in Table 1. As a result, the number of distinct projections (\# proj) by our method is much larger than that of Hauberg's method. It implies that the proposed principal curve continuously represents the data, whereas the method of Hauberg tends to cluster the data. We also measure a reconstruction error (RE) defined as \small{$\sum_{i=1}^n d^2_{Geo}\big(x_i,\hat{f}\big(\lambda_{\hat{f}}(x_i)\big)\big)$} \normalsize with observations $\{x_i\}_{i=1}^n$ and fitted values \small{$\{\hat{f}\big(\lambda_{\hat{f}}(x_i)\big)\}_{i=1}^n$}\normalsize. As listed in Table 1, our method outperforms Hauberg's method in terms of the reconstruction error. 

\begin{table}[!ht]
    \centering
    \caption{The values of RE and \# proj by the proposed methods and Hauberg's method on the earthquake data}
		\begin{tabular}{c||c|c|ccc}
		    \hline 
			\multicolumn{3}{c|}{}                                                        & \text{Extrinsic} & \text{Intrinsic} & \text{Hauberg} \rule[-1.0ex]{0pt}{3.2ex} \\ 
			\hline \hline
			& \multirow{2}{*}{$q=0.2$} & RE & 2.662     & 4.391     & 12.067   
			\rule{0pt}{2.2ex}  \\
			\multirow{7}{*}{$T=77$}  &  
			& \# proj & 74/77     & 72/77     & 22/77  \rule[-1.0ex]{0pt}{0pt}\\ 
			\cline{2-6} 
			& \multirow{2}{*}{$q=0.1$} &  RE & 0.463     & 0.467     & 4.920  \rule{0pt}{2.4ex}  \\
			&                         & \# proj & 76/77     & 76/77     & 9/77   \rule[-1.0ex]{0pt}{0pt} \\ \cline{2-6} 
			& \multirow{2}{*}{$q=0.05$} &  RE  & 0.359     & 0.359     & 1.313  \rule{0pt}{2.4ex} \\
			&                         & \# proj & 74/77     & 73/77     & 16/77   \rule[-1.0ex]{0pt}{0pt} \\ \cline{2-6} 
			& \multirow{2}{*}{$q=0.01$} &  RE & 0.061     & 0.061     & 0.227  \rule{0pt}{2.4ex} \\
			&                         & \# proj& 75/77     & 75/77     & 27/77   \rule[-1.0ex]{0pt}{0pt} \\ \cline{1-6} 
			& \multirow{2}{*}{$q=0.2$} &  RE & 2.193     & 3.460     & 11.300  \rule{0pt}{2.4ex} \\
			\multirow{7}{*}{$T=500$} &                         
			& \# proj & 75/77     & 72/77     & 30/77   \rule[-1.0ex]{0pt}{0pt} \\ 
			\cline{2-6} 
			& \multirow{2}{*}{$q=0.1$} &  RE  & 0.715     & 0.732     & 3.903 \rule{0pt}{2.4ex}  \\
			&                         & \# proj & 75/77     & 74/77     & 18/77   \rule[-1.0ex]{0pt}{0pt} \\ \cline{2-6} 
			& \multirow{2}{*}{$q=0.05$} &  RE  & 0.298     & 0.200     & 0.963  \rule{0pt}{2.4ex} \\
			&                         & \# proj & 75/77     & 75/77     & 27/77   \rule[-1.0ex]{0pt}{0pt} \\ \cline{2-6} 
			& \multirow{2}{*}{$q=0.01$} &  RE  & 0.036     & 0.036     & 0.121  \rule{0pt}{2.4ex} \\
			&                         & \# proj & 75/77     & 75/77     & 37/77  \rule[-0.9ex]{0pt}{0pt} \\ 
			\hline
		\end{tabular}
	\label{table1}
\end{table}

\subsubsection{Motion Capture Data on $S^2$}
We now consider a benchmark data on $S^2$, motion capture data of a person walking in a circular pattern \citep{Ionescu2011, Ionescu2014, Hauberg, Mallasto}. The data represent the orientation of the person's left \textit{thigh bone} and naturally lie on $S^2$. There are 338 data points in the data set that are periodic.

 \Cref{fig:motion} shows both results with $q= 0.03,\, 0.05$, where the red and yellow lines represent the fitted curves, and the blue lines represent the projections from the data to the curves. The proposed extrinsic principal curve continuously represents the given data on the curve, while the method of Hauberg projects several local points to a single location. Furthermore, \Cref{table2} lists the quantitative results of the proposed methods and the method of \cite{Hauberg}. As listed, the proposed methods outperform Hauberg's method in terms of the reconstruction error and represent the data more precisely.

\begin{figure}[ht]
	\centering
	\includegraphics[scale=0.13]{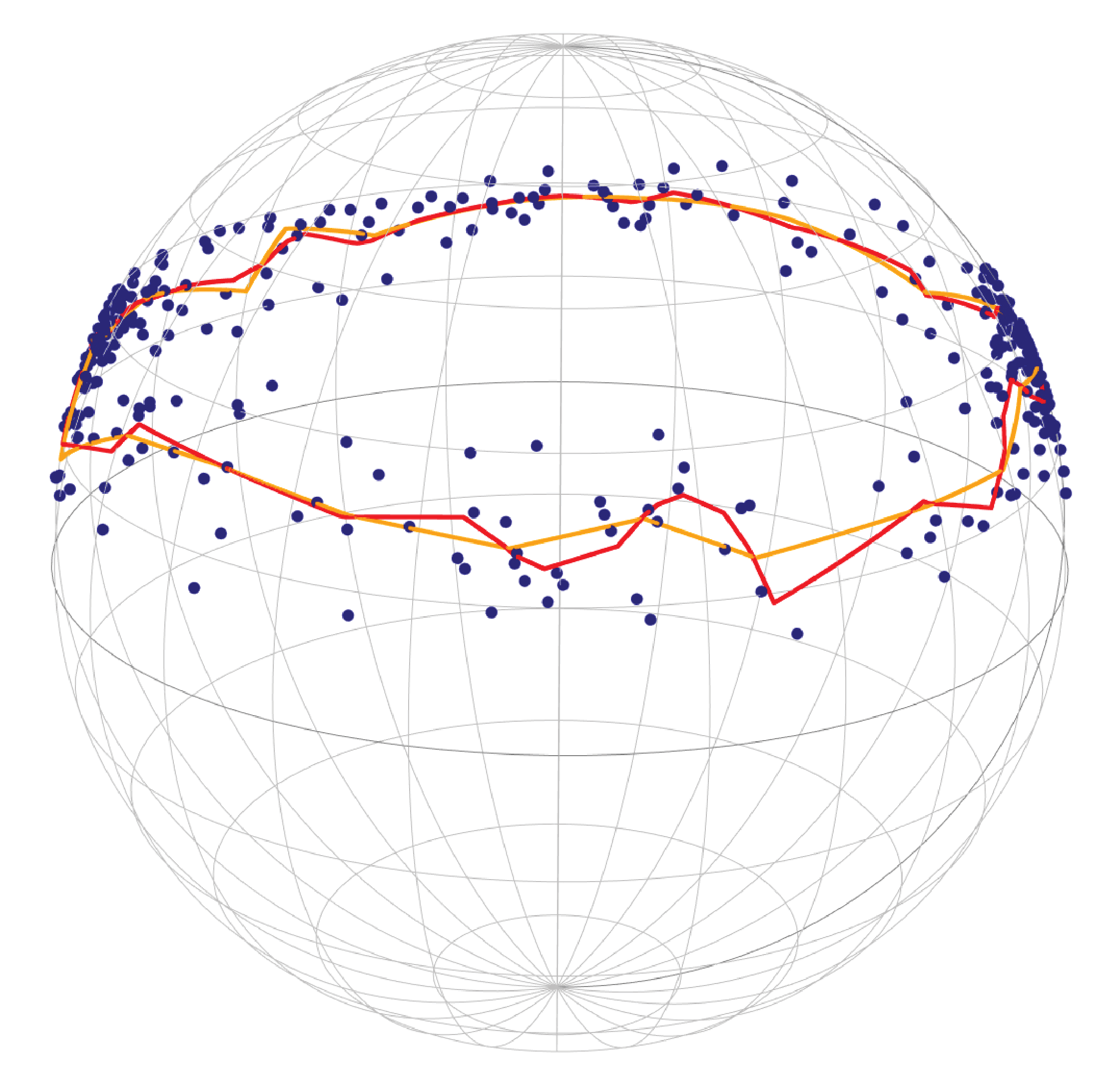}
	\includegraphics[scale=0.13]{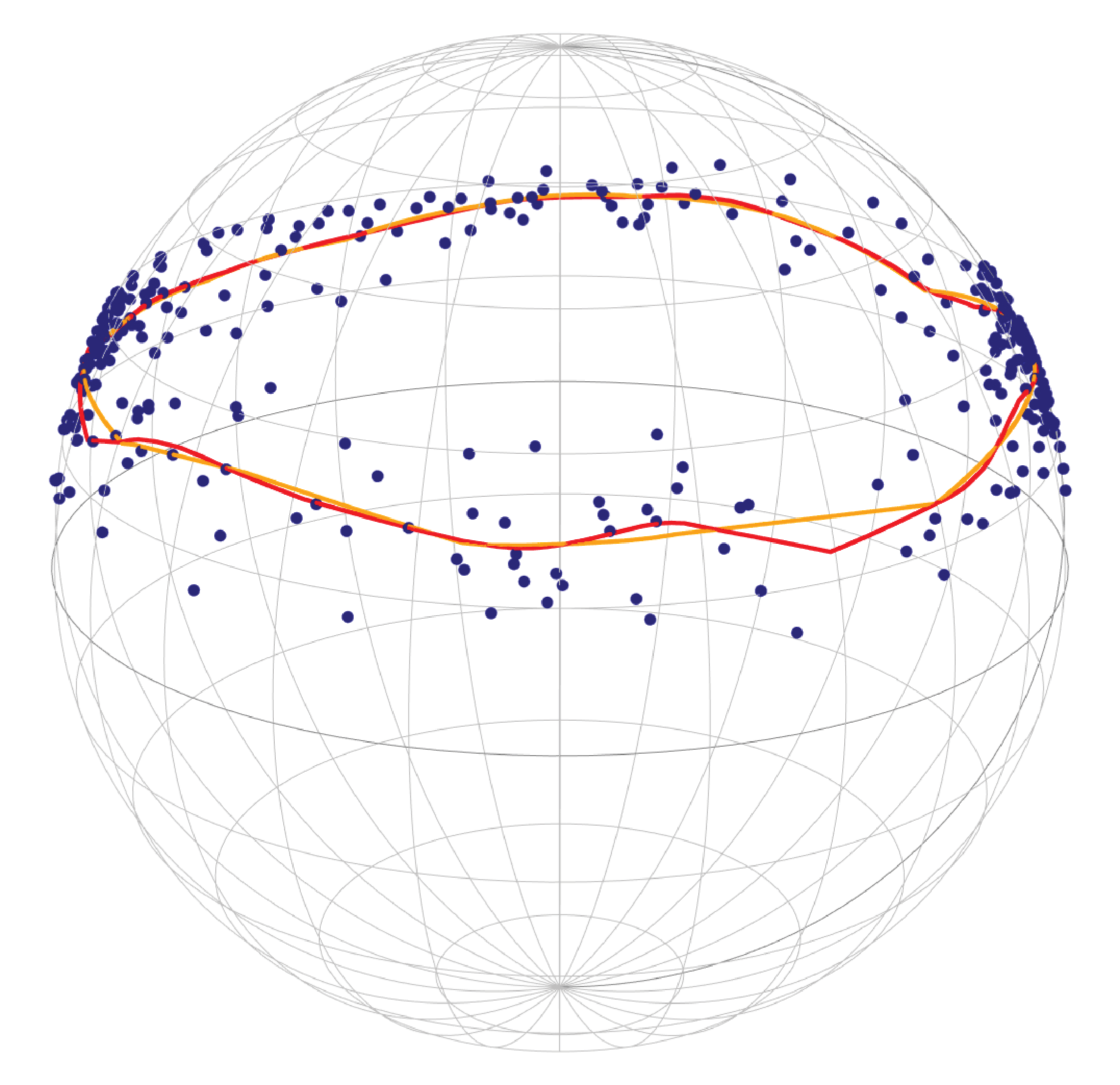}
	\includegraphics[scale=0.13]{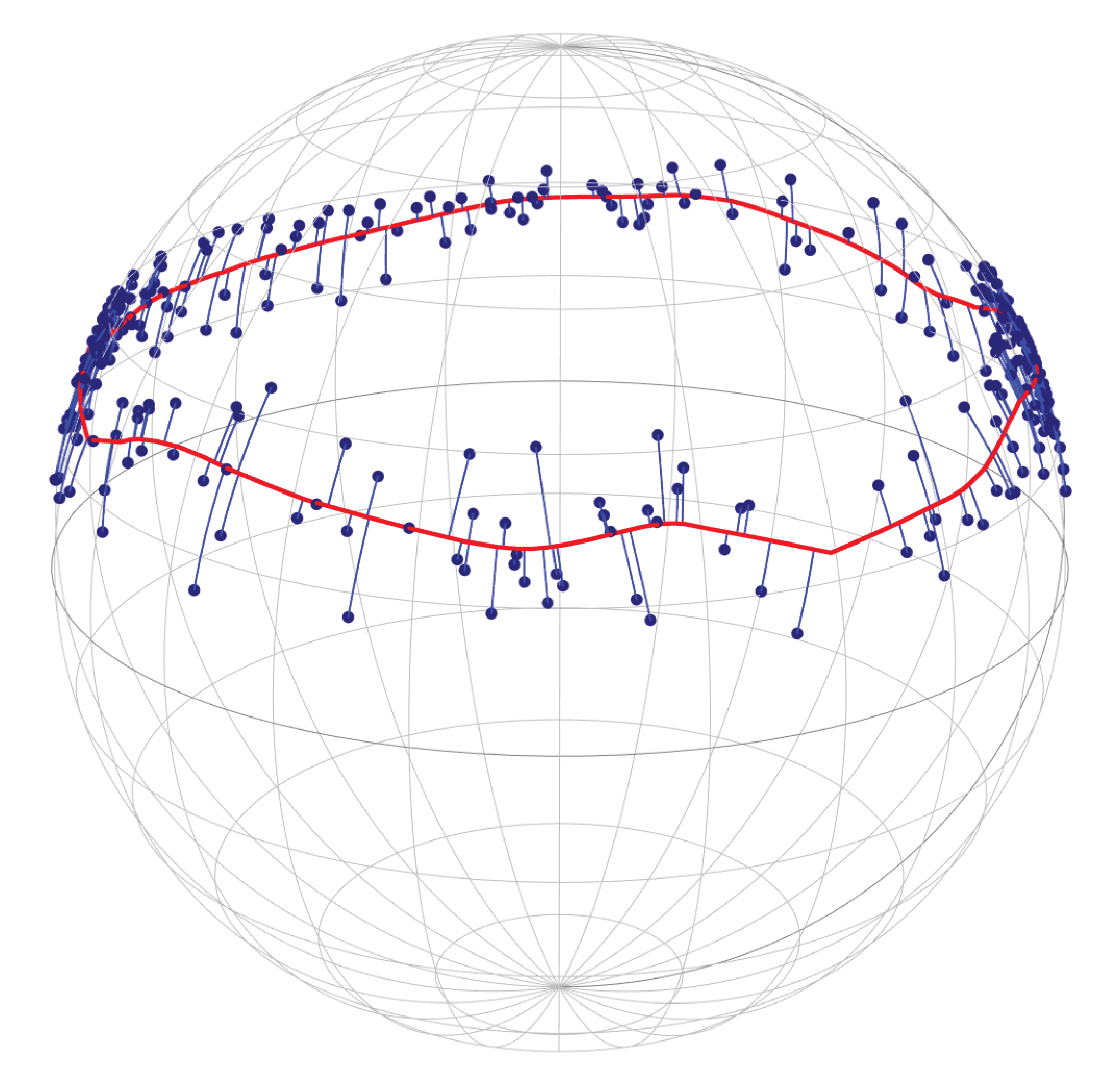}
	\includegraphics[scale=0.13]{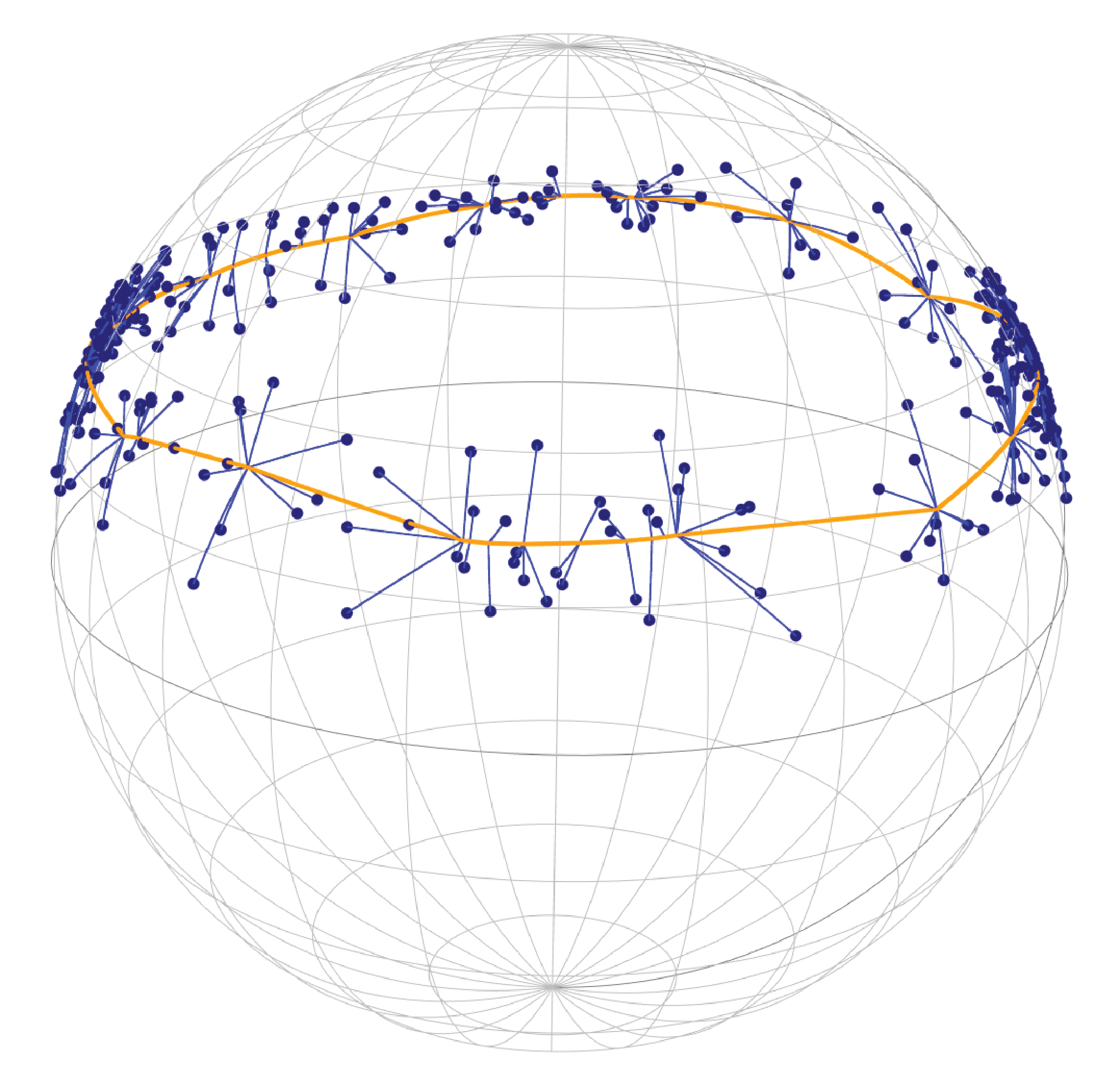}
	\vspace{-1mm}
	\caption{Results by the proposed extrinsic method (red) and Hauberg's method (yellow) with $T=100$. From top left to bottom right, results with $q=0.03$ and $q=0.05$, projection results by the two methods (blue) with $q=0.05$.}
	\label{fig:motion} 
\end{figure}

\begin{table}[!ht]
    \centering
    \caption{The values of RE and \# proj by the proposed methods and Hauberg's method on the motion capture data}
		\begin{tabular}{c||c|c|ccc}
			\hline
			\multicolumn{3}{c|}{}                      & \text{Extrinsic} & \text{Intrinsic} & \text{Hauberg}  \rule[-1.0ex]{0pt}{3.2ex} \\ 
			\hline \hline
			& \multirow{2}{*}{$q=0.05$} &  RE &  2.502  & 2.504 & 2.534 
			\rule{0pt}{2.2ex}  \\
			\multirow{5}{*}{$T=500$}  &  & \# proj & 336/338 & 337/338  & 223/338
			\rule[-1.0ex]{0pt}{0pt} \\ 
			\cline{2-6} 
			& \multirow{2}{*}{$q=0.03$} &  RE  & 1.741 & 1.741 & 2.637 \rule{0pt}{2.4ex}  \\
			&                         & \# proj & 332/338  & 333/338  & 119/338   
			\rule[-1.0ex]{0pt}{0pt} \\ 
			\cline{2-6} 
			& \multirow{2}{*}{$q=0.01$} &  RE  & 0.669 & 0.669 &  1.253 \rule{0pt}{2.4ex} \\
			&                         & \# proj & 315/338  & 317/338  & 92/338   
            \rule[-0.9ex]{0pt}{0pt} \\ 
			\hline
		\end{tabular}
	\label{table2}
\end{table}

\subsection{Simulation Study}
\subsubsection{Simulation on $S^2$}\label{simul:S^2}
We consider two types of functions on the unit sphere with spherical coordinates $(r=1,\, \theta,\, \phi)$, where $\theta$ is the azimuthal angle and $\phi$ is the polar angle: ~(Circle) it is formed of $(r=1,\, \theta,\, \phi)$ with $0\le \theta < 2\pi$ and $\phi=\pi/4$. ~(Wave) it is defined as $(r=1,\theta,\phi)$ with $0\le \theta < 2\pi$ and $\phi=\alpha\sin(\theta f)+\pi/2$, where the frequency $f=4$ and the amplitude $\alpha=1/3$.

\begin{figure*}[!ht]
	\centering
	\includegraphics[scale=0.13]{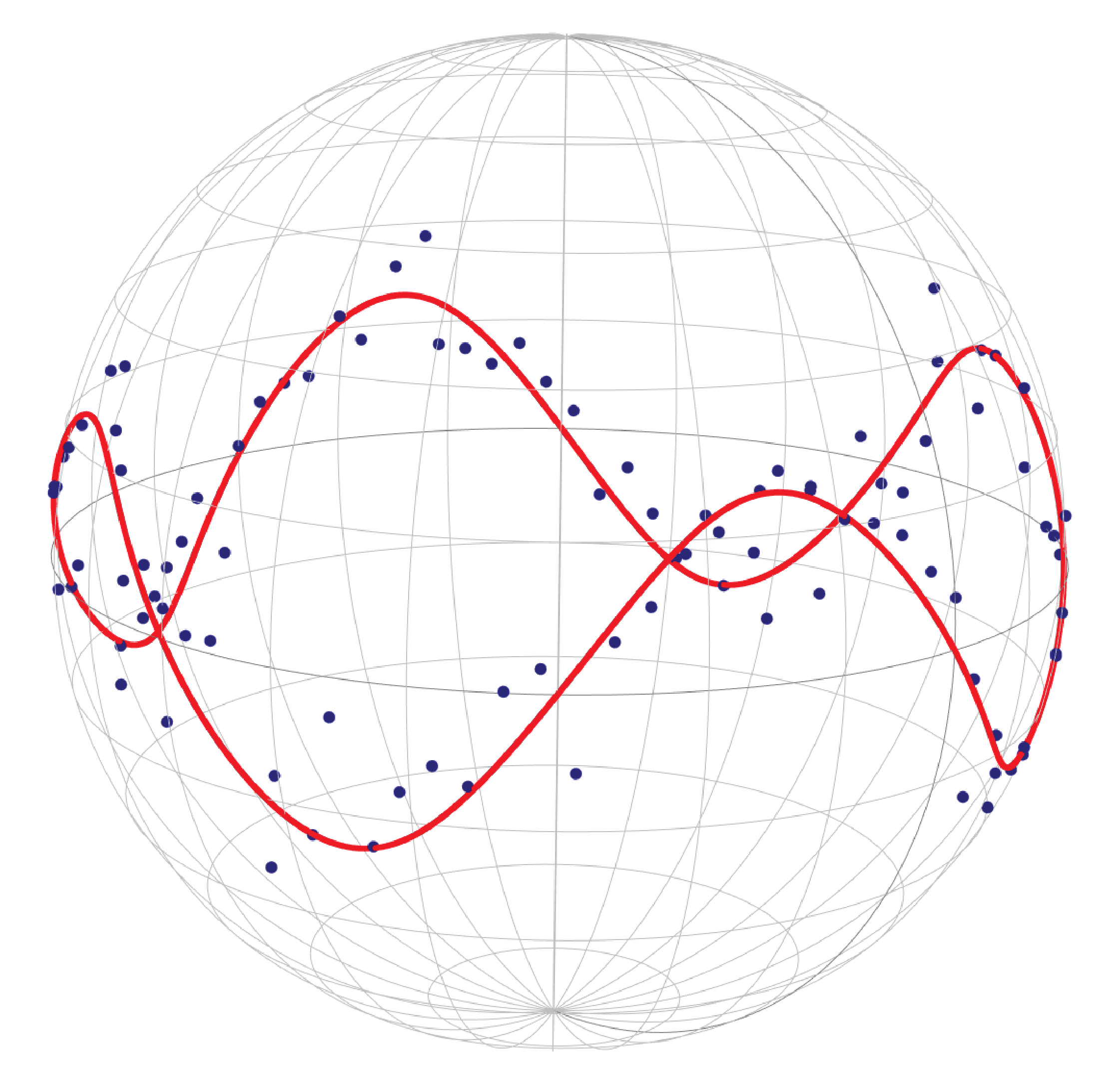}
	\hspace{0.0cm}
	\includegraphics[scale=0.13]{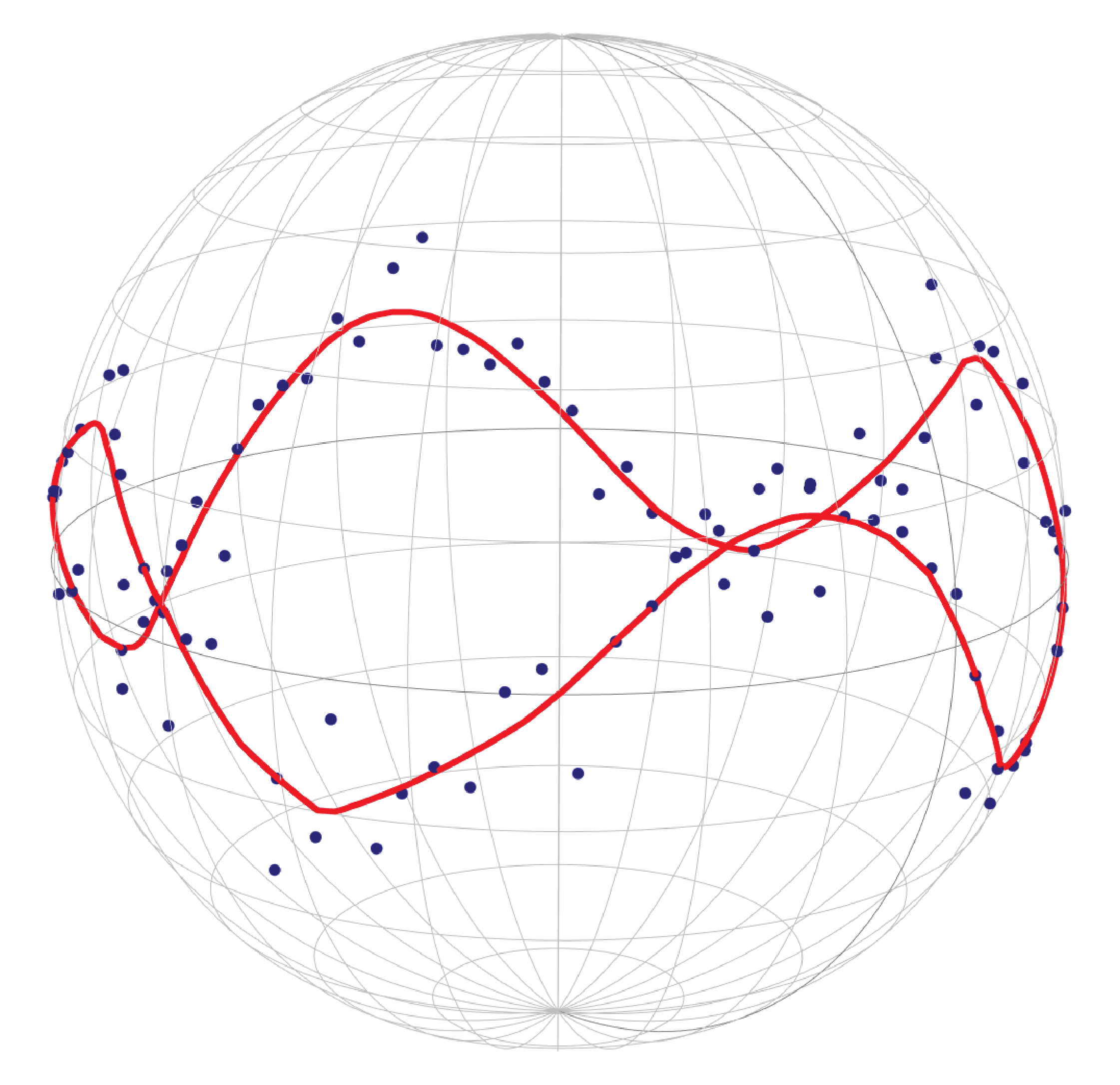}
	\hspace{0.0cm}
	\includegraphics[scale=0.13]{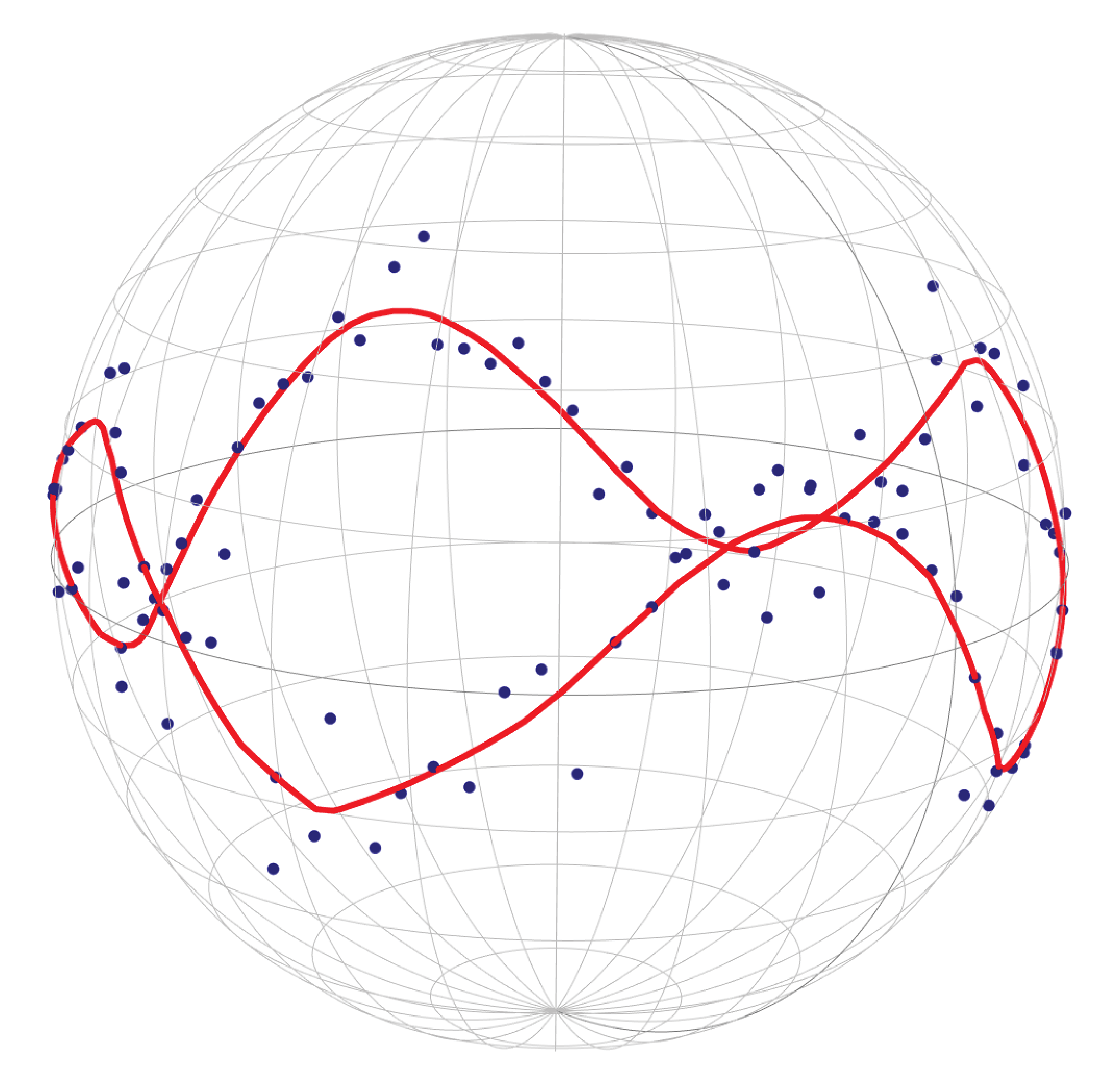}
	\hspace{0.0cm}
	\includegraphics[scale=0.13]{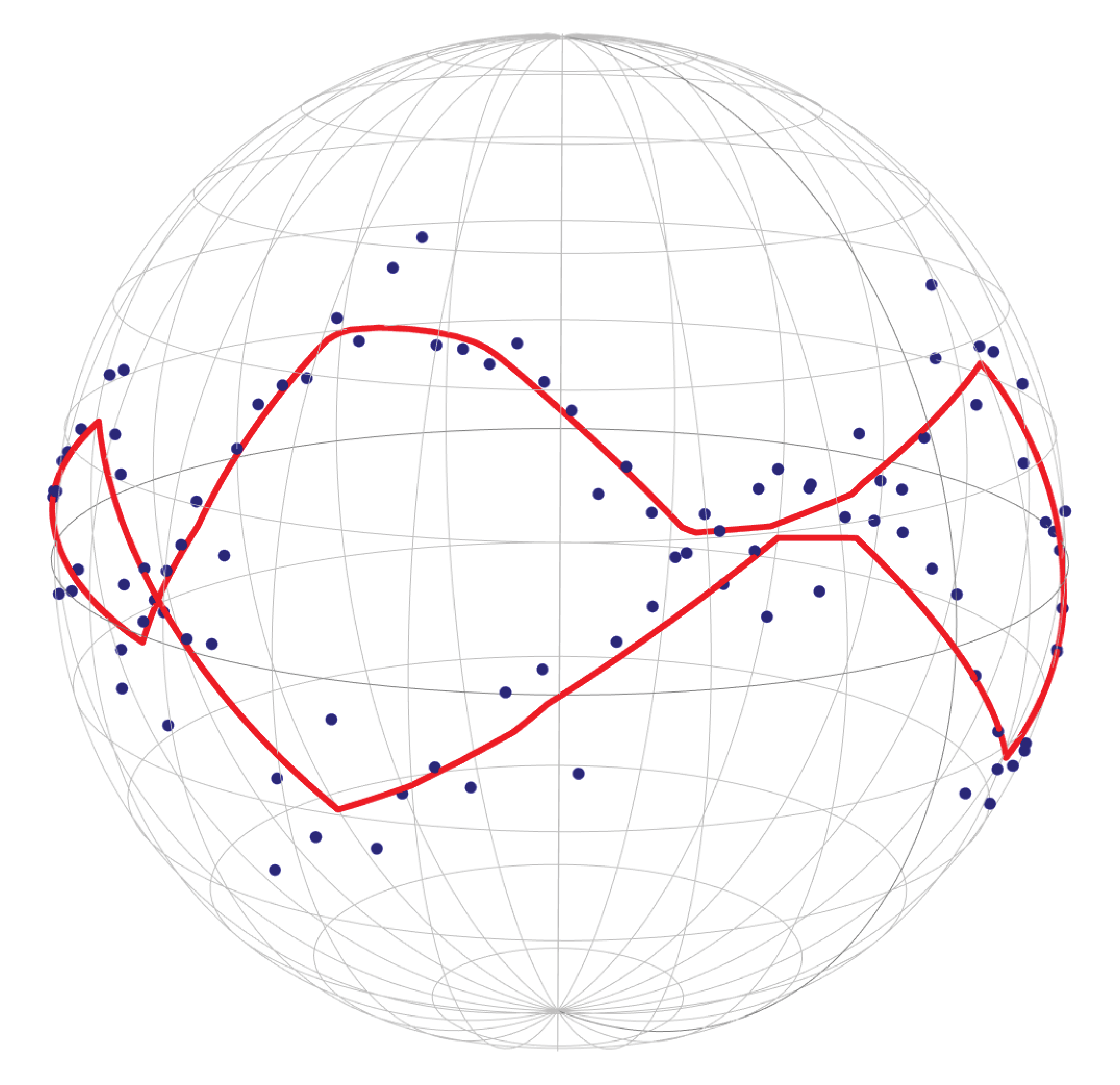}
	\vspace{-1mm}
	\caption{From top left to bottom right: True waveform and noisy data (blue dots), the extrinsic principal curve, the intrinsic principal curve, and the curve by Hauberg's method with $T=100$.} 
	\label{fig:wave}
\end{figure*}

For each type of functions, we generate $n=100$ data points by sampling $\theta$ uniformly in $[0,\, 2\pi)$ and adding Gaussian noises sampled from $N(0,\, \sigma^2)$ to $\phi$. \Cref{fig:wave} shows the results on the waveform data with $T=500$ and $q=0.05$.  Both extrinsic and intrinsic principal curves extract the true waveform effectively, while Hauberg's approach yields a rather sharp curve. In \Cref{influence:Tq}, we provide additional visual results with various parameter settings.

\begin{table*}[ht]
		\caption{Averages of reconstruction errors and their standard deviations in the parentheses by each method}
	\resizebox{1.0\textwidth}{!}{
		\begin{tabular}{c||c|c|cccccc}
			\hline
			\multirow{3}{*}{}      & \multirow{3}{*}{True form} & \multirow{3}{*}{Method} & \multicolumn{6}{c}{Noise level}\\
			&                            &                         & \multicolumn{3}{c}{$\sigma=0.07$}          & \multicolumn{3}{c}{$\sigma=0.1$}                   \\ \cline{4-9}
			&                            &                         & $q=0.05$         & $q=0.03$        & \multicolumn{1}{c|}{$q=0.01$}        & $q=0.05$        & $q=0.03$         & $q=0.01$   \rule{0pt}{2.2ex}   \\ \hline \hline
			\multirow{4}{*}{$T=100$} & \multirow{2}{*}{Circle}    
			                             & Proposed     & 0.093 (0.026)  & 0.12 (0.027) & \multicolumn{1}{c|}{0.095 (0.013)} & 0.201 (0.048) & 0.216 (0.046)  & 0.137 (0.025) \rule{0pt}{2.2ex} \\
			&                            & Hauberg                 & 0.117 (0.073) & 0.408 (0.149) & \multicolumn{1}{c|}{0.298 (0.038)}  & 0.370 (0.205) & 0.74 (0.208)   & 0.494 (0.063) \rule[-0.9ex]{0pt}{0pt} \\ \cline{2-9}
			& \multirow{2}{*}{Wave}      
			                             & Proposed     & 0.71 (0.114)   & 0.329 (0.097) & \multicolumn{1}{c|}{0.084 (0.023)} & 0.673 (0.150) & 0.346 (0.113)  & 0.124 (0.038) \rule{0pt}{2.2ex} \\
			&                            & Hauberg                 & 2.444 (0.059)  & 2.158 (0.155) & \multicolumn{1}{c|}{0.568 (0.055)} & 2.544 (0.118) & 2.103 (0.563)  & 0.796 (0.094) \rule[-0.9ex]{0pt}{0pt}  \\ \hline
			\multirow{4}{*}{$T=500$} & \multirow{2}{*}{Circle}    
			                             & Proposed     & 0.088 (0.026)  & 0.118 (0.023) & \multicolumn{1}{c|}{0.091 (0.018)} & 0.21 (0.050) & 0.207 (0.043)  & 0.129 (0.018) \rule{0pt}{2.2ex} \\
			&                            & Hauberg                 & 0.089 (0.027)  & 0.205 (0.079) & \multicolumn{1}{c|}{0.269 (0.034)} & 0.233 (0.087) & 0.453 (0.177)  & 0.397 (0.079) \rule[-0.9ex]{0pt}{0pt}  \\ \cline{2-9}
			& \multirow{2}{*}{Wave}      
			                             & Proposed     & 0.535 (0.065)  & 0.239 (0.056) & \multicolumn{1}{c|}{0.072 (0.020)} & 0.574 (0.094) & 0.237 (0.082)  & 0.110 (0.031) \rule{0pt}{2.2ex} \\
			&                            & Hauberg                 & 2.006 (0.697)  & 1.831 (0.146) & \multicolumn{1}{c|}{0.529 (0.043)} & 1.906 (0.847) & 1.756 (0.696)  & 0.688 (0.073) 
			\rule[-0.9ex]{0pt}{0pt} \\ \hline
		\end{tabular}%
	}
	\label{table3}
\end{table*}

We next quantify the performance of the proposed methods by measuring a reconstruction error between the fitted and true curves to measure the reconstruction ability of the methods. For the fitted curve $\hat{f}$, the reconstruction error is defined as \small{$\sum_{i=1}^n d^2_{Geo}\big(x_i,\, \hat{f}\big(\lambda_{\hat{f}}(\tilde{x}_i)\big)\big)$}\normalsize, where $\{x_i\}_{i=1}^n$ denote the true values of the generating curves and $\{\tilde{x}_i\}_{i=1}^n$ denote noisy sample values. We also count the number of distinct projection points to evaluate the continuity of resulting curves of the methods. Moreover, we compare the proposed spherical principal curves with Hauberg's method over various settings $T=100,\, 500$, $q=0.05,\, 0.03,\, 0.01$, and $\sigma= 0.07,\, 0.1$. 

\Cref{table3} lists the average values of reconstruction errors and their standard deviations over 50 simulation sets. As listed, the proposed (intrinsic) principal curves outperform the Hauberg's method, recovering the true curves accurately. \Cref{table4} provides the average values of distinct projection points and their standard deviations. The proposed method provides a very large number of distinct projection points compared to that of Hauberg's method. Overall, as listed in \Cref{table3} and \ref{table4}, our methods perform better than Hauberg's method, including the case that the number of points of the curves ($T=500$) is much larger than the number of data points ($n=100$). In addition, the results of the intrinsic and extrinsic principal curves are similar in terms of both reconstruction error and the number of distinct projection points, which appear with the fact that the intrinsic and extrinsic means are almost identical for localized data, as noted in \cite{Bhattacharya2005}. The results of the extrinsic approach are almost identical to those of the intrinsic one, and hence are omitted. 

\begin{table*}[ht]
		\caption{Averages of distinct projection points and their standard deviations in the parentheses}
	\resizebox{1.0\textwidth}{!}{%
		\begin{tabular}{c||c|c|cccccc}
			\hline
			& \multirow{3}{*}{True form} & \multirow{3}{*}{Method} & \multicolumn{6}{c}{Noise level}                                                           
			\\  
			&                            &                         & \multicolumn{3}{c}{$\sigma=0.07$}                                & \multicolumn{3}{c}{$\sigma=0.1$}                 
			\\  \cline{4-9} 
			&                            &                         & $q=0.05$       & $q=0.03$       &\multicolumn{1}{c|}{$q=0.01$}   &$q=0.05$        & $q=0.03$         & $q=0.01$        \rule{0pt}{2.2ex} \\ \hline \hline
			\multirow{4}{*}{$T=100$} & \multirow{2}{*}{Circle}     
			                            & Proposed     & 99.02 (0.32) & 98.92 (0.34)  &\multicolumn{1}{c|}{98.84 (0.47)} &99.08 (0.34)  & 98.72 (1.11)   & 98.20 (1.12)  \rule{0pt}{2.2ex} \\  
			&                            & Hauberg                 & 87.70 (7.95) & 56.68 (17.99) &\multicolumn{1}{c|}{64.70 (3.22)} &69.80 (12.28) & 47.04 (15.44)  & 60.42 (2.83) \rule[-0.9ex]{0pt}{0pt}\\ \cline{2-9} 
			& \multirow{2}{*}{Wave}       
			                            & Proposed     & 93.36 (4.47) & 97.28 (2.13)  &\multicolumn{1}{c|}{99.32 (0.51)} &95.82 (3.77)  & 96.72 (2.22)   & 99.10 (0.65)  \rule{0pt}{2.2ex} \\  
			&                            & Hauberg                 & 22.72 (2.77) & 25.94 (2.65)  &\multicolumn{1}{c|}{62.14 (2.49)} &24.5 (3.63)   & 32.16 (16.72)  & 58.84 (3.04) \rule[-0.9ex]{0pt}{0pt}\\ \hline 
			\multirow{4}{*}{$T=500$} & \multirow{2}{*}{Circle}      
			                            & Proposed     & 99.08 (0.27) & 99.02 (0.25)  &\multicolumn{1}{c|}{98.76 (0.69)} &99.1 (0.30)   & 99.04 (0.49)   & 99.30 (1.09)  \rule{0pt}{2.2ex} \\  
			&                            & Hauberg                 & 97.8 (1.47)  & 89 (8.63)     &\multicolumn{1}{c|}{79.28 (4.60)} &93.64 (5.29)  & 78.72 (13.08)  & 73.86 (7.28) \rule[-0.9ex]{0pt}{0pt}\\ \cline{2-9}  
			& \multirow{2}{*}{Wave}        
			                            & Proposed     & 99.18 (0.39) & 98.5 (1.27)   &\multicolumn{1}{c|}{99.26 (0.56)} &99.22 (0.42)  & 98.84 (1.20)   & 99.18 (0.66)  \rule{0pt}{2.2ex} \\  
			&                            & Hauberg                 & 45.2 (24.8)  & 43.38 (3.72)  &\multicolumn{1}{c|}{73.20 (3.42)} &52.04 (26.81) & 50.64 (19.99)  & 71.52 (4.38) \rule[-0.9ex]{0pt}{0pt}\\ \hline
		\end{tabular}
	}
	\label{table4}
\end{table*}

\subsubsection{Influence of $T$ and $q$}\label{influence:Tq}
Here we discuss the influence of the hyperparameters $T$ and $q$. To this end, we consider the waveform simulated data used in \Cref{simul:S^2}. \Cref{fig:sweep} visualizes the fitted curves by the proposed extrinsic method for various $q$'s in the range of $[0.01,\, 0.1]$ at intervals of 0.01 with a fixed $T=500$. As shown in the top panels of \Cref{fig:sweep}, the resulting curve with $q=0.01$ is wiggly, and  the curve with $q=0.1$ is almost flat. In general, the curves tend to overfit data when the $q$ value is small, whereas the curves tend to underfit data when the $q$ value is large. On the other hand, the bottom panels of \Cref{fig:sweep} show the fitted curves by the same method for a fixed $q=0.06$ and varying $T$ in $\{10,\, 20,\, 50,\, 100,\, 200,\, 500\}$. The curve of the bottom left panel implemented by a small $T$ value, such as $T=10$, does not represent the data well. For appropriate $T$ values, the spherical principal curves of the right panel successfully recover the underlying structure of the data.
\begin{figure}
	\centering
	\includegraphics[scale=0.127]{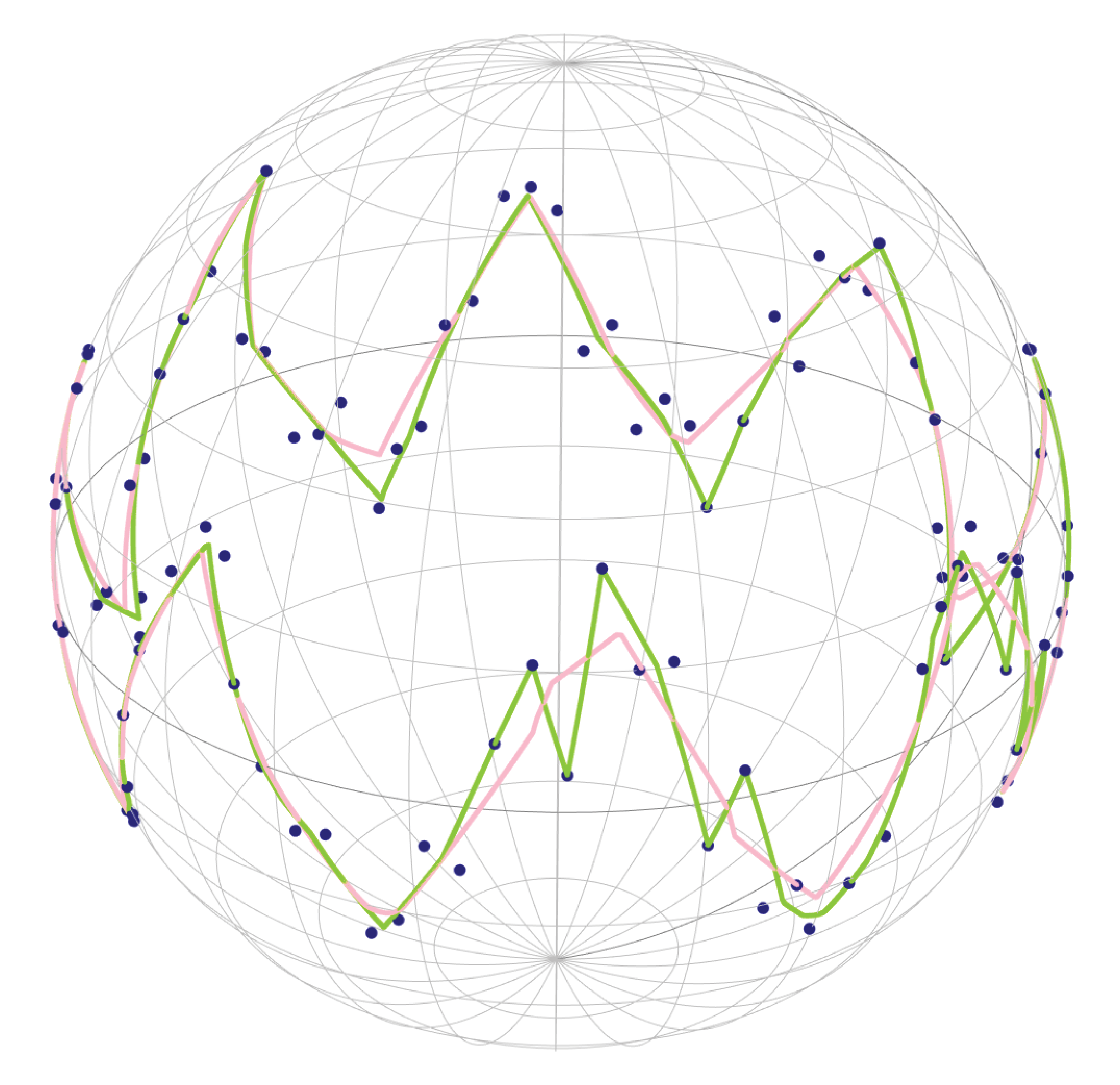}
	\hspace{0.3cm}
	\includegraphics[scale=0.127]{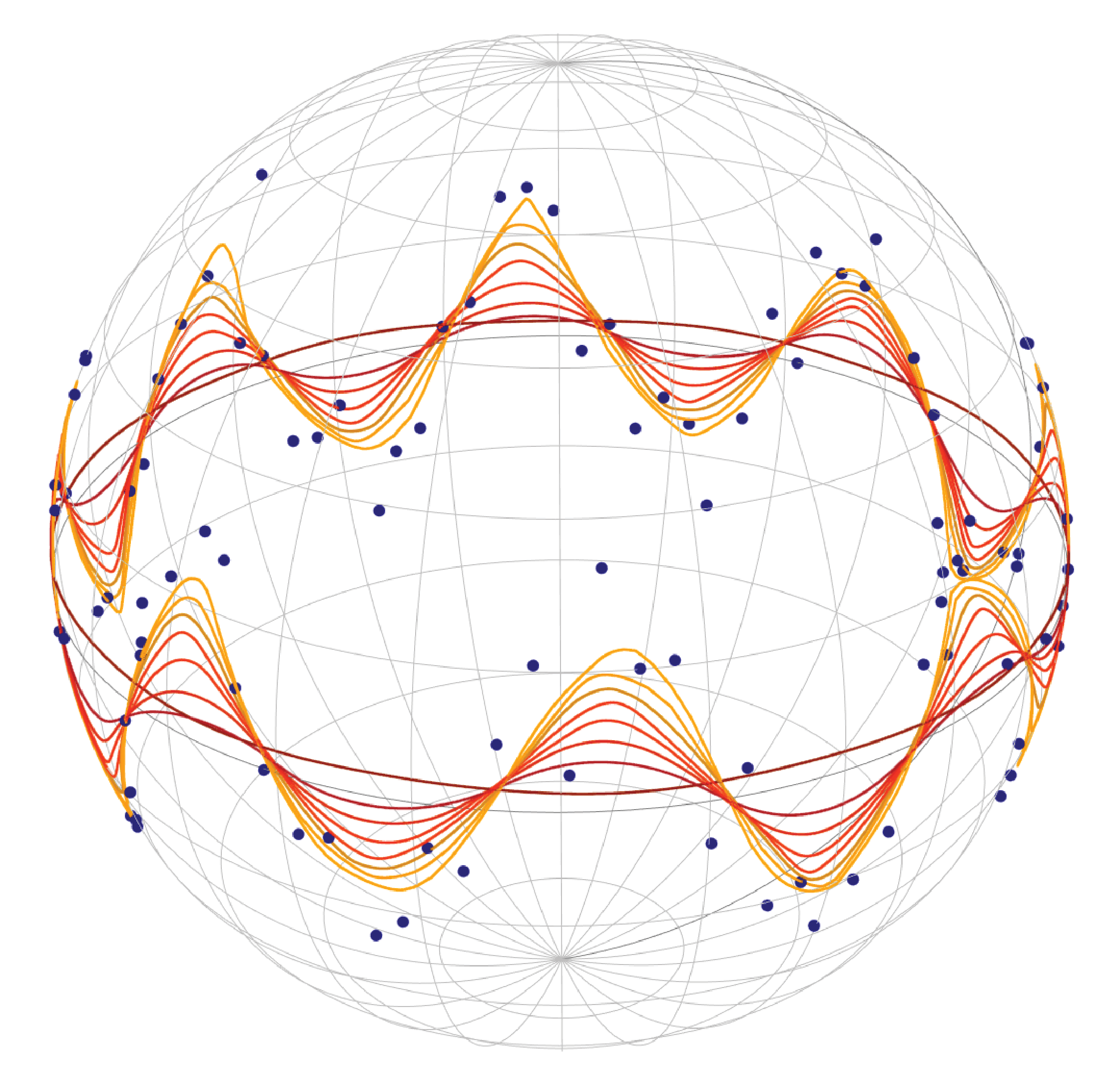}	
	\includegraphics[scale=0.127]{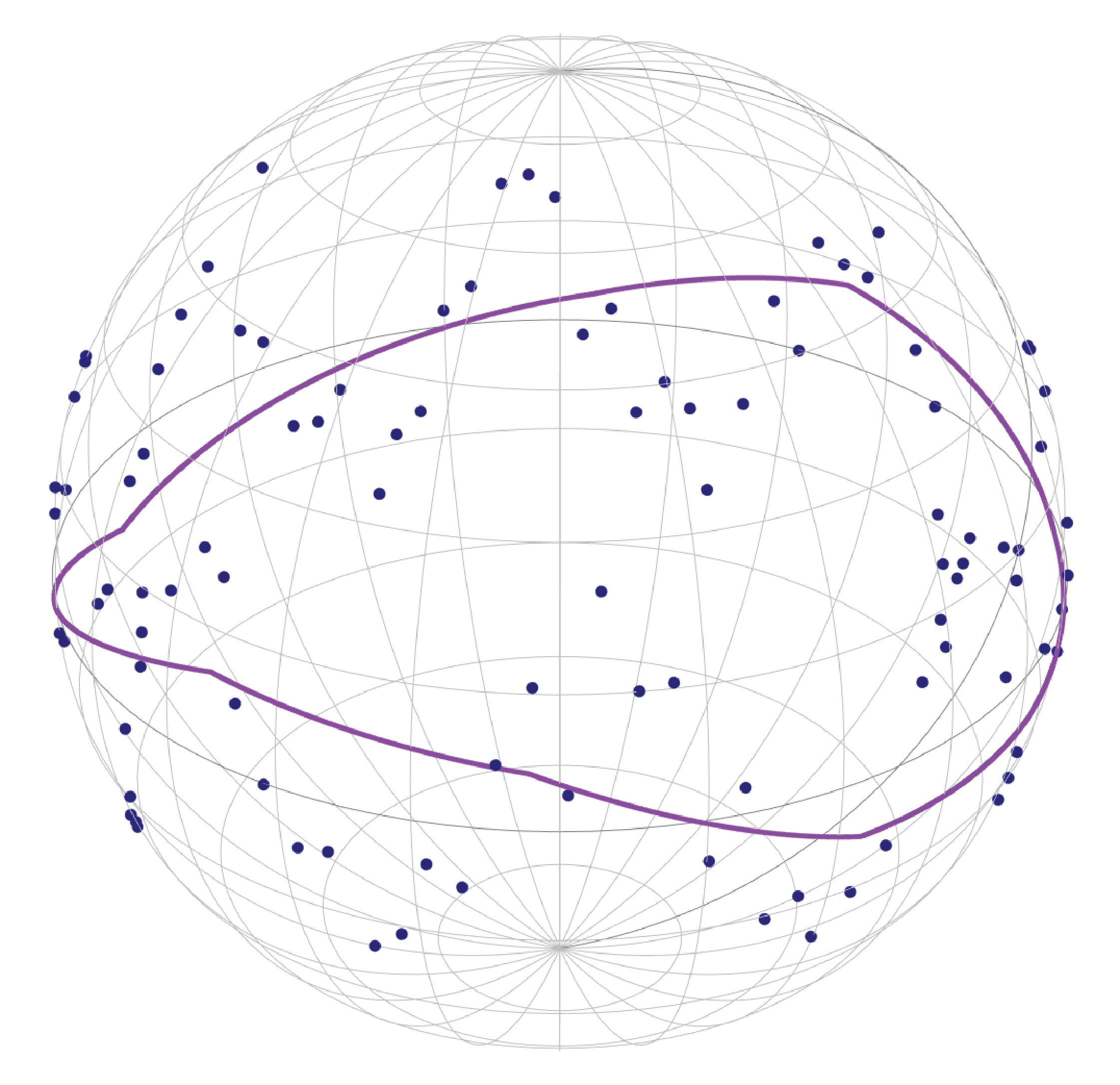}
	\hspace{0.3cm}
	\includegraphics[scale=0.127]{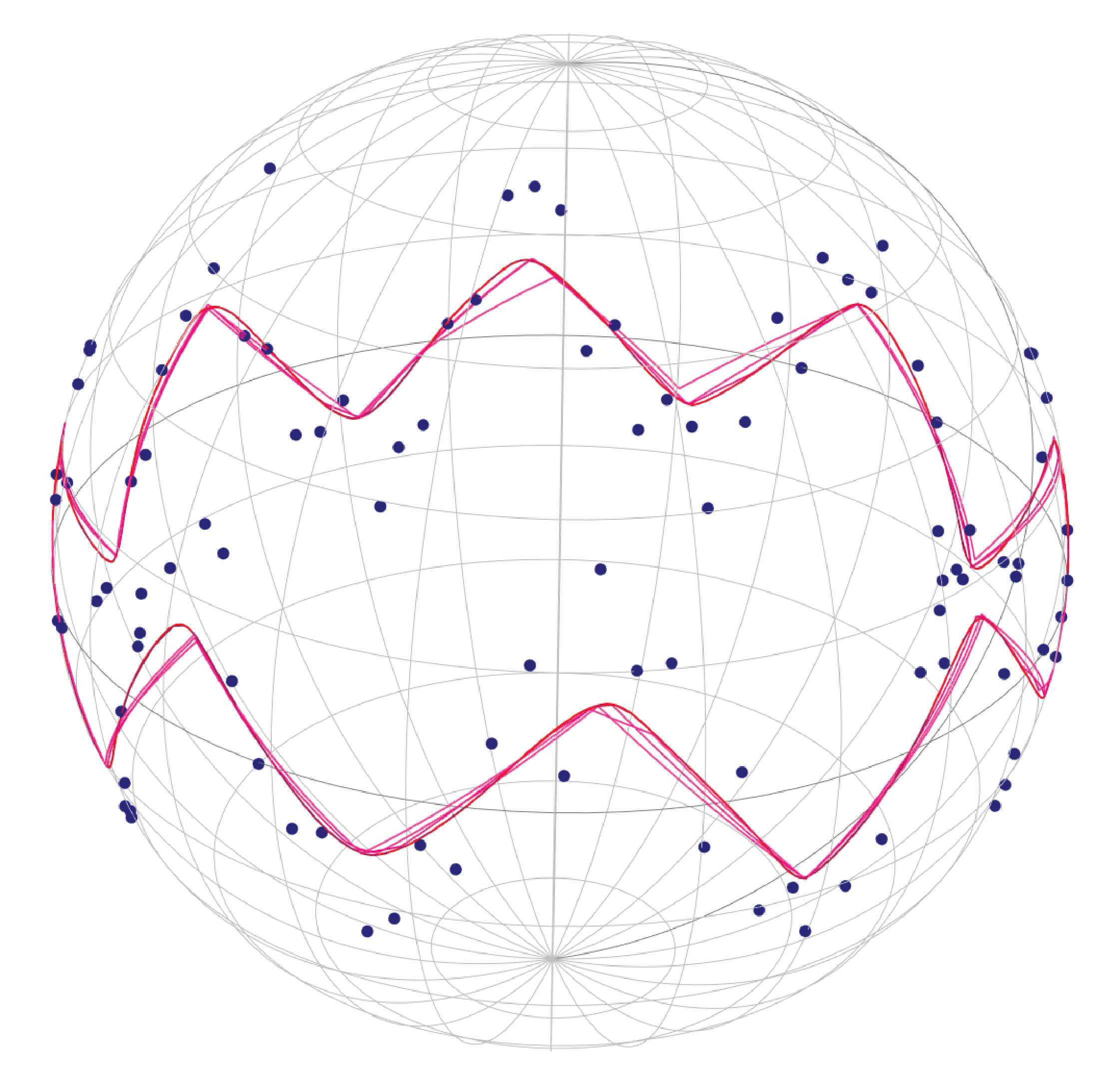}	
	\caption{Noisy waveform simulated data (blue). Top left: Extrinsic principal curves with $q=0.01$ (green) and $0.02$ (pink) for fixed $T=500$. Top right: Influence of varying $q$ over $[0.03,\, 0.1]$ with a step size 0.01 (from yellow to brown) for fixed $T=500$. Bottom left: Extrinsic principal curves (purple) with $T=10$ and $q=0.06$. Bottom right: Influence of varying $T$ in $\left\{20,\, 50,\, 100,\, 200,\, 500 \right\}$ (from violet to red) for a fixed $q=0.06$.}
	\label{fig:sweep} 
\end{figure}

\subsubsection{Simulation on Hypersphere}\label{simul:S^4}

\begin{table}[ht]
	\centering
		\caption{A simulation result of waveform data on $S^4$}
		\begin{tabular}{c|ccc}
		    \hline 
			\multirow{1}{*}{Method}           & $q=0.03$          & $q=0.005$           & $q=0.002$   \rule[-1.0ex]{0pt}{3.2ex}  \\ \hline \hline
			\text{Proposed (extrinsic)}        & 0.211 (0.230)     &  0.179 (0.162)      & 0.199 (0.235)          \rule{0pt}{2.2ex} \\
		    \text{Proposed (intrinsic)}        & 0.729 (0.493)     &  0.267 (0.264)      & 0.150 (0.232)          \rule{0pt}{2.2ex} \\
		    \text{Hauberg}                    & 1.990 (0.815)     &  0.481 (0.215)    & 0.357 (0.251)  \rule[-1.0ex]{0pt}{3.1ex}  \\ \hline
		\end{tabular}%
	\label{table5}
\end{table}

We conduct a simulation study on a hypersphere. To this end, we consider a waveform simulated data on $S^4$ represented by four angular parameters $\varphi_1,\, \varphi_2,\, \varphi_3 \in [0,\, \pi),\, \mbox{and}~ \varphi_4 \in [0,\, 2\pi)$. The explicit representation on $S^d$, $d\ge 3$ is given in \Cref{hypercircle}. Mimicking a waveform dataset on $S^2$ in \Cref{simul:S^2}, we craft simulation sets $(r=1,\, \varphi_1,\, \varphi_2,\, \varphi_3,\, \varphi_4)$ with $\varphi_1=\varphi_2=\varphi_3=\alpha\sin(\varphi_4 f) + \pi/2$ and $0\le \varphi_4 < 2\pi$, frequency $f=2$, and amplitude $\alpha=1/2$. Data points of $n=200$ are generated by sampling $\varphi_4$ uniformly in $[0,\, 2\pi)$ and adding the random noises sampled from $N(0,\, \sigma^2)$ to $\varphi_1$ with $\sigma=0.05$. \Cref{table5} lists the average values of reconstruction errors defined on \Cref{simul:S^2} and their standard deviations over 50 simulation sets for each method with $T=300$. As listed, the proposed principal curves outperform Hauberg's method, recovering the true curves more closely.

\section{Proofs}\label{proofs}
\subsection{Justification of the Projection Steps on $S^d$}\label{proof:proj}
Let $A=(a_{1},\, a_{2},\, \ldots,\, a_{d+1})$,\, $B=(b_{1},\, b_{2},\, \ldots,\, b_{d+1})$, $C=(c_{1},\, c_{2},\, \ldots,\, c_{d+1})\in S^{d} \subset \mathbb{R}^{d+1}$ with $(A\cdot B)^2 \neq 1$. Any point $P$ on $\wideparen{AB}$ is denoted by $P=\mu A + \lambda B$ for $\mu,~\lambda \in \mathbb{R_+}$ with $\mu^2+\lambda^2=1$. If $B\cdot C = C\cdot A=0$, then we have 
\begin{eqnarray*}
d_{Geo}(C,\, P)  =  \arccos\big(C\cdot (\mu A + \lambda B)\big) =  \pi/2.
\end{eqnarray*}
Hence, any points on $\wideparen{AB}$ have the same geodesic distance of $\pi/2$ from $C$. We may assume that $A$, $B$, and $C$ do not satisfy $B\cdot C=C\cdot A=0$. 

The orthogonal complement of $V$ in $\mathbb{R}^{d+1}$, $V^{\perp}$, has a dimension of $d-1$, owing to the fact that $\mathbb{R}^{d+1}=V\oplus V^{\perp}$ with $\oplus$ denoting the direct sum. As a column vector notation, we choose an orthonormal basis for $V$ as $R_1,\, R_2 \in \mathbb{R}^{d+1}$ and an orthonormal basis for $V^{\perp}$ as $R_3,\, R_4,\, \ldots,\, R_{d+1}\in \mathbb{R}^{d+1}$. Define a $(d+1)\times (d+1)$ matrix $R = [R_1, R_2,\, R_3,\, \ldots,\, R_d,\, R_{d+1}]^{T}$. Clearly, $R$ is a rotation (orthogonal) matrix, \textit{i.e.} $R\in O(n)=\left\{X\in M_{d+1,\, d+1}(\mathbb{R}) \ | \ X^{T}X=I \right\}$ and satisfies that $RA = (\tilde{a}_{1},\, \tilde{a}_{2},\, 0,\, 0,\, \ldots,0)$ and $RB = (\tilde{b}_{1},\, \tilde{b}_{2},\, 0,\, 0,\, \ldots,\, 0)$. Let $\tilde{A}=RA$,\, $\tilde{B}=RB$, and $\tilde{C}=RC=(\tilde{c}_1,\,  \tilde{c}_2,\, \ldots,\, \tilde{c}_d, \tilde{c}_{d+1})$. Let $\tilde{V}$ be a two-dimensional vector space spanned by $\tilde{A}$ and $\tilde{B}$, as shown in the right panel of \Cref{fig:hyper}. It follows that $\tilde{V}=\left\{x=(x_1,\, x_2,\, x_3,\, \ldots,\, x_{d+1}) \ | \ x_3=x_4=\cdots=x_{d+1}=0 \right\}$.
We denote the projection of $\tilde{C}$ onto $\tilde{V}$ as $\tilde{C'}=(\tilde{c}_1,\, \tilde{c}_2,\, 0,\, \ldots,\, 0) \in \mathbb{R}^{d+1}$ with $\tilde{c}_1^2+\tilde{c}_2^2 \neq 0$. For any $\tilde{P}=(\tilde{p}_1,\, \tilde{p}_2,\, 0,\, \ldots,\, 0) \in \tilde{V}\cap S^{d}$, it follows that 
\begin{equation}
\label{geodist}
d_{Geo}(\tilde{C},\, \tilde{P})=\arccos(\tilde{c_1}\tilde{p_1}+\tilde{c_2}\tilde{p_2}) \ge \arccos(\sqrt{\tilde{c}_1^2+\tilde{c}_2^2}), 
\end{equation}
where the last inequality holds due to the Cauchy-Schwarz inequality $(\tilde{c_1}\tilde{p_1}+\tilde{c_2}\tilde{p_2})^2 \le (\tilde{c}_1^2+\tilde{c}_2^2)(\tilde{p}_1^2+\tilde{p}_2^2)=(\tilde{c}_1^2+\tilde{c}_2^2)$. The equality of (\ref{geodist}) holds when $(\tilde{p}_1,\, \tilde{p}_2) = t(\tilde{c}_1,\, \tilde{c}_2)$ for some $t\in \mathbb{R}_{+}$. It means that the closest point $\tilde{P}$ on $\tilde{V}\cap S^d$ from $\tilde{C}$ is found by normalizing $\tilde{C'}$ so that it is in $S^d$. Since $R$ is an orthogonal matrix, for any $P \in V \cap S^d$ and $\tilde{P}=RP \in \tilde{V} \cap S^d$, it follows, as a column vector notation, that
\begin{eqnarray*}
    d_{Geo}(\tilde{C},\, \tilde{P}) =  \arccos(\tilde{C}^{T}\tilde{P})=\arccos(C^T R^T RP) 
    =  \arccos(C^T P) 
    =  d_{Geo}(C,\, P).
\end{eqnarray*}
Accordingly, $\mbox{proj}(C)$ is obtained by applying $R^{-1}$ to $\mbox{proj}(\tilde{C})$ that is the projection of $\tilde{C}$ onto $\tilde{V}\cap S^d$. Since the rotation is a rigid motion, it completes the proof.
\begin{figure}[!ht]
  \centering
  \includegraphics[scale=0.3]{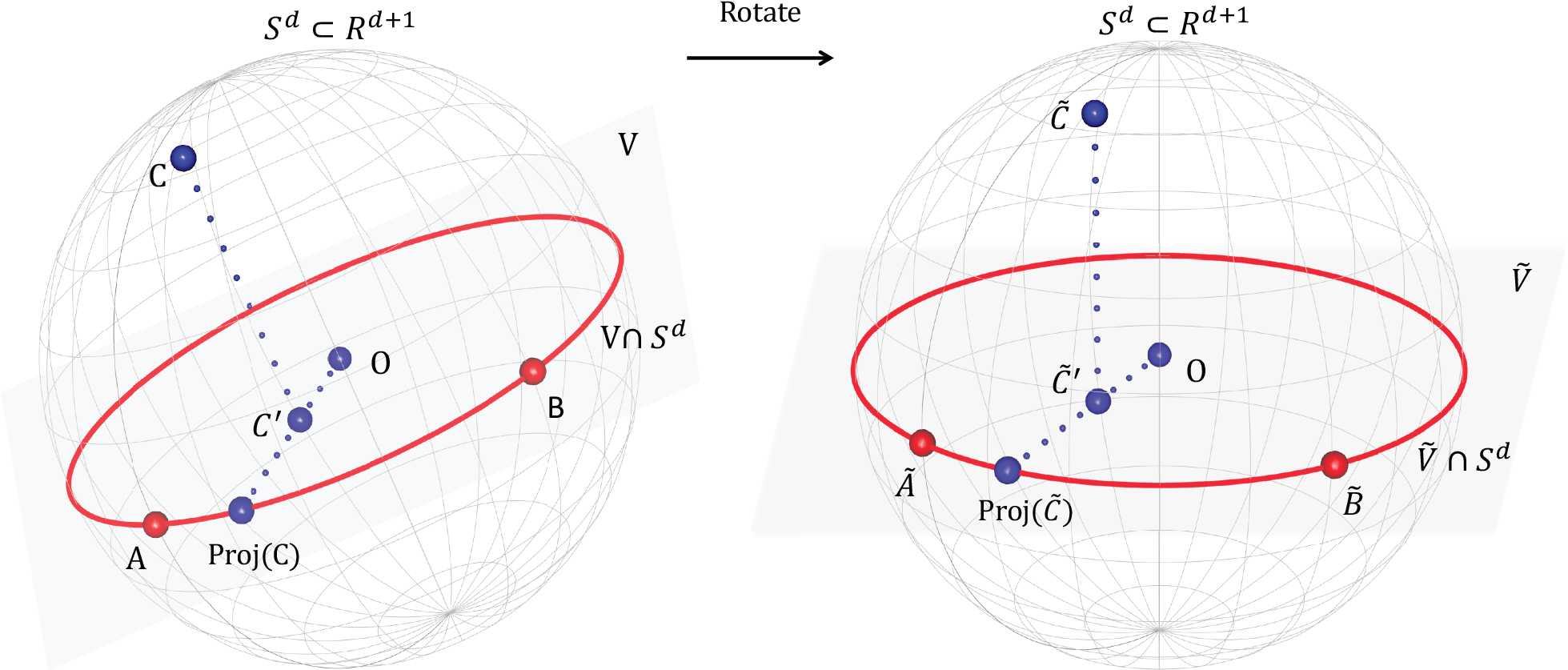}
  \caption{(Left) The projection process of $C$ onto the one-dimensional great circle $V\cap S^{d}$ (red) in a hypersphere $S^{d}\subset \mathbb{R}^{d+1}$: (i) find the projection of $C$ onto $V$, $C'$ and (ii) obtain the projection of $C'$ onto $V\cap S^{d}$, $\mbox{proj}(C)$. (Right) The rotated configuration of the objects.}
  \label{fig:hyper}
\end{figure}

\subsection{Stationarity of Principal Curves}\label{proof:stationarity}
Here we cover a smooth (infinitely differentiable) curve that does not cross on a sphere $\big(\textit{i.e.,}~ \lambda_{1}\ne \lambda_{2} \in [0,\, 1) \Rightarrow f(\lambda_1)\ne f(\lambda_{2})\big)$, including curves with end points and closed curves, which can be both parameterized over interval $[0,\, 1]$ by a constant speed, \textit{i.e.} $f'(\lambda)=s>0$ for any $\lambda\in [0,\, 1]$. In the latter case, a boundary condition is needed; any order partial derivatives of $f$ at end points are the same, \textit{i.e.,} $f^{(k)}(0)=f^{(k)}(1) ~ \mbox{for}~\mbox{all} \ k \ge 0$. For a random vector $X$ on a sphere, we further assume that the curve $f$ are not short enough to cover the support of $X$ well, \textit{i.e.,} $\lambda_{f}(X) \neq 0, 1 ~ \mbox{for}~\mbox{\textit{a.e.}}~X$. For example, any closed curve satisfies the condition $\lambda_f(X)\neq 0,\, 1$ for \textit{a.e.} $X$, meaning that almost all $X$ is orthogonally projected onto the curve $f$. Note that $f$ is smooth on $[0,\, 1]$, \textit{i.e.} $f$ is smoothly extended on $[0,\, 1]$; thus any order its derivatives are continuous on $[0,\, 1]$. Our main purpose is to prove the stationarity of extrinsic, intrinsic principal curves $f: [0,\, 1] \to S^d$ for $d\ge2$ that satisfy the equations (8) and (9) in Theorems 1 and 2. We first consider the 2-sphere, and then extend $d$-spheres, $d\ge2$.

When moving from Euclidean space to spherical surfaces, topological properties such as measurability and continuity are preserved, while the formula using specific distance should be modified. This modification could be obtained by embedding a spherical surface $S^d$ into a $(d+1)$-dimensional Euclidean space. Specifically, we embed a spherical surface as a unit sphere centered at the origin, \textit{i.e.} $S^d \hookrightarrow \mathbb{R}^{d+1}$, and investigate further derivations. When $d=2$, for a smooth curve $f: [0,\, 1] \to S^2$, suppose that $f$ is parameterized by a constant speed with respect to $\lambda$ and is expressed as three-dimensional coordinates $(f(\lambda)_1,\, f(\lambda)_2,\, f(\lambda)_3)$. Then the following lemmas are held. 
\begin{lemma}
	A curve $f$ satisfies $f'(\lambda) \cdot f(\lambda) = 0$ and $f''(\lambda) \cdot f'(\lambda) = 0$, $\forall\lambda \in [0, 1]$, where $\cdot$ denotes inner product in $\mathbb{R}^3$. 
\end{lemma}
\begin{proof}
It is directly obtained from $f(\lambda)\cdot f(\lambda)=1$ and $f'(\lambda)\cdot f'(\lambda) = constant$.
\end{proof}

\begin{lemma}
	Suppose that $f(\lambda)$ and $x$ are expressed as three-dimensional vectors. Then, it follows that $d_{Geo}(f(\lambda),\, x)=\arccos(f(\lambda)\cdot x)$, where $\arccos(f(\lambda)\cdot x)$ is the angle between $f(\lambda)$ and $x$. Then, $\frac{df}{d\lambda} (\lambda_f(0,\, 0,\, 1))_3 =0$. Thus, it follows that $\frac{df}{d\lambda} (\lambda_f(0,\, 0,\, 1)) = a\big(-f(\lambda_f(0,\, 0,\, 1))_2, f(\lambda_f(0,\, 0,\, 1))_1,\, 0\big)$ for some $a \in \mathbb{R}$. Note that $\lambda_f(x) $ denotes the projection index of point $x$ to the curve $f$. 
\end{lemma}
\begin{proof}
	For $p = (0,\, 0,\, 1)$, it follows that $d_{Geo}(f(\lambda),\, p)=\arccos(f(\lambda)_3)$. From the assumption that $f(\lambda)$ is a smooth curve and the fact that $d_{Geo}(f(\lambda),\, p)$ has the minimum at $\lambda_f(0,\, 0,\, 1)$,  the remaining part of the lemma follows by differentiation with respect to $\lambda$. 
\end{proof}

\begin{lemma} (Spherical law of cosines) 
	Let $u$, $v$, $w$ be points on a sphere, and $a$, $b$ and $c$ denote  $d_{Geo}(w,\, u)$, $d_{Geo}(w,\, v)$ and $d_{Geo}(u,\, v)$, respectively. If $C$ is the angle between $a$ and $b$, \textit{i.e.,} the angle of the corner opposite $c$, then, 
	$
	\cos c=\cos a \cos b + \sin a \sin b \cos C. 
	$
	Further, with three-dimensional vectors  $u$, $v$, $w$, it follows that 
	$
	\sin a \sin b \cos C=(w\times u)\cdot(w\times v), 
	$      
	where $\times$ denotes cross product in $\mathbb{R}^3$. 
\end{lemma} 

The following property can be obtained from Definition 2.
\begin{proposition}
	Under the same conditions in Definition 2, $f(\epsilon, \lambda):=(f+\epsilon g)(\lambda)$ is smooth on $[-1,\, 1] \times [0,\, 1]$. Hence, for each $\epsilon \in [-1,\, 1]$, $f+\epsilon g:[0,\, 1] \rightarrow S^d$ is a smooth curve on $S^{d}$ and $\lim\limits_{\epsilon \to 0}(f+\epsilon g)(\lambda) = f(\lambda)~ \mbox{for}~\lambda \in [0,\, 1]$.
	\begin{proof}
		For simplicity, we denote $f+g$ as $h$. Let $R_{a,\, b}(\theta)$ be a rotation matrix that rotates points on $S^d$ by $\theta$ in the direction along the geodesic from $a$ to $b$ with $a,\, b \in S^d\subset \mathbb{R}^{d+1}$ satisfying $ a\neq -b \in \mathbb{R}^{d+1}$ and $\theta$ ranging over $[0,\, \pi)$. Then, it has a closed form; formally, as a column vector notation, $R_{a,\,b}(\theta)= I_{d+1} + \sin(\theta)B + (\cos(\theta)-1)(bb^T + cc^T)$, where if $a \neq b$, then $c=\big(a-b(b^T a)\big)/\norm{a-b(b^T a)}$, otherwise $c=0$, and $B=bc^T-cb^T$. For more details, refer to the Section 8.1 in \cite{Jung2012}. Hence, we obtain $f(\epsilon,\, \lambda):=(f+\epsilon g)(\lambda) =R_{f(\lambda),\, h(\lambda)}(\theta)f(\lambda)$, where $\theta=\epsilon\arccos \big(f(\lambda)\cdot h(\lambda)\big)$. Thus, $(f+\epsilon g)(\lambda)\big(=f(\epsilon,\, \lambda)\big)$ is smooth on $[-1,\, 1] \times [0,\, 1]$ since all functions $f$, $h$, $R$, and $\theta$ are smooth. Therefore, for a fixed $\epsilon \in [0,\, 1]$, the smoothness of $(f+\epsilon g)(\lambda)$ with respect to $\lambda$ also follows.  Moreover, the last equality is directly followed by the definition of $f+\epsilon g$.
		\end{proof}
\end{proposition}

In Euclidean space, we have $g(\lambda)= \frac{\partial}{\partial \epsilon} f_{\epsilon}(\lambda)$, where 
$f_{\epsilon}(\lambda) := f+\epsilon g$. From this fact, the magnitude of perturbation $\norm{h-f}$ is defined as $g(\epsilon_0,\, \lambda) := \frac{\partial}{\partial \epsilon} \big|_{\epsilon=\epsilon_0} f_{\epsilon}(\lambda)$, $\norm{g(\lambda)} := d_{Geo}\big(f(\lambda),\, g(\lambda)\big) = \ \big|g(\epsilon_0,\, \lambda) \big|$, $\norm{g}:=\max_{\lambda} \norm{g(\lambda)}$, and finally $\norm{h-f} := \norm{g}\ne \pi.$ The boundedness of $\norm{g}$ guarantees that $\epsilon$-internal division of the geodesic from $f$ to $h$ converges to $f$ uniformly on $\lambda \in [0,\, 1]$ as $\epsilon$ goes to 0. Notice that from the compactness of the unit sphere, $\norm{h-f}$ is inherently equal or less than $\pi$; thus, the assumption of $\norm{h-f} \ne \pi$ implies that $\norm{h-f} < \pi $. 

Moreover, the norm of derivative of perturbation $\norm{(h-f)'}$ is defined as $g'(\epsilon,\, \lambda_0) := \frac{\partial}{\partial \lambda} \big|_{\lambda=\lambda_{0}} g(\epsilon,\, \lambda)$, $\norm{g'(\lambda_0)} := \max_{\epsilon} \norm{g'(\epsilon,\, \lambda_0)}$, $\norm{g'} := \max_{\lambda_0} \norm{g'(\lambda_0)}$, and finally $\norm{(h-f)'} := \norm{g^\prime}$. 

Let $x$ be a point on a sphere. By the continuity of $f$ and the compactness of the domain set, it follows that $\inf_{\lambda \in [0,\, 1]} d_{Geo}\big(x,\, f(\lambda)\big)$ can be attained. Let $d_{Geo}(x,\, f)$ denote the geodesic distance between $x$ and $f$, \textit{i.e.,} $d_{Geo}(x,\, f):=\min_{\lambda \in [0,\, 1]} d_{Geo}\big(x,\, f(\lambda)\big)$. By the continuity of $f$ again, $\{\lambda \in [0,\, 1] \ | \ d_{Geo}\big(x,\, f(\lambda)\big)=d_{Geo}(x,\, f)\}$ is closed and therefore compact. Thus, the projection indices $\lambda_f(x) = \inf\{\lambda \ | \ d_{Geo}\big(x,\, f(\lambda)\big)=d_{Geo}(x,\, f)\}$ and $\lambda_{f+\epsilon g}$ are well defined. The latter holds due to the fact that $f+\epsilon g$ is a continuous curve by Proposition 1. When $\mbox{card}\{\lambda \ | \ d_{Geo}\big(x,\, f(\lambda)\big)=d_{Geo}(x,\, f)\}>1$, the point $x$ is called an ambiguity point of $f$. The set of ambiguity points of the smooth curve has spherical measure 0; thus, the ambiguity points are negligible when calculating the expected value.

The next topological properties of the principal curves established in Euclidean space of \cite{Hastie} are still valid in spherical surfaces.
\begin{proposition} (Measurability of index function) 
	For a continuous curve $f$ on $S^d$, the index function $\lambda_f:S^d \rightarrow [0,\, 1]$ by $x \mapsto \lambda_f(x)$ is measurable.
	\label{measure}
	\begin{proof}
	It follows that of Theorem 4.1 in \cite{hastie1984principal}. It is enough to show that, for any constant $a \in [0,\, 1]$, $\left\{x\in S^d \ | \ \lambda_{f}(x)\ge a \right\}$ is a measurable set on $S^d$. If $a=0$, then the set is $S^d$; thus, we may assume that $c \in (0,\, 1]$. By the definition of the projection index $\lambda_f(x) = \inf\{\lambda \ | \ d_{Geo}\big(x,\, f(\lambda)\big)=d_{Geo}(x,\, f)\}$, the condition $\lambda_f(x)\ge a$ is equivalent to the property that for any $\lambda_1\in[0,\, a)$, there exists $\lambda_2\in[a,\, 1]$ such that $d_{Geo}\big(x,\, f(\lambda_2)\big) < d_{Geo}\big(x,\, f(\lambda_1)\big)$. Technically, we aim to prove that 
	
\noindent	$\lambda_f(x)\ge a
	 \xLeftrightarrow[]{(1)}$
	 for any $\lambda_1\in[0,\, a)$, there exists $\lambda_2\in[a,\, 1]\cap \mathbb{Q}$ such that $d_{Geo}\big(x,\, f(\lambda_2)\big) < d_{Geo}\big(x,\, f(\lambda_1)\big)\xLeftrightarrow{(2)}$ 
	 for any $\lambda_1\in[0,\, a]\cap\mathbb{Q}$, there exists $\lambda_2\in[a,\, 1]\cap\mathbb{Q}$ such that $d_{Geo}\big(x,\, f(\lambda_2)\big) < d_{Geo}\big(x,\, f(\lambda_1)\big)$. Here $\mathbb{Q}$ is the set of rational numbers. If $(1)$ and $(2)$ are verified, we obtain that
	 \begin{eqnarray*}
	     \left\{\lambda_{f}(x)\ge a \right\}	     = \bigcap_{\lambda_1 \in [0,\, a) \cap \mathbb{Q}} \bigcup_{\lambda_2 \in [a,\, 1] \cap \mathbb{Q}} \left\{d_{Geo}\big(x,\, f(\lambda_2)\big) < d_{Geo}\big(x,\, f(\lambda_1)\big) \right\}. \hspace{0cm}
	 \end{eqnarray*}Each set is measurable on $S^d$ because for any $\lambda_1 ~\mbox{and}~ \lambda_2$, the function $x \mapsto d_{Geo}\big(x,\, f(\lambda_1)\big) - d_{Geo}\big(x,\, f(\lambda_2)\big)$ is continuous. Accordingly, $\left\{\lambda_f(x)\ge a\right\}$ is measurable due to the fact that countable unions and intersections of measurable sets are also measurable. It completes the proof. \\\\
	 \textit{Proof of (1).} For any $\lambda_1\in[0,\, a)$, there is $\lambda_2\in [a,\, 1]$ such that $d_{Geo}\big(x,\, f(\lambda_2)\big) < d_{Geo}\big(x,\, f(\lambda_1)\big)$. Since $\mathbb{Q}$ is dense in $\mathbb{R}$ and $f$ is continuous, there is $\lambda_2'\in[a,\, 1]\cap\mathbb{Q}$ such that $d_{Geo}\big(x,\, f(\lambda_2')\big)<d_{Geo}\big(x,\, f(\lambda_1)\big)$, which completes $(1)$. \\\\
	 \textit{Proof of (2).} We want to show that 
	 \begin{eqnarray*} 
	 D &:=&\bigcap_{\lambda_1\in [0,\, a)}\bigcup_{\lambda_2\in [a,\, 1] \cap \mathbb{Q}} \left\{d_{Geo}\big(x,\, f(\lambda_2)\big) < d_{Geo}\big(x,\, f(\lambda_1)\big) \right\} \\ &=&\bigcap_{\lambda_1 \in [0,\, a)\cap \mathbb{Q}} \bigcup_{\lambda_2 \in [a,\, 1]\cap \mathbb{Q}} \left\{d_{Geo}\big(x,\, f(\lambda_2)\big) < d_{Geo}\big(x,\, f(\lambda_1)\big) \right\} \\
	 &=:& E.
	 \end{eqnarray*}The inclusion $D\subseteq E$ is clear. If $x\in E$, for any $\lambda_1\in[0,\, a)\cap\mathbb{Q}$, there is $\lambda_2\in[a, \,1]\cap\mathbb{Q}$ such that $d_{Geo}(x,\, f(\lambda_2)) < d_{Geo}(x,\, f(\lambda_1))$. For such $\lambda_1$ and $\lambda_2$, owing to the continuity of $f$, there is $\Tilde{\lambda_1}\in(\lambda_1,\, a)\cap\mathbb{Q}$ such that $\forall\lambda\in[\lambda_1,\Tilde{\lambda_1})\cap\mathbb{Q} \Rightarrow d_{Geo}(x,\, f(\lambda_2)) < d_{Geo}(x,\, f(\lambda))$. That is,  
	 \begin{eqnarray*}
	 x & \in & \bigcap_{\lambda\in [\lambda_1,\Tilde{\lambda_1}) \cap \mathbb{Q}} \bigcup_{\lambda_3 \in [a,\, 1] \cap \mathbb{Q}} \left\{d_{Geo}\big(x,\, f(\lambda_3)\big) < d_{Geo}\big(x,\, f(\lambda)\big) \right\} \\
	 & =: &F_{\lambda_1,\Tilde{\lambda_1}}. 
	 \end{eqnarray*}
	 Note that $\Tilde{\lambda_1}$ is automatically chosen for each $\lambda_1 \in [0,\, a)\cap\mathbb{Q}$. Since the above derivation is satisfied for any $\lambda_1\in[0,\, a]\cap\mathbb{Q}$, it follows that
	 \begin{eqnarray*}
	 x &\in&\bigcap_{\lambda_1\in[0,\, a)\cap\mathbb{Q},\, \lambda_1 < \Tilde{\lambda_1}}F_{\lambda_1,\, \Tilde{\lambda_1}} \hspace{5cm} \\
	 &=& \bigcap_{\lambda \in [0,\, a)} \bigcup_{\lambda_3 \in [a,\, 1] \cap \mathbb{Q}} \left\{d_{Geo}\big(x,\, f(\lambda_3)\big) < d_{Geo}\big(x,\, f(\lambda)\big) \right\} \\
	 &=& D.
	 \end{eqnarray*}
	 Hence, we obtain $E\subseteq D$; thus, $D=E$. 
	\end{proof}
\end{proposition}
According to \Cref{measure}, $\lambda_{f}(X)$ is a random variable with respect to $X$ as long as $X$ is a random vector on $S^d$ for $d\ge 2$. Thus, a conditional expectation on $\lambda_{f}(X)$ is feasible.

\begin{proposition} (Continuity of projection index under perturbation) 
	If $x$ is not an ambiguity point for continuous curve $f$, then $\lim\limits_{\epsilon \to 0}\lambda_{f+\epsilon g}(x) = \lambda_{f}(x)$. 
	\label{conti}
\begin{proof}
The proof follows the line of the proof of Lemma 4.1 in \cite{Hastie}. It is enough to show that, for any small $\eta >0$  there exists $\delta >0$ such that $|\epsilon| < \delta$ implies  $|\lambda_{f_\epsilon}(x)- \lambda_{f}(x)| < \eta$. 
Define a set $C:= [0,\, 1] \cap (\lambda_{f}(x)-\eta,\,  \lambda_{f}(x) + \eta)^c$ and $d_{C}:= \inf_{\lambda \in C}d_{Geo}\big(x,\, f(\lambda)\big)>d_{Geo}\big(x,\, f(\lambda_{f}(x))\big)$ where $d_C$ is achieved by some $\lambda \in C$ from the compactness of $C$, and the last inequality holds since $x$ is not an ambiguity point of $f$. Choose $\delta=\frac{1}{3}\big[d_{C} - d_{Geo}\big(x,\, f(\lambda_{f}(x))\big)\big]>0$. Then if $|\epsilon| < \delta$, it follows that 
\begin{eqnarray*}
\inf_{\lambda\in C} d_{Geo}\big(x,\, f_{\epsilon}(\lambda)\big) - d_{Geo}\big(x,\, f_{\epsilon}(\lambda_{f}(x))\big) 
&\ge& \inf_{\lambda\in C} d_{Geo}\big(x,\, f(\lambda)\big)
- d_{Geo}\big(f(\lambda),\, f_{\epsilon}(\lambda)\big) \\
& & - d_{Geo}\big(x,\, f(\lambda_{f}(x))\big) - d_{Geo}\big(f(\lambda_{f}(x)),\, f_{\epsilon}(\lambda_{f}(x))\big) \\
&\ge& d_C - d_{Geo}\big(x,\, f(\lambda_{f}(x))\big) - \delta - \delta \\
&=& 3\delta-2\delta > 0 
\end{eqnarray*}
By the definition of $\lambda_{f_\epsilon}(x)$, we obtain $\lambda_{f_{\epsilon}}(x) \notin C$; thus, $| \lambda_{f}(x)-\lambda_{f_{\epsilon}}(x)| < \eta$. It completes the proof.
\end{proof}
\end{proposition}
In the proof of \Cref{conti}, it is possible to apply the triangle inequality on a sphere because the sphere is a metric space equipped with its geodesic distance. The following proposition is an important tool for verifying Theorem 2. 
\begin{proposition} (Uniform continuity of projection index under perturbation) 
	$\lim\limits_{\epsilon \to 0}\lambda_{f+\epsilon g}(x) = \lambda_{f}(x)$ uniformly on the set of non-ambiguity points of $f$. That is, for every $\eta > 0$, there exists $\delta > 0$ such that for any non-ambiguity points $x$, if $|\epsilon| < \delta$, then $|\lambda_{f_{\epsilon}}(x)-\lambda_{f}(x)| < \eta$. 
\end{proposition} 
For guaranteeing uniform continuity of projection index, it is required that $|f''|$ is bounded. It is directly followed by smoothness of $f$ and compactness of $[0,\, 1]$. A proof is similar to that of Proposition 3; thus, we omit the proof.
\begin{proposition}
	Spherical measure of the set of ambiguity points of smooth curve $f$ is 0.
\end{proposition}
Detailed steps for a proof of Proposition 5  are similar with those of \cite{Hastie}.

Meanwhile, to prove Theorem 2, it is essential to verify that $\lambda_{f_{\epsilon}}$ is differentiable for $\epsilon$ and its derivative is uniformly bounded. Thus, it is necessary to define a subset of $S^2$ as $B(\zeta):= \{x \in S^2 \ | \ |f''(\lambda_{f}(x)) \cdot x| > \zeta \}$ for $\zeta \ge 0$. Obviously, as an inclusion of sets, $\{B(\zeta)\}_{\zeta\ge0}$ is decreasing as $\zeta$ goes to $0$. Moreover, the following lemma implies that, as $\zeta$ goes to 0, $B(\zeta)$ covers $S^2$ almost everywhere.
\begin{lemma}
	The image of smooth function from $[0,\, 1]$ to $S^2$ has measure 0. Moreover,  
	\[
	S^2\setminus B(0)=\{x\in S^2 \ | \ |f''(\lambda_{f}(x))\cdot x | =0 \}
	\]
	is an union of images of two smooth functions from $[0,\, 1]$ to $S^2$, which implies that $S^2\setminus B(0)$ are measure 0. Therefore, the measure of $S^{2} \setminus B(\zeta)$ goes to 0 as $\zeta \rightarrow 0$.
\end{lemma}
\begin{proof}
	Suppose that $I:[0,\, 1] \rightarrow S^2$ is smooth. The domain and range of $I$ are the second countable (with usual topology) differentiable manifolds whose dimensions are 1 and 2, respectively. Since $f$ is twice continuously differentiable function and the differential $dI$  has rank 1 which is less than intrinsic dimension of $S^2$, by a generalization of Sard's Theorem, the image $I([0,\, 1])=\{I(x) \in S^2 \ | \ x \in [0,\, 1] \}$ has measure zero. Next, each point $x\in S^2\setminus B(0)$ satisfying $\lambda_{f}(x)\neq 0,\, 1$ is characterized by two equations $f'(\lambda)\cdot x=0$ and $f''(\lambda) \cdot x=0$ for some $\lambda \in [0,\, 1]$. Therefore, we define functions $I_1,\, I_2$ as follows: For all $\lambda \in [0,\, 1]$, 
	\begin{eqnarray*}
	I_{1}(\lambda) &=& f'(\lambda) \times f''(\lambda) / \norm{f'(\lambda)\times f''(\lambda)},\\
	I_{2}(\lambda) &=&-f'(\lambda) \times f''(\lambda) / \norm{f'(\lambda)\times f''(\lambda)}.
	\end{eqnarray*}It is well known that the curvature of a smooth curve lying on the unit sphere is more than 1. It implies that $\kappa=\frac{|f''|}{s^{2}} \ge 1$, where $\kappa$ is the curvature of $f$ and $s := |f'(\lambda)| > 0$ for all $\lambda \in [0,\, 1]$, and hence $f''\neq 0$. We have already known that $f'\cdot f'' = 0$ by Lemma 1. Hence, it is obtained that $f' \times f''\ne 0$. It implies that $I_{1}$ and $I_{2}$ are well defined and smooth. Therefore, we have $S^2\setminus B(0)= I_{1}([0,\, 1]) \bigcup I_{2}([0,\, 1])$, which completes the proof. 
\end{proof}
Lemma 4 means that the constraints of random vector $X$ in Theorems 1 and 2 are almost negligible by setting $\zeta$ infinitesimally small. Denote the set of ambiguity points of smooth curve $f$ on a sphere as $A$, which is a measure zero set by Proposition 5.

\begin{lemma}
	Let $A$ be the set of ambiguity points of smooth curve $f$ on a sphere. Suppose that for any $x\in S^2$, $\lambda_{f}(x) \in (0,\, 1)$ and $x \in A^c \cap B(\zeta)$ for some small $\zeta > 0$. Then $\lambda(\epsilon):=\lambda_{f_\epsilon}(x)$ is a smooth function for $\epsilon$ on an open interval containing 0. Moreover, $\frac{\partial \lambda(\epsilon)}{\partial \epsilon}$ is uniformly bounded on $A^c \cap B(\zeta)$. That is, there are constants $C > 0$ and $\delta > 0$ such that if $|\epsilon_0| < \delta $ and $x \in A^c\cap B(\zeta)$, then $\big|\ \frac{\partial \lambda_{f_{\epsilon}}(x)}{\partial \epsilon}\big|_{\epsilon=\epsilon_{0}} \big| < C$.
\end{lemma} 
\begin{proof}
Since $x$ is a non-ambiguity point of $f$ and satisfies $\lambda_{f}(x)\neq 0,\, 1$, we obtain $\lambda_{f+\epsilon g}(x)\neq 0,\, 1$ for sufficiently small values of $|\epsilon|$ by Proposition 4. Hence, $\lambda(\epsilon)$ is characterized by orthogonality between $f'_{\epsilon}\big(\lambda(\epsilon)\big)$ and the geodesic through $x$ and $f_{\epsilon}\big(\lambda(\epsilon)\big)$ on a small $\epsilon$ near 0; that is, $ f'_{\epsilon}(\lambda) \cdot \big(x-f_{\epsilon}(\lambda)\big) = f'_{\epsilon}(\lambda) \cdot x = 0$ by the same argument in Lemma 1. Then, we define a map $F:[-1,\, 1] \times [0,\, 1] \to \mathbb{R}$ as $F(\epsilon, \lambda)=f'_{\epsilon}(\lambda)\cdot x$. $F$ is a smooth function by Proposition 1. It follows by the definition of $B(\zeta)$ that 
	\begin{eqnarray}
	\frac{\partial}{\partial \lambda} F(\epsilon,\lambda)\big|_{(0,\, \lambda_f)}
	& = & f''(\lambda_f)\cdot x\ne 0. \nonumber
	\end{eqnarray}By implicit function theorem, for each $x\in A^c \cap B(\zeta)$, $\lambda(\epsilon)=\lambda_{f_{\epsilon}}(x)$ is a smooth function for $\epsilon$ and $F\big(\epsilon, \lambda(\epsilon)\big)=0$ in an open interval containing zero. Next, in order to prove uniform boundedness of $\frac{\partial \lambda(\epsilon)}{\partial \epsilon}$, we should verify that $f''_{\epsilon}\big(\lambda_{f_{\epsilon}}(x)\big)$ uniformly converge to $f''\big(\lambda_{f}(x)\big)$ on $A^c \cap B(\zeta)$ as $\epsilon$ goes to 0. First of all, for each $\lambda$, we have  
	\begin{eqnarray}
	\label{unif_1}
	f_{\epsilon_0}(\lambda) &=& f(\lambda)+ \int_{0}^{\epsilon_0} g(\epsilon_0,\lambda)\, d\epsilon \nonumber \\ 
	&\Rightarrow& f''_{\epsilon_0}(\lambda)=f''(\lambda) + \int_{0}^{\epsilon_0}g''(\epsilon,\lambda)\, d\epsilon \nonumber \\
	&\Rightarrow& \norm{f''_{\epsilon_0}(\lambda)-f''(\lambda)} \le \int_{0}^{\epsilon_0}\norm{g''(\epsilon,\lambda)}\, d\epsilon\le \epsilon_0 M, 
	\end{eqnarray}for some $M > 0$. Note that the above derivatives are differentiation by $\lambda$. Also, the second equation holds since $g(\epsilon,\, \cdot)$ is a twice continuously differentiable function for all $\epsilon$; thus, it is able to change the order of derivative and the integration. The last inequality holds because $g''(\epsilon,\, \lambda)\big(=\frac{\partial^2 g(\epsilon, \lambda)}{\partial \lambda^2}\big)$ is continuous on $[-1,\, 1] \times [0,\, 1]$. Hence, it follows that 
	\begin{eqnarray}
		\label{unif_2}
		\norm{f''_{\epsilon}\big(\lambda_{f_{\epsilon}}(x)\big) - f''\big(\lambda_{f}(x)\big)} \le \norm{f''_{\epsilon}\big(\lambda_{f_{\epsilon}}(x)\big) - f''\big(\lambda_{f_{\epsilon}}(x)\big)} \nonumber 
		  + \norm{f''\big(\lambda_{f_{\epsilon}}(x)\big)-f''\big(\lambda_{f}(x)\big)} \nonumber
		  \to 0, 
	\end{eqnarray}as $\epsilon\rightarrow 0$ uniformly on $x \in A^c \cap B(\zeta)$, because the first term uniformly converges to 0 by (\ref{unif_1}) and the last one also converges to 0 uniformly by Proposition 4 and the boundedness of $f'''$. We have that $|x\cdot f''\big(\lambda_{f}(x)\big)| > \zeta$ owing to $x \in B(\zeta)$, and from (\ref{unif_2}), there exists a constant $\delta > 0$ such that $|\epsilon| < \delta$ $\Rightarrow$ $|x \cdot f''_{\epsilon}\big(\lambda_{f_{\epsilon}}(x)\big)|\ge \frac{\zeta}{2}$. Since $f_{\epsilon}(\lambda)=f(\epsilon, \lambda)$ has continuous second partial derivatives, it is able to change the order of partial derivatives by Schwarz's theorem, as
	\begin{eqnarray}
	\frac{\partial}{\partial \epsilon} \Big|_{\epsilon=\epsilon_0} f'_{\epsilon}\big(\lambda(\epsilon_0)\big) = \frac{\partial}{\partial \epsilon} \Big|_{\epsilon=\epsilon_0}  \frac{\partial}{\partial \lambda} \Big|_{\lambda=\lambda(\epsilon_0)} f_{\epsilon}(\lambda) 
	= \frac{\partial}{\partial \lambda} \Big|_{\lambda=\lambda(\epsilon_0)} \frac{\partial}{\partial \epsilon} \Big|_{\epsilon=\epsilon_0} f_{\epsilon}(\lambda) 
	= g'\big(\epsilon_0, \, \lambda(\epsilon_0)\big), \nonumber 
	\end{eqnarray}for all $|\epsilon_{0}| < \delta$. Therefore, if $|\epsilon_{0}| < \delta$ by applying implicit function theorem to $F$ again, then we obtain that $\lambda(\epsilon)$ is differentiable at $\epsilon=\epsilon_0$ and 
	\begin{eqnarray}
	|\lambda'(\epsilon_0)| = \Big| \frac{-\partial F(\epsilon,\, \lambda)/ \partial \epsilon} {\partial F(\epsilon,\, \lambda) / \partial \lambda} \Big|_{\big(\epsilon_{0},\, \lambda(\epsilon_0)\big)} \Big| &=& \frac{\Big|x\cdot \frac{\partial}{\partial \epsilon} \big|_{\epsilon=\epsilon_0} f'_{\epsilon}\big(\lambda(\epsilon_0)\big)\Big|}{\Big|x \cdot f''_{\epsilon_0}\big(\lambda(\epsilon_0)\big)\Big|}
	\le \frac{\norm{g'}}{\zeta/2} \le \frac{2}{\zeta}, \nonumber
	\end{eqnarray}
	which completes the proof.
\end{proof}

We further consider principal curves on hypersphere $S^d$ for $d\ge 2$. For a smooth curve $f: [0,\, 1] \to S^d$, suppose that $f$ is parameterized by a constant speed with respect to $\lambda$. Lemma 1 and all of the Propositions are still valid on $S^d$. Moreover, Lemmas 4 and 5 can be extended onto $S^d$ as follows.  
\begin{lemma}
Define $C(\zeta)=\left\{x\in S^d \ | \ |f''(\lambda_{f}(x))\cdot x| > \zeta \right\}$ for $d\ge 2$. Then,
\[
S^d\setminus C(0)=\{x\in S^d \ | \ |f''(\lambda_{f}(x)) \cdot x | =0 \}
\]
has spherical ($d$-dimensional Hausdorff) measure zero. Hence, the spherical measure of $S^{d} \setminus C(\zeta)$ goes to 0 as $\zeta \rightarrow 0$.
\end{lemma}
\begin{lemma}
Let $A$ be a set of ambiguity points of smooth curve $f$ on $S^d$ for $d\ge 2$. Suppose that $x \in A^c\cap C(\zeta)$ for a small $\zeta > 0$, and $\lambda_{f}(x)\in (0,\, 1)$. Then $\lambda(\epsilon):=\lambda_{f_\epsilon}(x)$ is a smooth function for $\epsilon$ on an open interval containing zero.
\end{lemma} 
Proofs of Lemmas 6 and 7 are similar to those of Lemmas 4 and 5, respectively. Thus, we omit the proofs.

\vskip 5mm
\noindent{\bf Proof of Theorem 1} 
\begin{proof}
First of all, we prove the theorem on $S^2$. If $f=h$, then nothing to prove. Thus, we assume that the curves $f$ and $f+g(\small = h)$ are not identical and further both are parameterized by $\lambda \in [0,\, 1]$. To prove the result, we need to show that the conditional expectation is zero after exchanging the order of the derivative and expectation.
	
	First, for order exchange, it is necessary to show that the following random variable  
	\begin{eqnarray}
	\label{eq_z}
	Z_\epsilon(X) &=& \frac{\cos\big(d_{Geo}(X,\, f+\epsilon g)\big)-\cos\big(d_{Geo}(X,\, f)\big)}{\epsilon} \nonumber\\ 
	&=& \frac{\cos\big(d_{Geo}\big(X,\, (f+\epsilon g)(\lambda_{f+\epsilon g}(X))\big)\big) - \cos\big(d_{Geo}\big(X,\, f(\lambda_f(X))\big)\big)}{\epsilon} 
	\end{eqnarray}is uniformly bounded for any sufficiently small $|\epsilon| > 0$. Then we apply bounded convergence theorem. Since the projection index of $X$ represents the closest point in the curve, it follows that 
	\begin{eqnarray}
	\label{eql4}
	Z_\epsilon(X) &\le& \frac{\cos\big(d_{Geo}\big(X,\, (f+\epsilon g)(\lambda_{f+\epsilon g}(X))\big)\big) -\cos\big(d_{Geo}\big(X,\, f(\lambda_{f+\epsilon g}(X))\big)\big)}{\epsilon}. 
	\end{eqnarray}For simplicity, let $f_g (\lambda_\epsilon):=(f+g)(\lambda_{f+\epsilon g}(X))$, $f_\epsilon (\lambda_\epsilon):=(f+\epsilon g)(\lambda_{f+\epsilon g}(X))$ and $f(\lambda_\epsilon):=f(\lambda_{f+\epsilon g}(X))$. By applying Lemma 3 to $\cos\big(d_{Geo}(X,\, f_\epsilon(\lambda_\epsilon))\big)$, the inequality of (\ref{eql4}) becomes 
	\begin{eqnarray}
	\label{ze_ub}
	Z_\epsilon(X) &\le& \frac{\cos\big(d_{Geo}\big(X,\, f_\epsilon(\lambda_\epsilon)\big)\big)-\cos\big(d_{Geo}\big(X,\, f(\lambda_\epsilon)\big)\big)}{\epsilon} \nonumber\\
	&=& \frac{\cos(d_{Geo}(X,\, f(\lambda_\epsilon))) (\cos(d_{Geo}(f_\epsilon (\lambda_\epsilon),\, f(\lambda_\epsilon)))-1) + (f(\lambda_\epsilon)\times f_\epsilon (\lambda_\epsilon))\cdot (f(\lambda_\epsilon)\times X)}{\epsilon}. \nonumber\\
	&=& \frac{\cos(d_{Geo}(X,\, f(\lambda_\epsilon))) (\cos (\epsilon d_{Geo}(f_g(\lambda_\epsilon),\, f(\lambda_\epsilon)))-1) + A_\epsilon (f(\lambda_\epsilon)\times f_g (\lambda_\epsilon))\cdot (f(\lambda_\epsilon)\times X)}{\epsilon}, \nonumber \\
	\end{eqnarray}where 
	\begin{align*}
	A_\epsilon = |f(\lambda_\epsilon)\times f_\epsilon (\lambda_\epsilon)|/|f(\lambda_\epsilon)\times f_g (\lambda_\epsilon)| = \sin(\epsilon d_{Geo}(f(\lambda_\epsilon), f_g(\lambda_\epsilon)))/|f(\lambda_\epsilon)\times f_g (\lambda_\epsilon)|.
	\end{align*}The last equality is done by Definition 1. To get the upper bound of $Z_\epsilon(X)$, we further use the following fact, $|\frac{\sin \epsilon C}{\epsilon}|\le |C|$ and $|\frac{1-\cos \epsilon C}{\epsilon}|\le \frac{|\epsilon| C^2}{2}$ for $C\in\mathbb{R}$ and $\epsilon\in\mathbb{R}$. Then, we have 
	\begin{eqnarray*}
	Z_\epsilon(X) &\le& \big|\cos\big(d_{Geo}\big(X,\, f(\lambda_\epsilon)\big)\big)\big| \frac{|\epsilon| B^2}{2} + \frac{B}{|f(\lambda_\epsilon)\times f_g (\lambda_\epsilon)|}\big|(f(\lambda_\epsilon)\times f_g (\lambda_\epsilon))\cdot (f(\lambda_\epsilon)\times X)\big|, 
	\end{eqnarray*}where $B = d_{Geo}\big(f(\lambda_\epsilon),\, f_g(\lambda_\epsilon)\big) \le \norm{g} < \pi$. Note that any smallest geodesic distance on a unit sphere is smaller than $\pi$. In addition, we can assume that $\epsilon$ is less than $1/\pi$ because we are only interested in $\epsilon$ near 0. Thus, we obtain the upper bound of $Z_\epsilon(X)$ in (\ref{ze_ub})   
	\[
	Z_\epsilon(X)\le \  \frac{\pi}{2} + \pi = \frac{3\pi}{2}. 
	\]
	A lower bound of $Z_\epsilon(X)$ can be similarly obtained. Let $f_g(\lambda):=(f+g)(\lambda_{f}(X))$, $f_\epsilon (\lambda):=(f+\epsilon g)(\lambda_{f}(X))$ and $f(\lambda):=f(\lambda_{f}(X))$. By following the same path, we have    
	\begin{eqnarray}
	\label{ze_lb}
	Z_\epsilon(X) &\ge& \frac{\cos\big(d_{Geo}\big(X,\, (f+\epsilon g)(\lambda_{f}(X))\big)\big) - \cos\big(d_{Geo}\big(X,\, f(\lambda_{f}(X))\big)\big)}{\epsilon} \nonumber\\
	&=& \frac{\cos\big(d_{Geo}\big(X,\, f(\lambda)\big)\big) \big(\cos\big(d_{Geo}\big(f_\epsilon (\lambda),\, f(\lambda)\big)\big)-1\big) + (f(\lambda)\times f_\epsilon (\lambda))\cdot (f(\lambda)\times X)}{\epsilon} \nonumber\\
	&=& \frac{\cos\big(d_{Geo}\big(X,\, f(\lambda)\big)\big) \big(\cos\big(\epsilon d_{Geo}\big(f_g(\lambda),\, f(\lambda)\big)\big)-1 \big) + B_\epsilon (f(\lambda)\times f_g (\lambda))\cdot (f(\lambda)\times X)}{\epsilon}, \nonumber \\
	\end{eqnarray}where 
	\begin{align*}
	B_\epsilon = |f(\lambda)\times f_\epsilon (\lambda)|/|f(\lambda)\times f_g (\lambda)| = \sin(\epsilon d_{Geo}(f(\lambda),\, f_g(\lambda)))/|f(\lambda)\times f_g(\lambda)|.
	\end{align*} 
	By the same way, it can be shown that 
	\[
	Z_\epsilon(X)\ge -\frac{3\pi}{2}. 
	\]
	Hence, we show that  
	\[
	|Z_{\epsilon}(X)|\le \frac{3\pi}{2},    
	\]
	which is bounded for any $0\neq |\epsilon|\le 1/\pi$. Then, by the bounded convergence theorem, it follows that  
	\begin{align*}
	\frac{\partial\mathbb{E}_X\cos(d_{Geo}\big(X,\, f+\epsilon g)\big)}{\partial\epsilon} \Big|_{\epsilon = 0} = \mathbb{E}_X\frac{\partial\cos(d_{Geo}\big(X,\, f+\epsilon g)\big)}{\partial\epsilon} \Big|_{\epsilon = 0}.  
	\end{align*}Thus, the proof is completed provided that the following equation holds 
	\begin{align*}
	\mathbb{E}\Big[\frac{\partial\cos\big(d_{Geo}(X,\, f+\epsilon g)\big)}{\partial\epsilon} \Big|_{\epsilon = 0} \ \Big| \ \lambda_f(X) = \lambda\Big] = 0, ~\mbox{for}~\mbox{\textit{a.e.}}~\lambda.
	\end{align*}
	By the definition of derivative,  
	\[
	\frac{\partial\cos\big(d_{Geo}(X,\, f+\epsilon g)\big)}{\partial\epsilon} \Big|_{\epsilon = 0} = \lim\limits_{\epsilon \to 0}Z_\epsilon(X), 
	\]
	and as shown, $Z_\epsilon(X)$ is bounded. Since $f$ and $f+g$ are continuous, by Proposition 3, if $X$ is not an ambiguity point of $f$ and $f+g$, then  
	\[
	\lim\limits_{\epsilon \to 0}f_g(\lambda_\epsilon)= f_g(\lambda),~ \ \lim\limits_{\epsilon \to 0}f(\lambda_\epsilon) = f(\lambda). 
	\]
	Next, to show the limit of $Z_\epsilon$, we use the fact that $\lim_{\epsilon\to 0}\frac{\sin\epsilon C}{\epsilon}=C$ and $\lim_{\epsilon\to 0}\frac{1-\cos\epsilon C}{\epsilon}=0$ for $C\in \mathbb{R}$ and $\epsilon\in\mathbb{R}$. When $f_g(\lambda)\ne f(\lambda)$, it follows that 
	\begin{align*}
	\lim\limits_{\epsilon \to 0}~\mbox{RHS of (\ref{ze_ub})} 
	= d_{Geo}\big(f(\lambda),\, f_g(\lambda)\big)\frac{f(\lambda)\times f_g (\lambda)}{|f(\lambda)\times f_g (\lambda)|}
	\cdot (f(\lambda)\times X)
	= \mu(\lambda)\cdot (f(\lambda)\times X), 
    \end{align*}where $\mu(\lambda) = d_{Geo}(f(\lambda),\, f_g(\lambda))(f(\lambda)\times f_g (\lambda))/|f(\lambda)\times f_g (\lambda)|$ if $f(\lambda) \ne (f+g)(\lambda)$ and $\mu(\lambda) =0$ otherwise. Similarly, we obtain that 
	\[
	\lim\limits_{\epsilon \to 0}~\mbox{RHS of (\ref{ze_lb})} = \mu(\lambda)\cdot (f(\lambda)\times X). 
	\]
	In summary, if $X$ is not an ambiguity point of $f$ and $f+g$, and $f(\lambda_{f}(X))\ne (f+g)(\lambda_{f}(X))$, then we have 
	\begin{equation}
	\label{limit1}
	\frac{\partial\cos\big(d_{Geo}(X,\, f+\epsilon g)\big)}{\partial\epsilon} \Big|_{\epsilon = 0} = \mu(\lambda_{f}(X))\cdot \big(f(\lambda_{f}(X))\times X\big). 
	\end{equation}
	In the case of $f(\lambda_{f}(X)) = (f+g)(\lambda_{f}(X))$, the equation (\ref{limit1}) also hold because its left and right hand side are 0. From Proposition 4, the limit of (\ref{limit1}) is established for \textit{a.e.} $X$. Note that, since $X$ is a random vector and $\lambda_{f}(X)$ is measurable with respect to $X$ according to the Proposition 2, $\lambda_{f}(X)$ is also a random variable depending on $X$. It implies that conditional expectation on $\lambda_{f}(X)$ is feasible. Hence, the following equality holds
	\begin{align}
	\label{result}
	\mathbb{E}_{X}\Big[\frac{\partial\cos\big(d_{Geo}(X,\, f+\epsilon g)\big)}{\partial\epsilon} \Big|_{\epsilon = 0}\Big] = \mathbb{E}_{X}\Big[\mu(\lambda_{f}(X)) \cdot \big(f(\lambda_{f}(X))\times X \big) \Big].
	\end{align}
	Finally, if $f$ is an extrinsic principal curve, then 
	\[
	\mathbb{E}\big[X \ | \ \lambda_f(X)=\lambda\big]=cf(\lambda)
	\]
	for $\exists c\in \mathbb{R}$. Hence, it follows that 
	\begin{align*}
	\mathbb{E}\big[\mu(\lambda_{f}(X))\cdot \big(f(\lambda_{f}(X))\times X \big) \ | \ \lambda_f(X)=\lambda \big] 
	&= \mathbb{E}\big[\mu(\lambda)\cdot (f(\lambda)\times X) \ | \ \lambda_f(X) = \lambda \big] \\
	&= \mu(\lambda)\cdot (f(\lambda)\times cf(\lambda)) = 0. \hspace{2.5cm} 
	\end{align*}
	Hence, we have 
	\begin{align*}
		\mbox{LHS of (\ref{result})} &= \mathbb{E}_{X}\big[\mu(\lambda)\cdot (f(\lambda)\times X)\big] \hspace{3.4cm} \\
		&= \mathbb{E}_{\lambda}\big[\mathbb{E} \big[\mu(\lambda)\cdot (f(\lambda)\times X) \ | \ \lambda_f(X) = \lambda\big]\big] = 0.
	\end{align*}
	
	To prove the converse, we assume that 
	\begin{align*}
	\mathbb{E}_{\lambda}\big[\mathbb{E}\big(\mu(\lambda)\cdot (f(\lambda)\times X) \ | \ \lambda_f(X) = \lambda\big) \big] \hspace{2.5cm} \\
	= \mathbb{E}_{\lambda} \big[\mu(\lambda) \cdot \mathbb{E}\big[f(\lambda)\times X \ | \ \lambda_f(X) = \lambda\big]\big] = 0, \hspace{0cm}
	\end{align*}
	for all smooth $f+g$ satisfying $\norm{g} \ne \pi$ and $\norm{g'}\le 1$. Since $f+g$ is only concerned with $\mu(\lambda)$, it follows that 
	\begin{align*}
	\mathbb{E}\big[f(\lambda)\times X \ | \ \lambda_f(X) = \lambda\big] = f(\lambda)\times  \mathbb{E}\big[X \ | \ \lambda_f(X) = \lambda\big]
	= 0, ~~\mbox{for}~\mbox{\textit{a.e.}}~\lambda.
	\end{align*}
	Therefore, we have 
	\[
	\mathbb{E}\big[X \ | \ \lambda_f(X)=\lambda\big]=cf(\lambda)
	\]
	for $\exists c\ge 0$, which completes the proof. \\

Next, we consider the hypersphere case $S^d$ for $d\ge 3$. For given smooth curves $f$ and $h(=f+g)$ parametrized by $\lambda \in [0,\, 1]$, if $f=h$, the result is obvious. Thus, we assume that $f$ and $f+g(=h)$ are not identical. Suppose that $X \in A^c \cap B(\zeta)$ for a small $\zeta > 0$ and $\lambda_{f}(X) \in (0,\, 1)~ \mbox{for}~\mbox{\textit{a.e.}}~ X$, where $A$ denotes the set of ambiguity points of $f$. As the proof of the case of $S^2$, we use the bounded convergence theorem to change the order of derivative and expectation. Since $d_{Geo}(x,\, y) = \arccos(x\cdot y)$ for any $x, y\in S^d\subset \mathbb{R}^{d+1}$, we have 
\begin{eqnarray}
\label{eq_zthm}
 Z_\epsilon(X) &:=& \frac{\cos\big(d_{Geo}(X,\, f+\epsilon g)\big) - \cos\big(d_{Geo}(X,\, f)\big)}{\epsilon}\nonumber \\ 
                       &=& \frac{\cos\big(d_{Geo}\big(X,\, (f+\epsilon g)(\lambda_{f+\epsilon g}(X))\big)\big) - \cos\big(d_{Geo}\big(X,\, f(\lambda_f(X))\big)\big)}{\epsilon} \nonumber \\
                       &\le& \frac{\cos\big(d_{Geo}\big(X,\, (f+\epsilon g)(\lambda_{f+\epsilon g}(X))\big) - \cos\big(d_{Geo}\big(X,\, f(\lambda_{f+\epsilon g}(X))\big)}{\epsilon} \nonumber \\
                       &=&  \frac{X\cdot (f+\epsilon g)(\lambda_{f+\epsilon g}(X))- X\cdot f(\lambda_{f+\epsilon g}(X))}{\epsilon} \nonumber \\
                       &=& X\cdot \frac{(f+\epsilon g)(\lambda_{f+\epsilon g}(X)) - f(\lambda_{f+\epsilon g}(X))}{\epsilon}, 
\end{eqnarray}   
where $\cdot$ denotes the standard inner product in $\mathbb{R}^{d+1}$. Hence, we obtain the upper bound of $Z_{\epsilon}(X)$, 
\begin{eqnarray*}
Z_\epsilon(X) &\le& \norm{X} \frac{\norm{(f+\epsilon g)(\lambda_{f+\epsilon g}(X)) - f(\lambda_{f+\epsilon g}(X))}}{\epsilon} \\
&\le& \frac{d_{Geo}\big((f+\epsilon g)(\lambda_{f+\epsilon g}(X)),\, f(\lambda_{f+\epsilon g}(X))\big)}{\epsilon} \\
&\le& \norm{g(\lambda_{f+\epsilon g}(X))}\le \norm{g} \\
&\le& \pi,
\end{eqnarray*}
where $\norm{\cdot}$ denotes the standard norm in $\mathbb{R}^{d+1}$. Similarly, it follows that
\begin{eqnarray*}
Z_\epsilon(X) &\ge&  \frac{\cos\big(d_{Geo}\big(X,\, (f+\epsilon g)(\lambda_{f}(X))\big)\big) - \cos\big(d_{Geo}\big(X,\,f(\lambda_{f}(X))\big)\big)}{\epsilon} \\
                       &=& \frac{X\cdot (f+\epsilon g)(\lambda_f(X))- X\cdot f(\lambda_f(X))}{\epsilon}  \\
                       &=& X\cdot \frac{(f+\epsilon g)(\lambda_f(X)) - f(\lambda_f(X))}{\epsilon} \\
                       &\ge& -\norm{X} \frac{\norm{(f+\epsilon g)(\lambda_f(X)) - f(\lambda_f(X))}}{\epsilon} \\
                       &\ge& -\frac{d_{Geo}\big((f+\epsilon g)(\lambda_f(X)),\, f(\lambda_f(X))\big)}{\epsilon} \\ 
                       &\ge& -\norm{g(\lambda_f(X))}\ge - \norm{g} \\
                       &\ge& -\pi.
\end{eqnarray*}
It means that $Z_{\epsilon}(X)$ is uniformly bounded for $0 \neq |\epsilon|\le 1$. Next, to find the limit of $Z_{\epsilon}(X)$, we have  
\begin{eqnarray*}
 Z_\epsilon(X)  &=& \frac{\cos\big(d_{Geo}\big(X,\, (f+\epsilon g)(\lambda_{f+\epsilon g}(X))\big)\big) - \cos\big(d_{Geo}\big(X,\, f(\lambda_f(X))\big)\big)}{\epsilon}\\
                &=& \frac{X\cdot (f+\epsilon g)(\lambda_{f+\epsilon g}(X))- X\cdot f(\lambda_f(X))}{\epsilon} \\
                &=& X\cdot \frac{(f+\epsilon g)(\lambda_{f+\epsilon g}(X)) - f(\lambda_f(X))}{\epsilon}.
\end{eqnarray*}   
According to the Proposition 3,
\begin{eqnarray}
 \lim_{\epsilon \to 0} Z_\epsilon(X) &=& X\cdot \lim_{\epsilon \to 0}\frac{(f+\epsilon g)(\lambda_{f+\epsilon g}(X)) - f(\lambda_f(X))}{\epsilon} \nonumber \\
                                     &=&: X\cdot \phi(\lambda_{f}(X)). \nonumber
\end{eqnarray}   
For each $X\in A^c\cap B(\zeta)$, define a curve $C:I \to S^d$ by $\epsilon \mapsto C(\epsilon)= (f+\epsilon g)(\lambda_{f+\epsilon g}(X)) \in S^d \subset \mathbb{R}^{d+1}$, where $I$ is an open interval containing zero and $C(0)=f(\lambda_f(X))$. For convenience, let $(f+\epsilon g)(\lambda)=f(\epsilon,\, \lambda)$,\, $\lambda_f(X) = \lambda(0)$ and $\lambda_{f+\epsilon g}(X) = \lambda(\epsilon)$. According to Lemma 7, $\lambda(\epsilon)$ is a smooth function on an interval $I$ containing zero. As $f(\cdot, \cdot)$ is smooth on $[-1,\, 1] \times [0,\, 1]$ by the Proposition 1 and $\lambda(\epsilon)$ is smooth on $\epsilon \in I$, $C(\epsilon)=f\big(\epsilon, \lambda(\epsilon)\big)$ is also smooth on $\epsilon \in I$. Thus, $\phi(\lambda)$ is well defined. Hence, by the definition of tangent space via tangent curves, it follows that 
\begin{eqnarray*}
\phi(\lambda)= \lim_{\epsilon \to 0} \frac{C(\epsilon)-C(0)}{\epsilon}=C'(0) \in T_{f(\lambda)}S^{d}, 
\end{eqnarray*}
where $T_{f(\lambda)}S^d$ is the tangent space of $S^d$ at $f(\lambda)$. Note that, by the symmetry of spheres, any tangent vector in $T_{f(\lambda)}S^d$ is orthogonal to the vector $f(\lambda)$, \textit{i.e.,} $\phi(\lambda)\cdot f(\lambda)=0$. Finally, if $f$ is an extrinsic principal curve, then 
\[
\mathbb{E}\big[X \ | \ \lambda_f(X)=\lambda\big]=cf(\lambda)
\]
for $\exists c\in \mathbb{R}$. Hence, it follows, by the bounded convergence theorem, that 
\begin{align*}
&\frac{\partial\mathbb{E}_{X} \big[\cos\big(d_{Geo}(X,\, f+\epsilon g)\big)\big]}{\partial\epsilon} \Big |_{\epsilon=0} \\ 
 &= \lim_{\epsilon \to 0} \frac{\mathbb{E}_{X}\big[\cos\big(d_{Geo}(X,\, f+\epsilon g)\big)\big] - \mathbb{E}_{X}\big[\cos\big(d_{Geo}(X,\, f)\big)\big]}{\epsilon}  \\
 &= \mathbb{E}_{X} \big[\lim_{\epsilon \to 0} \frac{\cos(d_{Geo}(X,\, f+\epsilon g))-\cos(d_{Geo}(X,\, f))}{\epsilon} \big] \\ 
 &=  \mathbb{E}_{\lambda}\big[\mathbb{E}\big[\lim_{\epsilon \to 0}Z_{\epsilon}(X) \ \big| \ \lambda_{f}(X)=\lambda \big] \big] \\
 &=  \mathbb{E}_{\lambda}\big[\mathbb{E}\big[\phi(\lambda) \cdot X \ \big| \ \lambda_{f}(X)=\lambda \big] \big]\\
 &=  \mathbb{E}_{\lambda}\big[\phi(\lambda) \cdot \mathbb{E}\big[X \ \big| \ \lambda_{f}(X)=\lambda \big] \big]\\
 &=  \mathbb{E}_{\lambda}\big[\phi(\lambda) \cdot cf(\lambda) \big]  \\
 &=  0.
\end{align*}
To prove the converse, we assume that $f$ satisfies 
\begin{eqnarray*}
0 & = & \frac{\partial\mathbb{E}_{X}\big[\cos\big(d_{Geo}(X,\, f+\epsilon g)\big)\big]}{\partial\epsilon} \Big |_{\epsilon=0} \\
  & = & \lim_{\epsilon \to 0} \frac{\mathbb{E}_{X}\big[\cos\big(d_{Geo}(X,\, f+\epsilon g)\big)\big] - \mathbb{E}_{X}\big[\cos\big(d_{Geo}(X,\, f)\big)\big]}{\epsilon}  \\ 
  & = & \mathbb{E}_{X} \big[\lim_{\epsilon \to 0} \frac{\cos\big(d_{Geo}(X,\, f+\epsilon g)\big)-\cos\big(d_{Geo}(X,\, f)\big)}{\epsilon} \big]  \\ 
  & = & \mathbb{E}_{\lambda}\big[\mathbb{E}\big[\lim_{\epsilon \to 0}Z_{\epsilon}(X) \ \big| \ \lambda_{f}(X)=\lambda \big] \big]  \\
  & = &  \mathbb{E}_{\lambda}\big[\mathbb{E}\big[\phi(\lambda) \cdot X \ \big| \ \lambda_{f}(X)=\lambda \big] \big] \\
  & = &  \mathbb{E}_{\lambda} \big[\phi(\lambda)\cdot \mathbb{E}\big[X \ \big| \ \lambda_{f}(X)=\lambda \big] \big],
\end{eqnarray*}
for any smooth curve $h:[0,\, 1] \to S^d$. Since $h$ is arbitrary, $\phi$ can become any vector in $T_{f(\lambda)}S^d$. In addition, $h$ is only concerned with $\phi$. We thus obtain, for \textit{a.e.} $\lambda$, the following condition:  
\begin{eqnarray*}
\phi\cdot \mathbb{E}\big[X \ | \ \lambda_{f}(X)=\lambda \big]=0 ~\mbox{for}~\mbox{any}~\phi \in T_{f(\lambda)}S^d.
\end{eqnarray*}
It means that $\mathbb{E}\big[X|\lambda_{f}(X)=\lambda \big]$ is orthogonal to $T_{f(\lambda)}S^d$. Therefore, it follows that
\begin{eqnarray*}
\mathbb{E}\big[X \ | \ \lambda_{f}(X)=\lambda \big]=cf(\lambda)
\end{eqnarray*}
for $\exists c\ge 0$, which completes the proof.
\end{proof}

\vskip 5mm
\noindent{\bf Proof of Theorem 2}
\begin{proof}
	In the case of $f=h$, the result is obvious. We thus assume that $f$ and $f+g(\small = h)$ are not identical. Further, suppose that $X \in A^c \cap B(\zeta)$ for a small $\zeta > 0$ and $\lambda_{f}(X) \in (0,\, 1)~ \mbox{for}~\mbox{\textit{a.e.}}~ X$. As the proof of Theorem 1, we use the bounded convergence theorem to change the order of derivative and expectation. For this purpose, we define
	\begin{eqnarray}
	Z_\epsilon(X) & = & \frac{d^{2}_{Geo}(X,\, f+\epsilon g)-d^{2}_{Geo}(X,\, f)}{\epsilon} \nonumber\\ 
	& = & \frac{d^{2}_{Geo}(X,\, f_{\epsilon}\big(\lambda_{f_{\epsilon} })\big)-d^{2}_{Geo}\big(X,\, f(\lambda_f)\big)}{\epsilon}, \nonumber
	\end{eqnarray}where $f_{\epsilon} := f + \epsilon g$ for $|\epsilon|\le 1$. Let $\theta(\lambda,\, X)$ be the angle between segments of geodesics from $f(\lambda)$ to $X$ and from $f(\lambda)$ to $(f+g)(\lambda)$. Then, from Lemma 3, it follows that 
	\begin{align*}
	F(\epsilon):=&\cos\big(d_{Geo}\big(X,\,  f_{\epsilon}(\lambda_{f_\epsilon})\big)\big) \\
	=&\cos\big(d_{Geo}\big(X,\, f(\lambda_{f_\epsilon})\big)\big)\cdot \cos\big(\epsilon \norm{g(\lambda_{f_\epsilon})}\big)  \\ 
	&+ \sin\big(d_{Geo}\big(X,\, f(\lambda_{f_\epsilon})\big)\big) \cdot \sin\big(\epsilon \norm{g(\lambda_{f_\epsilon})}\big)\cdot \cos\big(\theta(\lambda_{f_\epsilon},X)\big), 
	\end{align*}
	where $\norm{g(\lambda)} = d_{Geo}\big(f(\lambda),\,  (f+g)(\lambda)\big) < \pi$. 
	
	Firstly, we verify that $Z_{\epsilon}(X)$ is uniformly bounded for a small $|\epsilon| > 0$. By Lemma 5, there are constants $C>0$ and $\eta > 0$ such that if $0 < |\epsilon_0| < \eta$, then $\lambda(\epsilon)$ is differentiable at $\epsilon=\epsilon_0$ and $\big|\frac{\partial \lambda(\epsilon)}{\partial \epsilon} \big|_{\epsilon=\epsilon_0} \big| < C$, where $\lambda(\epsilon)=\lambda_{f_{\epsilon}}(X)$. For convenience, let $\lambda_{f_{\epsilon}}(X) = \lambda_{\epsilon}$ and $\lambda_f(X) = \lambda_0$. If $ 0 < |\epsilon_{0}| < \eta$, then by the triangle inequality on sphere and mean value theorem, it follows that 
	\begin{align*}
		|Z_{\epsilon_{0}}(X)| 
		&= \Bigg|\frac{d_{Geo}\big(X,\, f_{\epsilon_0}(\lambda_{f_{\epsilon_0}                      })\big)-d_{Geo}\big(X,\, f(\lambda_f)\big)}{\epsilon_0} \Bigg|
		\cdot \Big(d_{Geo}\big(X,\, f_{\epsilon_0}(\lambda_{f_{\epsilon_0}})\big) + d_{Geo}\big(X,\, f(\lambda_f) \big)\Big) \\
		&\le 2\pi \cdot  \frac{d_{Geo}\big(f(\lambda_0),\, f_{\epsilon_0}(\lambda_{\epsilon_0})\big)}{\epsilon_0}  \\
		&\le 2\pi \cdot \bigg[\frac{d_{Geo}\big(f(\lambda_0),\, f(\lambda_{\epsilon_0})\big)}{\epsilon}
		+ \frac{d_{Geo}\big(f(\lambda_{\epsilon_0}),\, f_{\epsilon_0}(\lambda_{\epsilon_0})\big)}{\epsilon_0}\bigg] \\
	    &<  2\pi \cdot \big(s\cdot \frac{|\lambda_{0}-\lambda_{\epsilon_0}|}{\epsilon_0} + \norm{g(\lambda_{\epsilon_0})}\big) \\
		&\le 2\pi \cdot (s \cdot C+ \pi), 
	\end{align*}where $s=|f'(\lambda)|$ for all $\lambda$. Therefore, $Z_{\epsilon}(X)$ is uniformly bounded on $X \in A^c\cap B(\zeta)$ for $0 < |\epsilon| < \eta $. 
	
	Secondly, we aim to find the limit of $Z_{\epsilon}(X)$. For this purpose, we define a map $u:(-1,\, 1]\to (1,\infty)$ by $u(x)=\arccos(x) \cdot \frac{1}{\sqrt{1-x^2}}$ if $x \in (-1,\,  1)$, and $u(1)=1$. By simple calculations, $u$ is a monotone decreasing continuous function on $(-1,\, 1]$. Note that $F(\epsilon)$ is differentiable for $|\epsilon| < \eta$. By the mean value theorem to find the limit of $Z_{\epsilon}(X)$, we have 
	\begin{eqnarray}
	\label{r_eq}
	Z_{\epsilon_{0}}(X) &=& \frac{d^{2}_{Geo}\big(X,\, f_{\epsilon_0}\big(\lambda_{f_{\epsilon_0}                      })\big)-d^{2}_{Geo}\big(X,\, f(\lambda_f)\big)}{\epsilon_0} \nonumber \\
	&=&   \frac{\arccos^{2}\big(F(\epsilon_0)\big)-\arccos^2\big(F(0)\big)}{\epsilon_0} \nonumber \\
	&=& -2 \arccos\big(F(\epsilon_{1})\big) \cdot \frac{1}{\sqrt{1-F^{2}(\epsilon_{1})}} \cdot      \frac{dF(\epsilon)}{d\epsilon}\Big|_{\epsilon = \epsilon_{1}}  
	\end{eqnarray}for $0 < |\epsilon_1| < |\epsilon_0| < \eta$. When $F(\epsilon_{1})=1$, the last equality is considered as a limit that is well-defined, because $\lim_{x \to 1}u(x) = 1$ and $u(x)$ is smoothly extended on an open interval containing 1 such that $u(x)$ is differentiable at $x=1$. By applying chain rule to the derivative of $F$, we obtain  
	\begin{align*}
	\lim_{\epsilon_{0} \to 0} \frac{\partial F(\epsilon)}{\partial \epsilon} \Big |_{\epsilon=\epsilon_0} 
	= \lim_{\epsilon_{0} \to 0} \Big[ &\sin\big(d_{Geo}\big(X,f(\lambda_{f_{\epsilon_0}})\big)\big) \\ 
	&\cdot \cos\big(\theta(\lambda_{f_{\epsilon_0}},X)\big) \cdot \Big(\lVert g(\lambda_{f_{\epsilon_0}}) \rVert + \epsilon_{0}\cdot \frac{\partial \norm{g(\lambda_{f_{\epsilon}})}}{\partial \epsilon}\Big |_{\epsilon=\epsilon_0}\Big)\Big] \\
	- \lim_{\epsilon_{0}\to 0}{} &\Big[\sin\big(d_{Geo}\big(X,\, f(\lambda_{f_{\epsilon_0}})\big)\big) \cdot \frac{\partial d_{Geo}\big(X,\, f(\lambda_{f_\epsilon})\big)}{\partial \epsilon} \Big |_{\epsilon=\epsilon_0}\Big].
	\end{align*}
	In addition,
	\[
	\frac{\partial \norm{g(\lambda_{f_\epsilon})}}{\partial \epsilon} \Big |_{\epsilon=\epsilon_0} = \frac{\partial \norm{g(\lambda)}}{\partial \lambda} \Big |_{\lambda=\lambda_{f_{\epsilon_0}}} \frac{\partial \lambda(\epsilon)}{\partial \epsilon} \Big |_{\epsilon=\epsilon_0}, 
	\]
	which exists and does not diverge as $\epsilon_{0}$ goes to 0, since $\norm{g(\lambda)}=d_{Geo}\big(f(\lambda),\,  (f+g)(\lambda)\big)$ is continuously differentiable for $\lambda$ and $\frac{\partial \lambda(\epsilon)}{\partial \epsilon} \big|_{\epsilon=0}$ is bounded by Lemma 5. Moreover, 
	\begin{align*}
	\lim_{\epsilon_{0}\to 0} \frac{\partial d_{Geo}\big(X,\, f(\lambda_{f_\epsilon})\big)}{\partial \epsilon} \Big |_{\epsilon=\epsilon_0}
	&= \lim_{\epsilon_{0} \to 0} \frac{\partial d_{Geo}\big(X,\,  f(\lambda) \big)}{\partial \lambda}\Big |_{\lambda=\lambda_{f_{\epsilon_0}}}
	\cdot \frac{\partial \lambda(\epsilon)}{\partial \epsilon} \Big |_{\epsilon=\epsilon_0} \\
	&= \frac{\partial d_{Geo}\big(X,\, f(\lambda)\big)}{\partial \lambda} \Big |_{\lambda=\lambda_{f}} \frac{\partial \lambda(\epsilon)}{\partial \epsilon} \Big |_{\epsilon=0}=0,
	\end{align*}where $\lambda(\epsilon)= \lambda_{f_\epsilon}$. The last equality is done by the definition of $\lambda_f$. Therefore, we have
	\begin{eqnarray}
	\label{thm2_limit_F}
	\lim_{\epsilon \to 0} \frac{\partial F(\epsilon)}{\partial \epsilon} = \norm{g(\lambda_f)}\cdot \cos\big(\theta(\lambda_{f},\, X)\big)
	\cdot \sin\big(d_{Geo}\big(X,\, f(\lambda_f)\big)\big). 
	\end{eqnarray}
	Thirdly, it follows from (\ref{r_eq}) and (\ref{thm2_limit_F}) that
	\begin{align}
	\label{thm2_limit_1}
	\lim_{\epsilon_{0}\to0}{Z_{\epsilon_{0}}(X)}
	 =& \lim_{\epsilon_{1}\to0}\Big[-2\arccos{F(\epsilon_{1})}\cdot\frac{1}{\sqrt{1-F^{2}(\epsilon_1)}}
	 \cdot\frac{dF(\epsilon)}{d\epsilon}\Big|_{\epsilon=\epsilon_{1}}\Big] \nonumber \\ 
	 =& -2u\Big(\cos\big(d_{Geo}\big(X,\, f(\lambda_f)\big)\big)\Big)\cdot \norm{g(\lambda_f)} 
	  \cdot \cos\big(\theta(\lambda_{f},\, X)\big)\cdot \sin\big(d_{Geo}\big(X,\, f(\lambda_f)\big)\big) \\ 
	 =& -2d_{Geo}\big(X,\, f(\lambda_f)\big)\cdot \frac{1}{\sin\big(d_{Geo}\big(X,\, f(\lambda_f)\big)\big)} \nonumber \\
	 & \cdot \norm{g(\lambda_f)} \cdot \cos\big(\theta(\lambda_{f},\, X)\big)\cdot \sin \big(d_{Geo}\big(X,\, f(\lambda_f)\big)\big) \\
	 \label{thm2_limit_2}
	 =& -2d_{Geo}\big(X,\, f(\lambda_f)\big)\cdot \norm{g(\lambda_f)}\cdot \cos\big(\theta(\lambda_f,\, X)\big),
	 \end{align}In the case of $d_{Geo}\big(X,\, f(\lambda_{f})\big) = 0$, the same result follows since both (\ref{thm2_limit_1}) and (\ref{thm2_limit_2}) are zero. Thus, by Proposition 5, the equation (\ref{thm2_limit_2}) is established for \textit{a.e.} $X \in B(\zeta)$. Next, we notice that, for a smooth curve $f$, it can be shown that $M_{\lambda}:=\{x \in S^2 \ | \ \lambda_{f}(x)=\lambda \}$ is a subset of the great circle perpendicular to $f$ at $f(\lambda)$ by Lemma 2. Let $S_{\lambda}$ be the great circle perpendicular to $f$ at $f(\lambda)$. That is, $M_{\lambda} \subset S_{\lambda} \cong S^1$. Moreover, a connected proper subset of $S_{\lambda}$ is isometric to a line with the same length in $\mathbb{R}$, which makes the intrinsic mean on $M_\lambda$ feasible. Note that if the length is less than $\pi/2$, the intrinsic mean is unique. Thus, $f$ is an intrinsic principal curve of $X$, by the definition of $\theta(\lambda_{f},X)$ and $\cos(\pi-\theta)= -\cos(\theta)$, if and only if
	\begin{align*}
	\mathbb{E}\big[d_{Geo}\big(X,\, f(\lambda_{f})\big) \cdot \cos\big(\theta(\lambda_{f},\, X)\big) \ \big| \ \lambda_{f}(X) = \lambda \big] = 0,
	~~ \mbox{for} ~ \mbox{\textit{a.e.}}~ \lambda. 
	\end{align*} 
	Finally, it follows, from (\ref{thm2_limit_2}) and by the bounded convergence theorem, that
	\begin{align*}
		& \frac{\partial\mathbb{E}_{X}\big[d^2_{Geo}(X,\, f+\epsilon g)\big]}{\partial\epsilon} \Big|_{\epsilon=0} \\  
		&= \lim_{\epsilon\to 0} \big[\frac{\mathbb{E}_{X}\big[d^2_{Geo}(X,\, f+\epsilon g)\big] - \mathbb{E}_{X}\big[d^2_{Geo}(X,\, f)\big]}{\epsilon} \big] \\ 
		 &= \mathbb{E}_{X} \big[\lim_{\epsilon\to 0} \frac{d^2_{Geo}(X,\, f+\epsilon g) - d^2_{Geo}(X,\, f)}{\epsilon} \big] \\ 
		 &= \mathbb{E}_{\lambda}\big[\mathbb{E}\big[\lim_{\epsilon \to 0}Z_{\epsilon}(X) \ \big| \ \lambda_{f}(X)=\lambda\big]\big] \\
		 &= -2\mathbb{E}_{\lambda} \big[\mathbb{E}\big[ d_{Geo}\big(X,\, f\big(\lambda_{f}(X)\big)\big) \cdot \norm{g\big(\lambda_{f}(X)\big)} 
		 \cdot \cos \big(\theta(\lambda_{f},\, X)\big) \ \big| \ \lambda_{f}(X)=\lambda \big] \big] \\
		 &= -2 \mathbb{E}_{\lambda}\big[\norm{g(\lambda)} \cdot \mathbb{E} \big[d_{Geo}\big(X,\, f\big(\lambda_{f}(X)\big)\big) 
		 \cdot \cos\big(\theta(\lambda_{f},\, X) \big) \ \big| \ \lambda_{f}(X)=\lambda \big) \big] \big] \\
		 &= 0.
	\end{align*}Conversely, we assume that  
	\begin{align*}
	\mathbb{E}_{\lambda}\big[\norm{g(\lambda)}\cdot \mathbb{E}\big[d_{Geo}\big(X,\, f\big(\lambda_{f}\big)\big) \cdot \cos\big(\theta(\lambda_{f},X)\big) 
	\big| \ \lambda_{f}(X)=\lambda \big] \big]= 0, 
	\end{align*}
	for all $f+g(=h)$ such that $\norm{g} \ne \pi$ and $\norm{g'}\le 1$. It follows that  
	\begin{align*}
	\mathbb{E}\big[d_{Geo}\big(X,\, f(\lambda_{f})\big)\cdot \cos\big(\theta(\lambda_{f},\, X)\big) \ \big| \ \lambda_{f}(X) = \lambda \big]=0, 
	~~ \mbox{for} ~ \mbox{\textit{a.e.}} ~ \lambda,
	\end{align*}
	which is equivalent to that $f$ is an intrinsic principal curve of $X$.
\end{proof}

\section{Concluding Remarks} 
In this paper, new principal curves are proposed for data on spheres. The extrinsic and intrinsic perspectives are considered, and the stationarity of the principal curves is investigated, supporting that the proposed methods are a direct generalization of the principal curves by \cite{Hastie} to spheres.

For the data on $S^d$, both extrinsic and intrinsic approaches yield similar performance. However, it is questionable whether the extrinsic approach of non-isotropic manifolds, like a torus, will still be valid. For some non-isotropic manifolds, the intrinsic approach may yield better performance because of its inherency. Finally, the principal curve algorithm proposed in this study is a top-down approach. It approximates the structure of data with an initial curve and then gradually improves the estimation. However, for complex structures divided into several pieces or containing intersections, the initialization can significantly affect the final estimate. To cope with this limitation, it is worth studying a bottom-up approach. This approach to spheres is left for future research.  

\section*{Acknowledgments} 
This research was supported by the National Research Foundation of Korea (NRF) funded by the Korea government (2018R1D1A1B07042933; 2020R1A4A1018207).


\bibliographystyle{apalike}
\bibliography{spc}

\end{document}